\documentclass[letter,11pt]{article}

\usepackage{amsmath,amssymb, amsthm}
\usepackage{float}
\usepackage{multirow}
\usepackage{graphicx}
\usepackage{amssymb,amsmath}
\usepackage[margin=1in]{geometry}
\usepackage[normalem]{ulem}

\usepackage[ruled,vlined,linesnumbered]{algorithm2e}
\usepackage{bm} 		
\usepackage{color}
\usepackage{dcolumn} 	
\usepackage{hyperref} 	
\usepackage{microtype}	

\usepackage[T1]{fontenc}
\usepackage[utf8]{inputenc}
\usepackage{authblk}

\newcommand{\ket}[1]{\left\vert{#1}\right\rangle}
\newcommand{\AR}[2][c]{$$\begin{array}[#1]{lllllllllllllll}#2\end{array}$$}

\def\ket#1{{|}#1\rangle}

\newtheorem{theorem}{Theorem}
\newtheorem{lemma}{Lemma}
\newtheorem{corollary}{Corollary}
\newtheorem{proposition}{Proposition}
\newtheorem{definition}{Definition}

\begin{document}

\title{Global Quantum Circuit Optimization}
\author[1]{R. Dias da Silva\thanks{The first two authors have contributed equally to this work.}}
\author[2]{E. Pius$^{\ast}$}
\author[2]{E. Kashefi}

\affil[1]{Instituto de F\'isica, Universidade Federal Fluminense, Av. Gal. Milton Tavares de Souza s/n, 
Gragoat\'a, Niter\'oi, RJ, 24210-340, Brazil}
\affil[2]{School of Informatics, University of Edinburgh, 10 Crichton Street, Edinburgh EH8 9AB, UK}

\maketitle

\begin{abstract}

One of the main goals in quantum circuit optimisation is to reduce the number of ancillary qubits and the depth of computation, to obtain robust computation. However, most of known techniques, based on local rewriting rules, for parallelising quantum circuits will require the addition of ancilla qubits, leading to an undesired space-time tradeoff. Recently several novel approaches based on measurement-based quantum computation (MBQC) techniques attempted to resolve this problem. The key element is to explore the global structure of a given circuit, defined via translation into a corresponding MBQC pattern. It is known that the parallel power of MBQC is superior to the quantum circuit model, and hence in these approaches one could apply the MBQC depth optimisation techniques to achieve a lower depth. However, currently, once the obtained parallel pattern is translated back to a quantum circuit, one should either increase the depth or add ancilla qubits. In this paper we characterise those computations where both optimisation could be achieved together. In doing so we present a new connection between two MBQC depth optimisation procedures, known as the \emph{maximally delayed generalised flow} and \emph{signal shifting}. This structural link will allow us to apply an MBQC qubit optimisation procedure known as \emph{compactification} to a large class of pattern including all those obtained from any arbitrary quantum circuit. We also present a more efficient algorithm (compared to the existing one) for finding the maximally delayed generalised flow for graph states with flow.
\end{abstract}

\newpage
\tableofcontents
\newpage

\section{Introduction and background material} \label{sec_background}

We are slowly reaching the classical limit of current technology in decreasing the size of computer chips. Hence to avoid quantum effect and also to reduce heat generation parallel computing has become the dominant paradigm in computer architecture. On the other hand in the quantum domain reducing the depth would be essential for keeping the computation coherent to avoid classical effect. In both scenarios designing parallel circuit remains a challenging task. Hence one attempts to rewrite locally a given circuit to reduce the depth of the computation. Usually such an approach will add ancilla registers to achieve the parallelisation, this being an undesired effect in the quantum setting as addition of ancilla qubits could increase the decoherence, breaching the initial purpose of parallelisaiton.  In this paper we present a general optimisation techniques for quantum circuit that exploits the global structure of a given computation to achieve parallelisation with no ancilla addition. 

Our technique is based on the translation of a given quantum circuit \cite{Deutsch89, NielsenCbook00} into a measurement-based quantum computation (MBQC) \cite{RaussendorfB01,DanosKP07}. These two models utilise remarkably different information processing tools: while the former is based on unitary evolution of an initially non-entangled set of qubits, the latter needs an initial highly-entangled multi-qubit state, where the information processing is driven by measurements only. Naturally since the two aforementioned models use different information processing tools,  hence each model has its own optimisation techniques. In the MBQC model, for instance, most of the optimization techniques are based on the identification of a more efficient correction structure that is directly linked to the geometry of the underlying global entanglement structure. Examples of these techniques are the \textit{signal shifting} \cite{DanosKP07} and the \textit{generalised flow} \cite{BrowneKMP07}, both discussed in following sections. The so-called \textit{standardisation} procedure \cite{DanosKP07} can also reduce the number of computational steps by rearranging the MBQC operations into a normal form. Moreover, \textit{all} Pauli measurements in this model can be performed in the beginning of the computation \cite{RaussendorfBB01}, which is a surprising difference from the quantum circuit model.

On the other hand, most optimization techniques for quantum circuits are based on template identification and substitution. For instance in \cite{EscartinP11}, some circuit identities are used to modify the teleportation and dense coding protocols, with the purpose of giving a more intuitive understanding of those protocols. Similarly in \cite{SedlakP08} and \cite{MaslovDMN08} a set of circuit identities for reducing the number of gates in the circuit for size optimization was given. In contrast to that, in \cite{MooreN01} a useful set of techniques for circuit parallelisation was provided, where the number of computational steps is reduced by using additional resources. However, as noted in \cite{SedlakP08}, all the aforementioned circuit optimization techniques are basically exchanging a sequence of gates for a different one without any consideration on the structure of the complete circuit being optimised. The translation into MBQC would allow us to explore the global structure of a given circuit. 

The first such a scheme by back and forth translation between the two models was presented in \cite{BroadbentK09}. However the backward translation into the circuit required the addition of many ancilla qubits. On the other hand several recent works presented an optimised translation scheme from MBQC into the circuit where all the non-input qubits are removed, referred hereinafter as compact translation \cite{DanosKP07,DuncanP10,daSilvaG12,Duncan12}. The price for doing that happened to be the loss of the optimal depth of the original MBQC. The key result of our paper is a new scheme for the parallelisation where the obtained MBQC pattern could be translated compactly. Our scheme is based on a new theoretical connection between two MBQC depth optimisation procedures, known as the \emph{maximally delayed generalised flow} \cite{MhallaP07} and \emph{signal shifting}, which could lead to other interesting observation about MBQC, beyond the purpose of this paper as we discuss later. Also our result highlights the fundamental role of the MBQC Pauli optimisation in obtaining MBQC parallel structure beyond anything obtainable in the quantum circuit model, described later. We prove how our proposed scheme is more optimal in both depth and space compare to the scheme in  \cite{BroadbentK09}. We conclude with a new algorithm for finding the maximally delayed generalised flow for graphs with flow with $O(n^3)$ steps compared to the exciting algorithm in \cite{MhallaP07}, where $n$ is the number of the nodes in the graph.

\subsection{The MBQC Model}

We review the basic ideas behind the measurement-based quantum computation, with special attention to its description in terms of the formal language known as Measurement Calculus \cite{DanosKP07},  and the flow theorems \cite{DanosK06,BrowneKMP07}.

In 1999 Chuang and Gottesman described how one could apply arbitrary quantum gates using an adaptation of the quantum teleportation model \cite{GottesmanC99}. This approach was further developed by other researchers \cite{HuelgaPV01, HuelgaVCP00, Leung01, Nielsen01}, enabling one in principle to perform arbitrary computations given a few primitives: preparation of maximally entangled systems of fixed, small dimension; multi-qubit measurements on arbitrary set of qubits; and the possibility of adapting the measurement bases depending on earlier measurement outcomes.

These models of computation draw on measurements to implement the dynamics, and as such named collectively the measurement-based model of quantum computation (MBQC), for an overview see the paper by Jozsa \cite{Jozsa05}. An MBQC model using only single qubit measurements was proposed by Raussendorf and Briegel in 2001, which became known as the one-way model \cite{RaussendorfBB01}. The one-way model achieves universality through the preparation of a special type of entangled states, the so-called cluster states \cite{RaussendorfBB03}. These states are created with the CZ gate acting on qubits prepared in the state $\ket{+} = \frac{1}{\sqrt{2}}(\ket{0} + \ket{1})$ arranged in a regular lattice, usually the two-dimensional ones. This can be relaxed to create more general states with the same interaction over general graphs, creating the so-called graph states \cite{HeinDERvdNB06}. Both cluster and graph states can be represented graphically, using vertices denoting the qubits and edges for the two-qubit entangling gate CZ. Therefore, the entangled resource for the one-way model can be fully represented as graphs. Although two-dimensional cluster states can be used as resource for universal quantum computation in the one-way model, arbitrary graph states may, or may not, serve for the same purpose; investigating which kinds of entangled states are useful resources for MBQC is an active area of research \cite{GrossFE08, VanDenNestMDB6, MoraPMNDB10, WeiTAR11}.

A formal language to describe in a compact way the operations needed for the one-way model was proposed in \cite{DanosKP07}. The language could be easily adapted to any other type of measurement-based model hence in the rest of this paper we refer to the general MBQC term instead of the specific one-way model as our scheme could be applicable to any MBQC models. In this framework every MBQC algorithm (usually referred to as an MBQC pattern) involves a sequence of operations such as entangling gates, measurements and feed-forwarding of outcome results to determine further measurement bases.
A \textit{measurement pattern}, or simply a pattern, is defined by a choice of a set of working qubits $V$, a subset of input qubits ($I$), another subset of output qubits ($O$), and a finite 
sequence of commands acting on qubits in $V$. Therefore, we consider patterns associated to the so-called \textit{open graphs}:
\begin{definition}[\textbf{open graph}]
	\label{def_opengraph}
	An open graph is a triplet $(G, I, O)$, where $G = (V, E)$ is a undirected graph, and $I, O \subseteq V$ are respectively called input and output vertices.
\end{definition}
An example of an open graph is shown in Figure \ref{def_gflow}. There are four types of commands, the first is the qubit initialisation command $N_i$ that prepares qubit $i$ in the state $\ket{+}$. The input qubits are already given as a prepared state. The entangling command $E_{ij} \equiv CZ_{ij}$ corresponds to the $CZ$ gate between qubits $i$ and $j$, where
\begin{align*}
\wedge Z \equiv \left(
\begin{array}{cccc}
1 & 0 &0&0\\
0 &1 &0&0\\
0&0&1&0\\
0&0&0&-1
\end{array}
\right).
\end{align*}
The single-qubit measurement command $M_i^{\theta}$ corresponds to a measurement of qubit $i$ in the basis $\ket{\pm_\theta}\equiv  \frac{1}{\sqrt{2}}(\ket{0}\pm e^{i\theta}\ket{1})$, with outcome $s_i=0$ associated with $\ket{+_\theta}$, and outcome $1$ with $\ket{-_\theta}$. The measurement outcomes are usually referred as \emph{signals}. Finally, the corrections may be of two types, either Pauli $X$ or Pauli $Z$, and they may depend on any prior measurement results, denoted by $s = \oplus_{j \in J \subset V} s_j$ ($s_j =0$ or $1$ and the summation is done modulo two). This dependency can be summarised as correction commands: $X_i^{s}$ and $Z_i^s$ denoting a Pauli $X$ and $Z$ corrections on qubit $i$ which must be applied only when the parity of the measurement outcomes on qubits $j \in J \subset V$ equals one (as $Z^0 = X^0=I$).
A characteristic of the MBQC model is that the choice of measurement bases may depend on earlier measurement outcomes. These dependent measurements can be conveniently written as $_t[M_i^{\theta}]^s$, where
\begin{equation}
_t[M_i^{\theta}]^s \equiv M_i^{\theta}X_i^s Z_i^t = M_i^{(-1)^s\theta+t\pi},
\end{equation}
where it is understood that the operations are performed in the order from right to left in the sequence. The left ($t$) and right ($s$) dependencies of the measurement $M_i$ are called its $Z$ and $X$ {\it dependencies}, respectively.

A pattern is runnable, that is, corresponds to a physically sound sequence of operations, if it satisfies the following requirements: (R0) no command
depends on outcomes not yet measured; (R1) no command acts on a qubit already measured or not yet prepared, with the obvious exception of the preparation commands; (R2) a qubit undergoes 
measurement (preparation) iff it is not an output (input) qubit.

As an example, take the pattern consisting of the 
choices $V=\{1,2\}, I=\{1\}, O=\{2\}$ and the sequence of commands:
\begin{equation} \label{jblock}
X_2^{s_1}M_1^{-\theta}E_{12}N_2^{0}.
\end{equation}
This sequence of operations does the following: first it initialises the output qubit 2 in the state $\ket{+}$; then it applies $\wedge Z$ on qubits 1 and 2; followed by a measurement of input qubit 1 onto the basis $\{1/\sqrt{2}(\ket{0}+e^{-i\theta}\ket{1}),1/\sqrt{2}(\ket{0}-e^{-i\theta}\ket{1})\}$. If the result is the latter vector then the one-bit outcome is $s_1=1$ and there is a correction on the second qubit ($X_2^1=X_2$), otherwise no correction is necessary. A simple calculation shows that this pattern implements the unitary $J_{\theta}$ on the state prepared in qubit 1, outputting the result on qubit 2, where
\begin{equation}
J_{\theta}\equiv \frac{1}{\sqrt{2}} \left(
\begin{array}{cc}
1 & e^{i\theta}\\
1 & -e^{i \theta}
\end{array}
\right).
\end{equation}
The simple sequence above is a convenient building block of more complicated computations in the MBQC model. This is because the set of single qubit $J_{\theta}$ ($\forall \theta$) together with CZ on arbitrary pairs of qubits can be shown to be a universal set of gates for quantum computation \cite{NielsenCbook00}.

The following rewrite rules (\cite{DanosKP07}) put the command sequence in the \emph{standard} form, where preparation is done first followed by the entanglement, measurements and corrections: 
\begin{eqnarray} 
E_{ij}X_i^s &\Rightarrow& X_i^sZ_j^sE_{ij} \label{rw1}\\ 
		E_{ij}Z_i^s &\Rightarrow& Z_i^sE_{ij} \label{rw2}\\
_t[M_i^{\theta}]^sX_i^r &\Rightarrow& _t[M_i^{\theta}]^{s+r} \label{rw3}\\
_t[M_i^{\theta}]^sZ_i^r &\Rightarrow& _{r+t}[M_i^{\theta}]^s \label{rw4}
\end{eqnarray}
This procedure is called \textit{standardisation} and can directly change the dependencies structure commands, possibly reducing the computational depth, without breaking the causality ordering \cite{BroadbentK09}. 

\subsection{Determinism in MBQC}

Due to the probabilistic nature of quantum measurement, not every measurement pattern implements a \emph{deterministic} computation -- a completely positive, trace-preserving (cptp) map that sends pure states to pure states. We will refer to the collection of possible measurement outcomes as a \emph{branch} of the computation. In this paper, we consider deterministic patterns which satisfies three conditions: (1) the probability of obtaining each branch is the same, called \emph{strong determinism}; (2) for any measurement angle we have determinism, called \textit{uniform determinism}; and (3) which are deterministic after each single measurement, called \textit{stepwise determinism}. We will call those patterns simply \textit{deterministic patterns}. Here since we are only working with quantum circuits we don't need to be concerned with other stronger notions of determinism defined for MBQC pattern such as in \cite{MhallaMPST11}.

In \cite{DanosK06} conditions over a graph (knows as \emph{flow}) are presented in order to identify a dependency structure for the measurement sequence associated to open graph to obtain determinism. In what follows, we call non-input vertices as $I^C$ (complement of $I$ in the graph) and non-output vertices as $O^C$ (complement of $O$ in the graph).

\begin{definition}[\textbf{Flow} \cite{DanosK06}]
	\label{def_flow}
	We say that an open graph $(G,I,O)$ has \textit{flow} iff there exists a map  $f:O^C \to I^C$ and a strict partial order $\prec_f$ over all vertices in the graph such that for all $i\in O^C$
	\begin{itemize}
	\item (F1) $i \prec_f f(i)$;
	\item (F2) if $j \in N(f(i))$, then $j = i$ or $i \prec_f j$, where $N(v)$ is the neighbourhood of $v$;
	\item (F3) $i \in N(f(i))$;
	\end{itemize}
\end{definition}

Efficient algorithms for finding flow (if it exist) can be found in \cite{Beaudrap06a,MhallaP07}.
The flow function $f$ is a one-to-one function. The proof is trivial, but as this property is extensively used in this work we will present the proof in this paper.
\begin{lemma}
	Let $(f, \prec_f)$ be a flow on an open graph $(G, I, O)$. The function $f$ is an injective function, \emph{i.e.} for every $i \in O^C$, $f(i)$ is unique.
\end{lemma}
\begin{proof}
	Let us assume that for some $i \in O^C$, $f(i)$ is not unique, \emph{i.e.} there exists $j \in O^C$ such that $i \neq j$ but $f(i) = f(j)$.
	Then according to the flow definition:
	\begin{align*}
		j \in N(f(i)) \Rightarrow i \prec_f j, \\
		i \in N(f(j)) \Rightarrow j \prec_f i,
	\end{align*}
	and we arrive to a contradiction because $i \prec_f j$ and $j \prec_f i$ cannot be true at the same time. Hence $f(i)$ has to be unique.
\end{proof}

In the case where $|I|=|O|$, the flow function induces a circuit-like structure in a graph, in the sense that for each input qubit $i \in I$, there exists a number $n$ such that $f^n(i) \in O$ and the vertex sequence
\begin{equation} \label{eq_ordered}
\{ i, f(i), f[f(i)], f[ f[f(i)]]\}, ..., f^n(i)\}, 
\end{equation}
can be translated to a single wire in the circuit model. A simple example can be seen in Figure \ref{fig_flowgflow}. This circuit-like structure is an interesting feature of the flow function, since it allows a very simple translation procedure called star decomposition introduced in \cite{DanosK06}.

Flow provides only a sufficient condition for determinism  but one can generalise the above definition to obtain a condition that is both necessary and sufficient. This generalisation allows correcting sets with more than one element. In those cases, we say that the graph has \emph{generalised flow} (or simply \emph{gflow}). In what follows we define $Odd (K) =  \{k\, ,\,  |N_G(k)\cap K|=1 \mod 2\}$ to be the set of vertices where each element is connected with the set $K$ by an odd number of edges.

\begin{definition} [\textbf{Generalised flow} \cite{BrowneKMP07}]
\label{def_gflow}
	We say $(G,I,O)$ has generalised flow if there exists a map $g: O^C \rightarrow P^{I^C}$ (the set of all subsets of non-input qubits) and a partial order $\prec_g$  over all vertices in the graph such that for all $i\in O^C$,
	\begin{itemize}
		\item (G1)  if $j\in g(i)$ then $i \prec_g j$;
		\item (G2)  if $j\in Odd(g(i))$ then $j = i$ or $i \prec_g j$;
		\item (G3)  $i \in Odd(g(i))$;
	\end{itemize}
\end{definition}

The set $g(i)$ is often referred to as the \emph{correcting set} for qubit $i$.
It is important to note that flow is a special case of gflow, where $g(i)$ contains only one element.
This is a key difference regarding the translation of measurement patterns to quantum circuits. An example of a graph with gflow (but no flow) is shown in Figure \ref{fig_flowgflow}.

\begin{figure}
\center

	\label{fig_flowgflow}
	\includegraphics[scale=0.6]{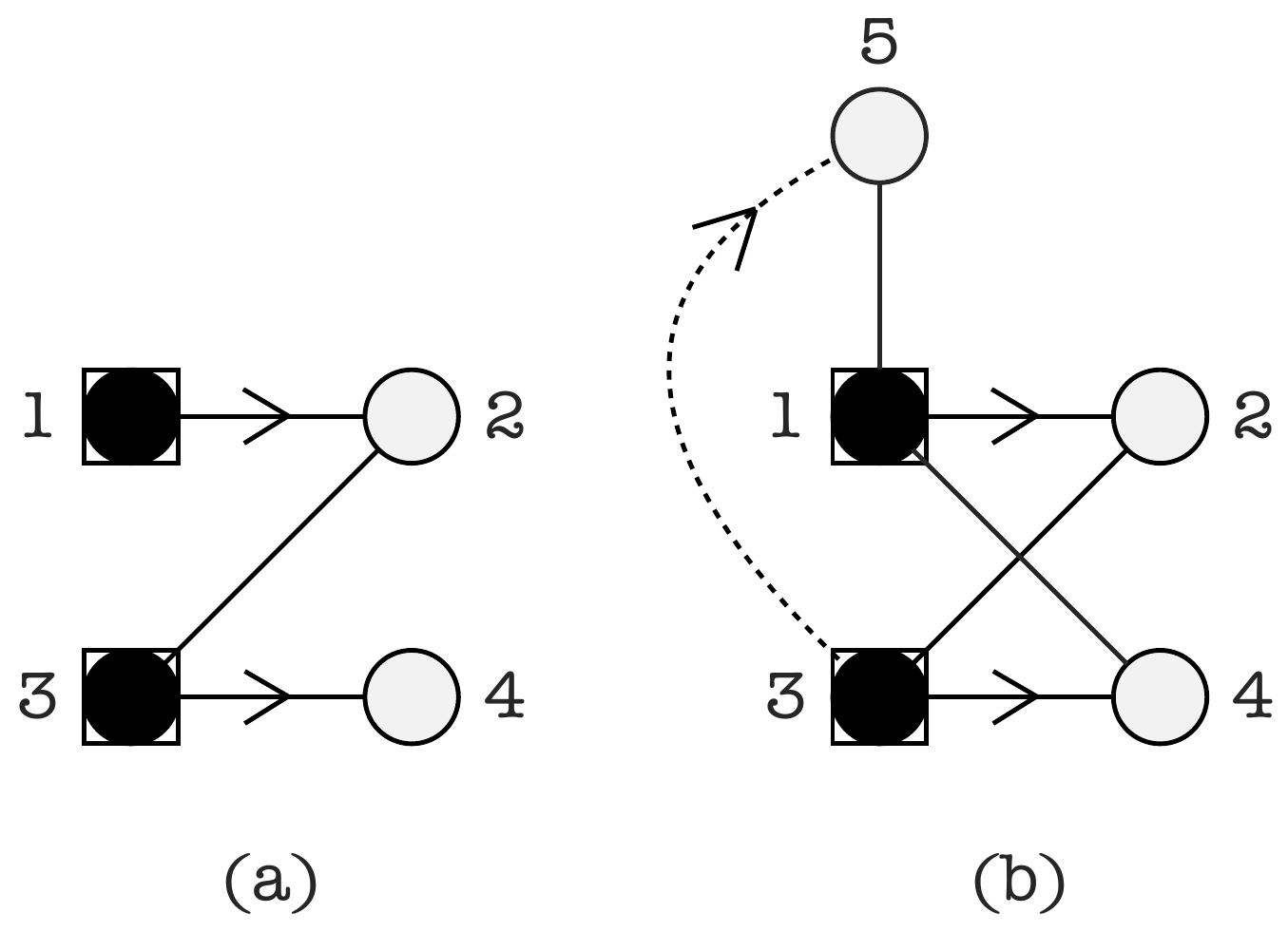}
	\caption{Example of graphs satisfying the conditions of (a) flow and (b) gflow. Each vertex represent a qubit, where the black vertices represent measured qubits and the white ones are unmeasured qubits. The input qubits are represented by boxed vertices and the arrows highlight the dependency structure, since it points to vertices belonging to the correcting set of the measured qubit from where it came from. The dashed line with an arrow represent the case where the vertex from the correcting set is a non-neighbouring vertex (vertices not linked by an edge). }
\end{figure}

The gflow partial oder leads to an arrangement of the vertices into layers (see below), in which all the corresponding measurements can be performed simultaneously. The number of layers corresponds to the number of parallel steps in which a computation could be finished, known as the \emph{depth} of the pattern. 

\begin{definition}[\textbf{Depth of a gflow} \cite{MhallaP07}]
	\label{def_gflow_depth}
	For a given open graph $(G, I, O)$ and  a gflow $(g, \prec_g)$ of $(G, I, O)$, let
	\begin{align*}
		V_k^{\prec_g} =
		\begin{cases}
			\max_{\prec_g}(V(G)) & \text{ if } k = 0\\
			\max_{\prec_g}(V(G) \setminus (\cup_{i < k} V_i^{\prec_g})) & \text{ if } k > 0
		\end{cases}
	\end{align*}
	where $\max(X)_{\prec_g} = \{ u \in X \text{ s.t. } \forall v \in X, \lnot (u \prec_g v) \}$ is the set of maximal elements of $X$ according to $\prec_g$. The \emph{depth} $d^{\prec_g}$ of the gflow is the smallest $d$ such that $V_{d + 1}^{\prec_g} = \emptyset$, $(V_k)_{k = 0 \dots d^{\prec_g}}$ is a partition of $V(G)$ into $d^{\prec_g} + 1$ layers.
\end{definition}

We define the \emph{layering function} of a gflow based on the above distribution of vertices into layers.

\begin{definition}[\textbf{Layering function}] Given a gflow $(g, \prec_g)$ on an open graph $(G, I, O)$ we define its \emph{layering function}  $L_g: V(G) \rightarrow \mathbb{N}$ to be the natural number $k$ such that $i \in V_{k}^{\prec_g}$.
\end{definition}

There is another useful way to understand the depth of a gflow.
A gflow can be represented as a directed graph on top of an open graph as shown in Figure \ref{fig_flowgflow}.
The longest path from inputs to outputs over those directed edges corresponds to the depth of the gflow.
In \cite{MhallaP07} it was shown, that a special type of gflow, called a \emph{maximally delayed gflow}, has minimal depth.

\begin{definition}[\textbf{Maximally delayed gflow} \cite{MhallaP07}]
	\label{def_delayed_gflow}
	For a given open graph $(G, I, O)$ and two given gflows $(g, \prec_g)$ and $(g', \prec_{g'})$ of $(G, I, O)$, $(g, \prec_g)$ is more delayed than $(g', \prec_{g'})$ if $\forall k$, $|\cup_{i = 0 \dots k} V_i^{\prec_g}| \geq |\cup_{i = 0 \dots k} V_i^{\prec_{g'}}|$ and there exists a $k$ such that the inequality is strict. A gflow $(g, \prec_g)$ is \emph{maximally delayed} if there exists no gflow of the same graph that is more delayed.
\end{definition}
We will simply refer to the maximally delayed gflow as the \emph{optimal gflow}. Note that in \cite{MhallaP07} it was proven that the layering of the vertices imposed by an optimal gflow $(g, \prec_g)$ is always unique, however the gflow itself might not be unique.
This is an important property together with the following lemmas that we will exploit later for our main result on linking gflow to other known structures for MBQC. 

\begin{lemma}[Lemma 1 from \cite{MhallaP07}]
	\label{lem_last_layer}
	If $(g, \prec)$ is a maximally delayed gflow of $(G, I, O)$ then $V_0^\prec = O$.
\end{lemma}

\begin{lemma}[Lemma 2 from \cite{MhallaP07}]
	\label{lem_penultimate_layer}
	If $(g, \prec)$ is a maximally delayed gflow of $(G, I, O)$ then $(\tilde{g}, \prec_{\tilde{g}})$ is a maximally delayed gflow of $(G, I, O \cup V_1^\prec)$ where $\tilde{g}$ is the restriction of $g$ to $V(G) \setminus (V_1^\prec \cup V_0^\prec)$ and $\prec_{\tilde{g}} = \prec \setminus V_1^\prec \times V_0^\prec$.
\end{lemma}

\section{Depth optimisation tools for MBQC}

Th parallel power of MBQC is proven to be equivalent to quantum circuit augmented with free unbounded fanout \cite{BKP11}. This motivates to use MBQC as an automated tool for circuit parallelisation as it was first presented in \cite{BroadbentK09}. Another way to obtain parallel MBQC structure is to extract the entanglement graph of the pattern and obtain the optimal gflow of the graph \cite{MhallaP07}. Then one performs the required corrections according to this structure. Our first main result is to show the equivalence between these two seemingly very different technique for the patterns obtained from a quantum circuit, that is those with flow. More precisely we show how the effect of performing signal shifting optimisation (that is the core idea in \cite{BroadbentK09}) result in a maximally delayed gflow. This structural link shed further lights on the complicated structure of maximally delayed gflow and permit us to find a new efficient algorithm for finding it for the large class of patterns obtained form a circuit. 

\subsection{Signal shifted flow} \label{sec_ssf}

We proceed with reviewing the rules for signal shifting defined in \cite{DanosKP07}:
\begin{align}
	_t[M_i^{\alpha}]^s &\Rightarrow S_i^t \; [M_i^{\alpha}]^s \label{eq_ss1} \\
	X_j^s \; S_i^t &\Rightarrow S_i^t \; X_j^{s[(t+s_i)/s_i]} \label{eq_ss2} \\
	Z_j^s \; S_i^t &\Rightarrow S_i^t \; Z_j^{s[(t+s_i)/s_i]} \label{eq_ss3}
\end{align}
where $S_i^t$ is the signal shifting command (adding $t$ to $s_i$) and $s[t/s_i]$ denotes the substitution of $s_i$ with $t$ in $s$. Signal shifting, can be utilised to parallelise MBQC patterns and quantum circuits \cite{BroadbentK09}. The rest of this section is focused on various structural properties of the signal shifting. 

As can be seen from the above rules, signal shifting rewrites the $X$- and $Z$-corrections of a measurement pattern in a well defined manner.
In particular, it will move all the $Z$-corrections to the end of the pattern, thereby introducing new $X$-corrections when Rule \ref{eq_ss2} is applied.
It is proven in \cite{BroadbentK09} that signal shifting will never increase the depth of an MBQC pattern, although it can decrease it.
In the case when the depth decreases, it is the consequence of the removal of the $Z$-corrections on the measured qubits by applying Rule \ref{eq_ss1}.

The Rules \ref{eq_ss1} - \ref{eq_ss3} can be interpreted in the following way.
Signal shifting takes a signal from a $Z$-correction on a measured qubit $i$ (Rule \ref{eq_ss1}) and adds it to the corrections  that depend on the outcome of the measurement of $i$ (Rules \ref{eq_ss2} - \ref{eq_ss3}).
When the signal moves to an $X$-correction command, then it won't propagate any further.
If the signal was added to another $Z$-correction of a measured vertex, then signal shifting can be applied again until no $Z$-corrections are left on non-output vertices.
Therefore signals move along a path created by the $Z$-corrections.
The propagation of signals in an MBQC pattern can be described by a \emph{$Z$-path} as defined below.

\begin{definition}
	\label{def_zpath}
	Let $M$ be a measurement pattern on an open graph $(G, I, O)$.
	Then we  define a directed acyclic graph, called $G_Z$, on the vertices of $G$ such that there exists a directed edge from $i$ to $j$ iff there exists a correction command $Z^{s_i}_j$ in $M$.
	A path in $G_Z$ between two vertices $v$ and $u$ is called a \emph{$Z$-path}.
\end{definition}

The above definition allows us to state a simple observation about connectivity of a graph with flow.

\begin{lemma}
	\label{lem_zpath_conn}
	If $(f, \prec_f)$ is a flow on an open graph $(G, I, O)$, and there exists a $Z$-path from vertex $i$ to vertex $j$, then the vertices $i$ and $f(j)$ cannot be connected.
\end{lemma}
\begin{proof}
	The existence of a $Z$-path from $i$ to $j$ implies that $i \prec_f j$.
	The Z dependency graph is an acyclic graph, thus $i \neq j$.
	If $i$ would be connected to $f(j)$, then according to the flow property (F2):
	\begin{align*}
		i \in N(f(j)) \quad \land \quad i \neq j \quad \Rightarrow \quad j \prec_f i
	\end{align*}
	Now we have two contradicting strict partial order relations $i \prec_f j$ and $j \prec_f i$.
	Therefore $i$ cannot be connected to $f(j)$.
\end{proof}

Recall that the addition of signals is done modulo 2, therefore, if an even number of signals from a measured vertex $i$ is added to a correction command on vertex $j$, the signals will cancel out (since $Z^2 = X^2 = I$). Furthermore, it is evident from the rewrite Rules of \ref{eq_ss1} - \ref{eq_ss3} that after signal shifting, the measurement result of vertex $i$ will create a new $X$-correction over vertex $j$ if there exist an odd number of $Z$-paths from $i$ to a vertex $k$ that $j$ is $X$-dependent on directly before signal shifting. Similarly a new $Z$-correction from $i$ to $j$ will be created if there exists an odd number of $Z$-paths from $i$ to $j$. Either way,  \emph{the number of $Z$-paths from a vertex $i$ to another vertex $j$}, denoted as $\zeta_i(j)$, can be used to determine if the signal from $i$ should be added to a correction. We define $\zeta_i(i)$ to be 1 to simplify further calculations and definitions in this paper. The importance of the number of $Z$-paths will manifest itself in the next subsection, when the relation between signal shifting and gflows is studied.

We define a new structure called the \emph{signal shifted flow} (SSF), and show that it satisfies the three gflow properties in Definition \ref{def_gflow}. Before constructing the SSF, some definitions and lemmas are needed to justify our definition. Note that if an open graph $(G, I, O)$ has a flow $(f, \prec_f)$, then we can write the MBQC pattern of a deterministic computations on this open graph as \cite{DanosK06}:

\begin{align}
	\label{eq_flow_pattern}
	P = \prod\limits_{i \in O^C}^{\prec_f}
	\left(
		X^{s_i}_{f(i)} Z^{s_i}_{N(f(i)) \setminus \lbrace i \rbrace} M_i^{\alpha_i}
	\right)
	E_G
	N_{I^C}
\end{align}
where the product follows the strict partial order $\prec_f$ of the flow $(f, \prec_f)$.
From Equation \ref{eq_flow_pattern} we see that a $Z$-correction on a vertex $j$ depending on the measurement outcome of another vertex $i$ appears only if $j$ is a neighbour of $f(i)$.
This is formally stated in the next corollary as we will refer to it several times. 

\begin{corollary}
	\label{cor_zcorr}
	If $(G, I, O)$ is an open graph with a flow $(f, \prec_f)$, then there exists a $Z$-correction from vertex $i$ to another vertex $j$ iff $j \in N(f(i)) \setminus \lbrace i \rbrace$.
\end{corollary}

We define \emph{$Z$-dependency neighbourhood} of a vertex $j$ to be the set of vertices from which $j$ is receiving a $Z$-correction from. This set has an explicit form given as $N_Z(j) = \lbrace k \in O^C | f(k) \in N(j) \setminus \lbrace f(j) \rbrace \rbrace$, this is due to the following facts: for all vertices $k \in O^C$, from flow definition $f(k)$ exists also since $f(k) \notin \lbrace f(j) \rbrace$ hence  $k$ cannot be equal to $j$ and moreover since $f(k) \in N(j) \Rightarrow j \in N(f(k))$ therefore according to Corollary \ref{cor_zcorr} there exists a $Z$-correction from $k$ to $j$. It is easy to see, that $\zeta_i(j)$ can be written as:
\begin{align}
	\label{eq_paths}
	\zeta_i(j) = \sum\limits_{k \in N_Z(j)} \zeta_i(k)
\end{align}
There exists a $Z$-correction from every $k \in N_Z(j)$ to $j$.
These $Z$-corrections can be used to extend every such $Z$-path to $k$ to reach $j$.
If $i$ is in the sum, then because $\zeta_i(i) = 1$ the correct number of $Z$-paths is obtained with Equation \ref{eq_paths}.

We now present the complete algorithm (Algorithm \ref{alg_ss}) for signal shifting a flow pattern shown in Equation \ref{eq_flow_pattern}. We keep in mind that the order in which we apply the signal shifting rules does not matter \cite{DanosK06}.
\begin{algorithm} 
	\caption{$SignalShift$}
	\label{alg_ss}

	\KwIn{A measurement pattern $P$ with flow $(f, \prec_f)$ as defined in Equation \ref{eq_flow_pattern}.}
	
	\KwOut{The signal shifted pattern $P'$ of $P$.}
	
	\Begin
	{
		$B = O^C$\;
		$P' = P$\;
		\While{$B \neq \emptyset$}
		{
			select any vertex $i \in B$ which is smallest according to $\prec_f$\;
			$B = B \setminus \lbrace i \rbrace$\;
			\While{$\exists k \in B$ s.t. $Z_k^{s_i} \in P'$}
			{
				Move the $Z_k^{s_i}$ command next to the $M_k^{\alpha_k}$ command\;
				Use Rule \ref{eq_ss1} on $P'$ to create the signal command $S_k^{s_i}$\;
				\CommentSty{// Removes the $Z_k^{s_i}$ command from $P'$}\;
				Use Rule \ref{eq_ss2} on $P'$ to create a new $X_{f(k)}^{s_i}$ command\;
				\For{$j \in N(f(k)) \setminus \lbrace k \rbrace$}
				{
					Use Rule \ref{eq_ss3} on $P'$ to create a new $Z_j^{s_i}$ command.
				}
				Move $S_k^{s_i}$ to the end of the pattern and remove it.
			}
		}
	}
\end{algorithm}

\begin{proposition}
	\label{lem_alg_ss}
	Given the measurement pattern $P$ of a flow $(f, \prec_f)$ as defined in Equation \ref{eq_flow_pattern} as input to Algorithm \ref{alg_ss}, the output will be the signal shifted measurement pattern of $P$.
\end{proposition}
\begin{proof}
	We will prove this proposition by showing that:
	\begin{itemize}
		\item Algorithm \ref{alg_ss} terminates.
		\item Every step in Algorithm \ref{alg_ss} that modifies the pattern $P'$ is a valid application of a signal shifting rewrite rule.
		\item The output of Algorithm \ref{alg_ss}, the pattern $P'$, is signal shifted.
	\end{itemize}

	We begin by showing that Algorithm \ref{alg_ss} will terminate.
	The first ``while'' loop will obviously terminate, as we decrease the number of elements on each loop iteration and never add anything to the set $B$.
	The second ``while'' loop will not terminate only if some $Z_k^{s_i}$ command will be added to the pattern an infinite number of times.
	As the underlying graph is finite and a $Z_k^{s_i}$ command represents a directed edge in the $Z$-correction graph, this implies the existence of a cycle in the graph, however this is impossible according to the flow definition.
	The ``for'' loop in the algorithm terminates because the graph itself is finite, hence Algorithm \ref{alg_ss} has to terminate.
	
	For Algorithm \ref{alg_ss} to actually perform the signal shifting, its operations have to be either trivial commuting rules or the three signal shifting Rules \ref{eq_ss1} - \ref{eq_ss2}.
	As can be easily seen from the algorithm, the operations done are indeed the signal shifting rewrite Rules \ref{eq_ss1} - \ref{eq_ss2}.
	We still need to prove, that these rules can be applied in the order shown in the algorithm. Obviously we can use Rule \ref{eq_ss1} on line 8 to create the signal command due to the fact that $k \in B \subseteq O^C$ and that every non-output qubit is measured.
	Hence we have the measurement required for the creation of the signal command in the pattern. We know that $Z^{s_i}_k$ has to be in the pattern after the command $M_i^{\alpha_i}$ and before $M_k^{\alpha_k}$.
	The entanglement and creation commands are the first commands in the pattern and we do not need to move the $Z_k^{s_i}$ command past them.
	Hence we only need to move $Z_k^{s_i}$ past measurement commands on qubits that are not $i$ and $k$ and other correction commands.
	These can be done trivially and hence we can always move the $Z_k^{s_i}$ command next to $M_k^{\alpha_k}$ to apply Rule \ref{eq_ss1}.
	
	Next we want to move the newly created $S_k^i$ command to the end of the measurement pattern.
	To do that we need to commute it past the commands that appear after it.
	The only commands $S_k^i$ commutes non-trivially with are the ones that depend on the measurement of qubit $k$ as can be seen from Rules \ref{eq_ss1} - \ref{eq_ss3}.
	Those are the $X$- and $Z$-corrections depending on the measurement outcome of qubit $k$.
	According to Equation \ref{eq_flow_pattern} there is exactly one such $X$-correction in the pattern $P$, namely $X_{f(k)}^k$.
	Also the previous steps of the algorithm could not have created any dependencies from qubit $k$. The $Z$-correction commands have only been created depending on vertices that we already moved from $B$.
	Therefore we need to create exactly one new $X$-correction command using Rule \ref{eq_ss2}.
	We also look at the $Z$-corrections depending on $k$ and from Equation \ref{eq_flow_pattern} we see that in the original pattern these are on vertices from the set $N(f(k)) \setminus \lbrace k \rbrace$.
	As for the $X$-corrections we also have not created any new $Z$-corrections from $k$ in the previous steps of the algorithm.
	Hence this is exactly the set of corrections we need to commute with and apply Rule \ref{eq_ss3}. We are only left with commands after $S_k^i$ in the pattern that commute trivially with $S_k^i$. We can move the command at the end of the pattern.
	The signal command at the end of the pattern does not influence the computation and we will not add any new commands to the end of the pattern. Hence we can remove the $S_i^k$ command from the pattern.

	Finally we show that no more signal shifting rules can be applied after the completion of Algorithm \ref{alg_ss}, \emph{i.e.} the pattern $P'$ is signal shifted.
	We eliminate all $Z$-corrections acting on a non-output qubit depending on a vertex $i$ after removing it from the set $B$ and will afterwards never create any new $Z$-corrections depending on that vertex.
	At the end of the algorithm the set $B$ is empty, hence there cannot exist any non-output qubit that has a $Z$-correction acting on it and Rule \ref{eq_ss1} cannot be applied anymore. Moreover, since every signal command is at the end of the pattern, we cannot apply the Rules \ref{eq_ss2} and \ref{eq_ss3} neither, that completes the proof.
\end{proof}

We consider any trivial commutation of a pattern commands resulting to an equivalent pattern. Therefore the above algorithm defines the unique pattern obtained after signal shifting. Note that Algorithm \ref{alg_ss} works almost like a directed graph traversal, where there is a directed edge from vertex $i$ to $k$ iff there exists the command $Z_k^{s_i}$ in the measurement pattern.
The only difference from a classical directed graph traversal is that we  allow visiting of a vertex more than only once.
Hence we will traverse through every different path in the graph however we do that exactly once.

As mentioned before, the evenness of the number of $Z$-paths can be used to determine if a signal is added to a correction command.
Let $parity(n)$ be the function that determines the oddness or evenness of the integer $n$, \emph{i.e.} $parity(n) = n \mod 2$.
Then if an open graph has a flow, the oddness of $\zeta_i(j)$ can be found as described in the following lemma.

\begin{lemma}
	\label{lem_oddness}
	For every two vertices $i$ and $j$ in an open graph $(G, I, O)$ with flow $(f, \prec_f)$
	\begin{align*}
		parity(\zeta_i(j)) = | k \in \lbrace N_Z(j) | parity(\zeta_i(k)) = 1 \rbrace | \mod 2
	\end{align*}
	\emph{i.e.} $parity(\zeta_i(j))$ depends only on the number of vertices in the $Z$-dependency neighbourhood which have odd number of $Z$-paths from $i$.
\end{lemma}

\begin{proof}
	The oddness of $\zeta_i(j)$ can be written as
	\begin{align*}
		parity(\zeta_i(j)) &= \left(\sum\limits_{k \in N_Z(j)} \zeta_i(k)\right) \mod 2 = \\
		&= \sum\limits_{k \in N_Z(j)} \left( \zeta_i(k) \mod 2 \right) \mod 2 = \\
		&= \sum\limits_{\lbrace k \in N_Z(j) | 1 = \zeta_i(k) \mod 2 \rbrace} \left( \zeta_i(k) \mod 2 \right) \mod 2 = \\
		&= | \lbrace k \in N_Z(j) | 1 = \zeta_i(k) \mod 2 \rbrace | \mod 2 = \\
		&= | \lbrace k \in N_Z(j) | parity(\zeta_i(k)) = 1 \rbrace | \mod 2
	\end{align*}
\end{proof}

All these notions will allow us to define the structure of the pattern after signal shifting is being performed. 
\begin{proposition}
	\label{prop_ssf}
	Given a  flow $(f, \prec_f)$ on an open graph $(G, I, O)$, let $s$ be a function from $O^C \mapsto P^{I^C}$ such that $j \in s(i)$ iff $parity(\zeta_i(f^{-1}(j))) = 1$.
	Also define $L_s$ to be a layering function from $V(G)$ into a natural number:
	\begin{align*}	
		L_s(i) &= 0 &\forall i \in O \\
		L_s(i) &= \max_{j \in s(i)}(L_s(j) + 1) &\forall i \notin O
	\end{align*}
	Define the strict partial order $\prec_s$ with:
	\begin{align*}
		i \prec_s j \quad \Leftrightarrow \quad L_s(i) > L_s(j)
	\end{align*}
	Then, the application of signal shifting Rules \ref{eq_ss1} - \ref{eq_ss3} over an MBQC pattern with flow $(f, \prec_f)$ will lead to the following pattern:
	\begin{align}
		\label{eq_ssf_pattern}
		P = \prod_{j \in O, i \in I^C} Z_j^{s_i . parity(\zeta_i(j))}
		\prod\limits_{i \in O^C}^{\prec_s}
		\left(
			X^{s_i}_{s(i)} M_i^{\alpha_i}
		\right)
		E_G
		N_{I^C}
	\end{align}
\end{proposition}
\begin{proof}
	The proof is divided into three parts.
	First we will show that signal shifting creates exactly the pattern commands shown in Equation \ref{eq_ssf_pattern}.
	We proceed by showing, that the layering function $L_s$ is defined for every $i \in V(G)$.
	Lastly, we need to prove that using the partial order $\prec_s$ derived from $L_s$ for ordering the commands as in Equation \ref{eq_ssf_pattern} gives a valid measurement pattern.
	
	Note that the preparation commands ($N_I^C$), entanglement commands ($E_G$) and measurement commands ($M_i^{\alpha_i}$) are the same for Equations \ref{eq_flow_pattern} and \ref{eq_ssf_pattern}.
	Because signal shifting would not change these commands (Rules \ref{eq_ss1} - \ref{eq_ss3}) these are as required for a signal shifted pattern.
	Hence we need only to consider the correction commands.
	
	We will look at the correction commands that would appear in a signal shifted pattern.
	We do this by examining the signal shifting algorithm (Algorithm \ref{alg_ss}).
	As mentioned before, the algorithm works as a directed graph traversal, in a way that every distinct path is traversed o.
	As seen in the algorithm every $Z_k^{s_i}$ correction acting on a non-output qubit is removed from the pattern.
	This is in accordance with the proposed pattern  in Equation \ref{eq_ssf_pattern}.
	Let us examine which new corrections are created.
	
	The number of newly created $X_j^{s_i}$ depends on the number of times we enter the first loop with command $Z^{s_i}_{f^{-1}(j)}$.
	As the algorithm is a directed graph traversal algorithm, this happens as many times as there are different paths over the $Z$-dependency graph from $i$ to $f^{-1}(j)$.
	Because the same two $X_l^{s_i}$ corrections cancel each other, hence a new $X$-correction appears in a signal shifted pattern only if $parity(\zeta_i(f^{-1}(j))) = 1$. 
	We also note that no new $X_(f(i))^{s_i}$ correction is created since there exist no $Z$-path between $i$ and $f^{-1}(f(i))$. On the other hand Algorithm \ref{alg_ss} leaves the already existed $X$ corrections unchanged and moreover since we have defined $\zeta_i(i) = 1$ therefore $f(i) \in s(i)$. This implies that the set $s(i)$ does indeed contain all the vertices that have an $X$-correction depending on $s_i$ after signal shifting is performed.
	
	The number of newly created $Z$-corrections on an output vertex $j$ depending on a vertex $i$ appearing in the signal shifted pattern is equal to the number of different paths from $i$ to $j$.
	The difference with non-output qubits is that these will not be removed through the process of signal shifting.
	As with $X$-corrections, two $Z$-correction commands on the same qubit will cancel each other out and hence the existence of a $Z^{s_i}_j$ in the final pattern depends on the parity of the number of paths from $i$ to $j$.
	This can be written in short as:
	\begin{align*}
		Z_j^{s_i \cdot parity(\zeta_i(j))}
	\end{align*}
	Hence the measurement pattern in Equation \ref{eq_ssf_pattern} has exactly the same commands as the signal shifted pattern in Equation \ref{eq_flow_pattern}.
	
	Another thing we need to proof is that the layering function $L_s$ is defined for every $i \in V(G)$.
	As proven above, the $X$-corrections depending on the measurement of qubit $v$ correspond to the set $s(v)$.
	Hence we can interpret the definition of $L_s(v)$ as finding the maximum value of $L_s$ for every vertex that has an $X$-correction from $v$ and adding $1$ to it.
	The recursive definition of $L_s(v)$ is well defined, if for every non-output qubit we can find a path over $X$-corrections ending at an output qubit. We know that signal shifting of a valid pattern creates another valid pattern.
	This implies that the $X$-corrections cannot create a cyclic dependency structure and hence every path over the $X$-corrections has an endpoint. Moreover such a path cannot end on a non-output qubit $k$ since $f(k) \in s(k)$ and one could always extend that path with $f(k)$. Therefore $L_s(v)$ is well defined.

	Finally, it is easy to show that the partial order $\prec_s$ as used in Equation \ref{eq_ssf_pattern} gives a valid ordering of the commands.
	Every vertex $j$ that has an $X$-correction depending on the measurement of qubit $i$ has a smaller $L_s$ number and hence $i \prec_s j$.
	This way no $X$-correction command acts on an already measured qubit and because the $Z$-corrections are applied only on output qubits, the correction ordering is valid. Every other command is applied before the measurement command and hence the pattern in Equation \ref{eq_ssf_pattern} is a valid measurement pattern.
\end{proof} 

Given an open graph with a flow, we refer to the construction of the above proposition as its corresponding \emph{signal shifted flow} (SSF).
The main theorem of this section states that every SSF is actually a special case of a gflow. 

\begin{corollary}
	\label{cor_neighbour}
	If $(G, I, O)$ is an open graph with flow $(f, \prec_f)$ and SSF $(s, \prec_s)$ then for every vertex $i$ and $j$ such that $f(j) \in s(i) \setminus \lbrace f(i) \rbrace$, we can find another vertex $k$, such that $f(k) \in s(i) \cap N(j)$.
\end{corollary}

\begin{proof}
	If $f(j) \in s(i)$, then from the Proposition \ref{prop_ssf} of SSF we can conclude that $parity(\zeta_i(j)) = 1$.
	We know that $j \neq i$ from the assumptions.
	Lemma \ref{lem_oddness} says, that there must exist at least one other vertex $k$ from which $j$ has a $Z$-correction, such that $parity(\zeta_i(k)) = 1$.
	The flow definition says that $j$ must therefore be a neighbour of $f(k)$.
	Definition \ref{prop_ssf} of SSF states that $f(k)$ must therefore be in $s(i)$, hence $f(k) \in s(i) \cap N(j)$.
\end{proof}

\begin{theorem}
	\label{theorem_ssfisgflow}
	Given any open graph $(G, I, O)$ with flow $(f, \prec_f)$, the corresponding signal shifted flow $(s, \prec_s)$ is a gflow.
\end{theorem}

The proof is based on the following lemmas, demonstrating that $s$ is a gflow by satisfying all the properties of Definition \ref{def_gflow}.
The first property of gflow (property G1) is satisfied by SFF implicitly from Definition \ref{prop_ssf}, \emph{i.e.} for every $i \in V(G)$ it holds that $i \prec_s j$ if $j \in s(i)$.
Consider the second gflow property (G2), \emph{i.e.} if $j \in Odd(g(i))$ then $j = i$ or $i \prec_s j$.
We will show that every vertex with odd many connections to $s(i)$ has to be either $i$ itself or an output qubit.
This is a stronger condition than is needed to show G2, but as shown in Section \ref{sec_compactification}, necessary for creating compact circuits from SSF.

\begin{lemma}\label{lem_evenconnections}
	If $(s,\prec_s)$ is an SSF then every non-output vertex $v \neq i$ connected to $s(i)$ has an even number of connections to $s(i)$, \textit{i.e.},
	\begin{align*}
		\forall v \in N(s(i)) \setminus O \quad \land \quad v \neq i \quad \Rightarrow \quad v \notin Odd(s(i))
	\end{align*}
\end{lemma}

\begin{proof}
	Let $v \neq i$ be a vertex connected to $s(i)$, we show the following two sets have the same number of elements.
	\begin{align*}
		\lbrace k \in N_Z(v) \; | \; parity(\zeta_i(k)) = 1 \rbrace
		\quad \text{ and } \quad
		s(i) \cap N(v) \setminus \lbrace f(v) \rbrace
	\end{align*}
	For every $j \in s(i) \cap N(v) \setminus \lbrace f(v) \rbrace$, we prove $f^{-1}(j)$ is the unique element in 
\AR{
\lbrace k \in N_Z(v) \; | \; parity(\zeta_i(k)) = 1 \rbrace
}
Because $j \in s(i)$ from Proposition \ref{prop_ssf} there must exist $f^{-1}(j)$.
	Also since $j \in N(v)$ therefore $v \in N(j) = N(f(f^{-1}(j)))$. Moreover since $j\not = f(v)$, Corollary \ref{cor_zcorr} implies the existence of a $Z$-correction from $f^{-1}(j)$ to $v$, \emph{i.e.} $f^{-1}(j) \in N_Z(v)$.
	Proposition \ref{prop_ssf} says that because $j \in s(i)$, it must hold that $parity(\zeta_i(f^{-1}(j))) = 1$.
	Therefore $f^{-1}(j) \in \lbrace k \in N_Z(v) \; | \; parity(\zeta_i(k)) = 1 \rbrace$.
	
	On the other hand, for every vertex $u \in \lbrace k \in N_Z(v) \; | \; parity(\zeta_i(k)) = 1 \rbrace$, as $parity(\zeta_i(u)) = 1$ then from Proposition \ref{prop_ssf} we have$f(u) \in s(i)$. Also $f(u) \in N(v)$ because of Corollary \ref{cor_zcorr} and finally, $f(u) \not = f(v)$ because $v$ cannot have a $Z$-correction from itself, \emph{i.e.} $v \notin N_Z$. Hence it holds that $f(u) \in s(i) \cap N(v) \setminus \lbrace f(v) \rbrace$. Therefore 
		\begin{align*}
		| s(i) \cap N(v) \setminus \lbrace f(v) \rbrace | =
		| \lbrace k \in N_Z(v) \; | \; parity(\zeta_i(k)) = 1 \rbrace |
	\end{align*}
	According to Lemma \ref{lem_oddness} $parity(\zeta_i(v)) = | \lbrace k \in N_Z(v) \; | \; parity(\zeta_i(k)) = 1 \rbrace | \mod 2$. If $parity(\zeta_i(v)) = 0$ then $s(i) \cap N(v) \setminus \lbrace f(v) \rbrace$ must have an even number of elements.
	Proposition \ref{prop_ssf} says that $f(v)$ cannot be in $s(i)$ and therefore $v$ can have only even number of connections to $s(i)$. If $parity(\zeta_i(v)) = 1$ then we know that $s(i) \cap N(v) \setminus \lbrace f(v) \rbrace$ must have an odd number of elements.
	If $f(v)$ exists it must be in $s(i)$ because of Proposition \ref{prop_ssf}.
	In the case of $f(v) \in s(i)$, we can conclude that $parity(|s(i) \cap N(v)|) = 0$ and $v$ has even many connections to $s(i)$.
	On the other hand if $f(v)$ does not exist, $v$ has to be an output qubit because the flow function $f$ is defined for every non-output vertex.
	The only possibility of $k$ having odd many connections to $s(i)$ is therefore if $k$ is an output vertex, which proves the lemma.
\end{proof}

The next lemma directly proves that an SSF also satisfies the last gflow property (G3) which states that $i \in Odd(s(i))$.

\begin{lemma}
	\label{lem_oddi}
	If $(s, \prec_s)$ is an SSF, then for every $i \in O^C$ it holds that $i \in Odd(s(i))$.
\end{lemma}
\begin{proof}
	First we show that, performing signal shifting creates new $X$-corrections only between unconnected vertices.
	Recall that signal shifting creates a new $X$-correction between vertices $i$ and $j$ iff there exists a $Z$-path from $i$ to $f^{-1}(j)$ and an $X$ correction from $f^{-1}(j)$ to $j$ therefore from the Flow definition we have:
	\begin{align*}
		i \prec_f f^{-1}(j) \prec_f j
	\end{align*}
	Let us assume that there exists an edge between $i$ and $j$.
	According to the Flow definition we have that
	\begin{align*}
		\left.
		\begin{aligned}
			i \in N(j) \\
			j = f(f^{-1} (j))
		\end{aligned}
		\right\rbrace \quad \Rightarrow  \quad i \in N(f(f^{-1} (j))) \quad \Rightarrow \quad f^{-1} (j) \prec_f i
	\end{align*}
	This contradicts the partial order $i \prec_f f^{-1}(j) \prec_f j$ of the Flow $(f, \prec_f)$ and therefore there cannot be an edge between vertices $i$ and $j$.
	
	Next we claim that 	there is exactly one edge between $i$ and $s(i)$.
	According to Definition \ref{prop_ssf} of SSF, the set $s(i)$ consists only of the vertex $f(i)$ and the vertices to which signal shifting created a new X dependency from $i$.
	We showed that signal shifting does not create $X$ dependencies between connected edges.
	Hence, $f(i)$ is the only vertex in $s(i)$ that can be connected to $i$, and there must be an edge between $i$ and $f(i)$ because of the  flow property (F3) ($i \in N(f(i))$).
	\end{proof}

\begin{proof} Now to obtain the proof of Theorem \ref{theorem_ssfisgflow}, we note that the definition of SSF implies the gflow property \textit{(1.)}. Lemma \ref{lem_evenconnections} implies that every SSF satisfies the gflow condition \textit{(2.)}.
As the third and last gflow condition is satisfied by SSF according to Lemma \ref{lem_oddi}, SSF is indeed a gflow and Theorem \ref{theorem_ssfisgflow} holds.
\end{proof}

The above theorem for the first time presents an structural link between two seemingly different approach for parallelisation, gflow and signal shifting, for those patterns having already flow. As mentioned in the introduction this is the key step in obtaining our simultaneous depth and space optimisation. The next section explores further the link with gflow, showing optimality of SSF in parallelisation. 

\subsection{Influencing paths} \label{sec_ip}

The notions of \emph{influencing walks} and \emph{partial influencing walks} on open graphs with flow was introduced in \cite{BroadbentK09} to describe the set of all vertices that a measurement depends on. An influencing walk starts with an input and ends with an output vertex, a partial influencing walk starts with an input vertex but can end with a non-output vertex.  We will use a modified definition of influencing walks that can start from any non-output vertex $i$ and end at any vertex $j \in s(i)$ and call it a \emph{stepwise influencing path}. This will allow us to conveniently explore the dependency structure of a pattern with SSF.
 
\begin{definition}
	\label{def_path}
	Let $(s,\prec_s)$ be an SSF that is obtained from a flow $(f,\prec_f)$ of an open graph $(G,I,O)$ and vertices $i$ and $j$ in $V(G)$ such that $j \in s(i)$.
	We say that a path between vertices $i$ and $j$ is an \emph{stepwise influencing path}, noted as $\wp_i(j)$, iff
	\begin{itemize}
		\item The path is over the edges of $G$.
		\item The first two elements on the path are $i$ and $f(i)$.
		\item Every even-placed vertex $k$ on the path $\wp_i(j)$, starting from $f(i)$, is in $s(i)$.
		\item Every odd-placed vertex on the path $\wp_i(j)$ is the unique vertex $f^{-1}(k)$ of some $k \in s(i)$ such that $k$ is the next vertex on the path $\wp_i(j)$.
	\end{itemize}
\end{definition}

It is easy to see that every second edge, in particular the edges between $f^{-1}(k)$ and $k \in s(i)$, in the stepwise influencing path is a flow edge. Hence the path contains no consecutive non-flow edges.
If we restrict the first vertices of the stepwise influencing path to be input vertices, the stepwise influencing path would be a partial influencing path, but not \emph{vice versa}.
Stepwise influencing paths are useful because of their appearance in the SSF as proven by the following lemma.

\begin{lemma}
	\label{lem_path}
	Let $(s,\prec_s)$ be an SSF obtained from a flow $(f,\prec_f)$ of an open graph $(G,I,O)$ and vertices $i$ and $j$ in $V(G)$ such that $j \in s(i)$.
	Then there always exists a stepwise influencing path $\wp_i(j)$.
\end{lemma}
\begin{proof}
	We start by constructing such a path backward from $j$ to $i$.
	We select $j$ and $f^{-1}(j)$ as the last two vertices on the path and apply Corollary \ref{cor_neighbour} to find the vertices on the path, until we reach $i$.
	The formation of cycles is impossible, as this would imply a cyclic dependency structure, impossible for a flow.
	We have to reach $i$ as the set of vertices we choose from is finite.
\end{proof}

The above lemma will be used in Section \ref{sec_compactification} to obtain compact circuits from SSF. Note that there might be more than one stepwise influencing path from $i$ to $j$. We conclude the section about influencing paths with the following two lemmas which will be used to prove the optimality of SSF. First, the structure of stepwise influencing paths imposes a strict restriction on the way a vertex on the stepwise influencing path can be connected.

\begin{lemma}
	\label{lem_past_conn}
	Let $\wp_i(j)$ be a stepwise influencing path from $i$ to $j$ in an open graph $(G, I, O)$ with flow $(f, \prec_f)$ and corresponding SSF $(s, \prec_s)$.
	Then $f^{-1}(j)$ is the only odd-placed vertex in $\wp_i(j)$ that $j$ is connected to.
\end{lemma}
\begin{proof}
	According to the definition of stepwise influencing path, for every three consecutive vertices $v_1$, $v_2$, $v_3$ in $\wp_i(j)$ such that $v_1$ and $v_3$ are odd-placed we have that $v_2 = f(v_1)$ and $v_3 \in N(v_2) = N(f(v_1))$.
	According to Corollary \ref{cor_zcorr} there must exist a $Z$-correction from $v_1$ to $v_3$.
	Therefore the odd-placed vertices in $\wp_i(j)$ are on a $Z$-path from $i$ to $f^{-1}(j)$ and obviously from every odd-placed vertex in $\wp_i(j)$ there exists a $Z$-path to $f^{-1}(j)$.
	Lemma \ref{lem_zpath_conn} says that $j$ cannot be connected to any of the odd-placed vertices in $\wp_i(j)$.
\end{proof}

The previous lemma shows, that the stepwise influencing paths can be used to describe some properties of the connectivity in open graphs with SSF.
The next lemma (illustrated in Figure \ref{fig_ext_path}) will explain how a stepwise influencing path can be extended.
\begin{lemma}
	\label{lem_ext_path}
	Let $(G, I, O)$ be an open graph with flow $(f, \prec_f)$ and corresponding SSF $(s, \prec_s)$ and let $i$ and $j$ be two non-output vertices of the open graph such that $f(j) \in s(i)$.
	If $v \in N(j) \cap s(i) \setminus \lbrace f(j) \rbrace$ then every stepwise influencing path $\wp_i(v)$ can be extended by the vertices $j$ and $f(j)$ to create another stepwise influencing path $\wp_i(f(j))$.
\end{lemma}
\begin{proof}
	Adding $j$ and $f(j)$ to $\wp_i(v)$ satisfies the conditions for stepwise influencing paths.
	There exists an edge between vertices $j$ and $v$ and vertices $j$ and $f(j)$, hence it is a valid path. Moreover, $f(j) \in s(i)$ would be an even-placed vertex on the extended path, and $j$ would be the unique oddly-placed vertex with $f(j) \in s(i)$.
\end{proof}
\begin{figure}[ht]
	\begin{minipage}[t]{0.47\linewidth}
	\begin{center}
		\includegraphics[width=\textwidth]{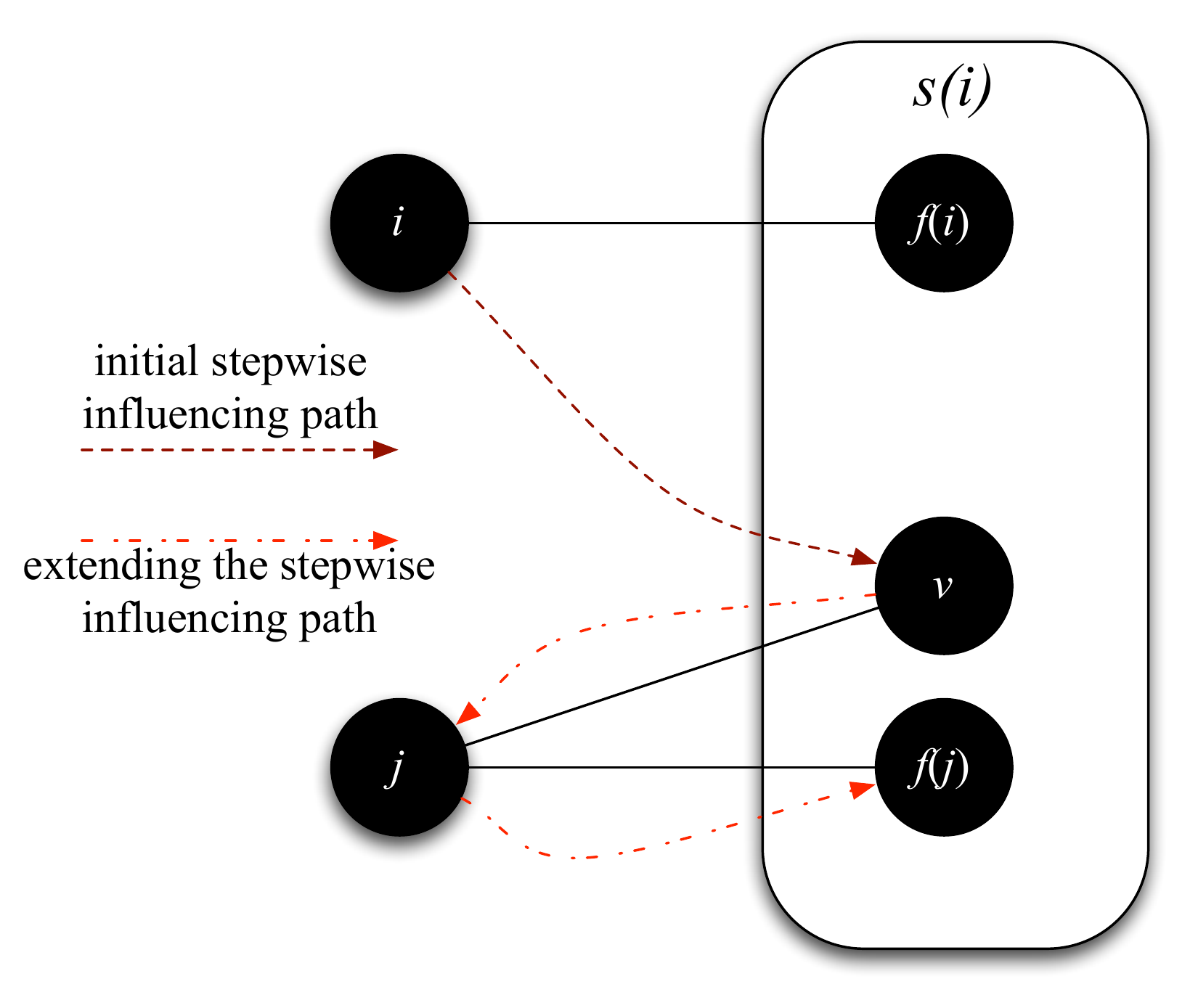}
	\end{center}
	\caption{
		Extending a stepwise influencing path ending at vertex $v$ according to Lemma \ref{lem_ext_path}.}
	\label{fig_ext_path}
	\end{minipage}
	\begin{minipage}[t]{0.1\linewidth}
		\hspace{\textwidth}
	\end{minipage}
	\begin{minipage}[t]{0.41\linewidth}
	\begin{center}
		\includegraphics[width=\textwidth]{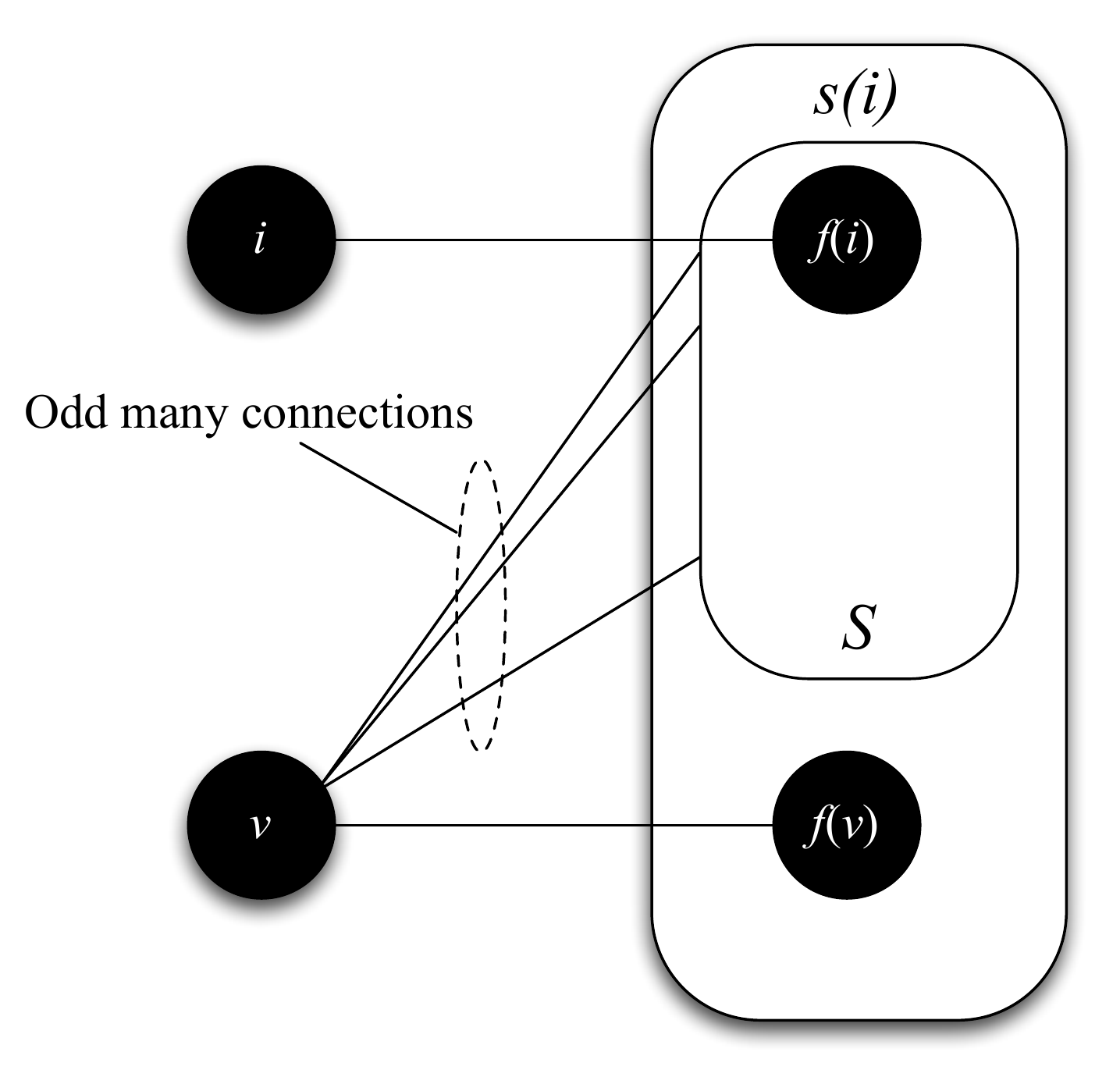}
	\end{center}
	\caption{
		For every strict subset $S$ of $s(i)$ containing $f(i)$ we can find a vertex $v$ in the odd neighbourhood of $S$ such that $f(v)$ is not contained in $S$. This is proven in Lemma \ref{lem_strict_subset}.}
	\label{fig_strict_subset}
	\end{minipage}
\end{figure}

\subsection{Optimality of signal shifting}

Given an MBQC pattern with gflow, finding the maximally delayed gflow of its underlying graph could potentially further reduce the depth of the computation \cite{MhallaP07}. A natural question that arises is how SSF is linked with the optimal gflow. In this section, we prove that if the input and output sizes of the pattern are equal, then SSF is indeed the optimal gflow. Hence we can conclude the most optimal parallelisation that one could obtain via translation of a quantum circuit into an MBQC pattern is achieved by the simple rewriting rules of SSF. This will also lead to a more efficient algorithm than the one presented in \cite{MhallaP07} for finding the maximally delayed gflow of a graph as we discuss later.
\begin{theorem}
	\label{thm_optimality}
	Let $(G, I, O)$ be an open graph with flow $(f, \prec_f)$ such that $|I|=|O|$.
	Let $(s, \prec_s)$ be the SSF obtained from $(f, \prec_f)$.
	Then $(s, \prec_s)$ is the optimal gflow of $(G, I, O)$.
\end{theorem}

The proof of the theorem is rather long, an outline is presented below. A general reader could omit the next subsections, however various novel constructions has been introduced in the proof that could be explored for other MBQC results and hence could be valuable for an MBQC expert. In Section \ref{sec_optimality_pen_layers} we show that the penultimate layers of an optimal gflow and an SSF of an open graph where $|I| = |O|$, are equal. Next we introduce the concept of a \emph{reduced open graph} in Section \ref{sec_optimality_reduced_og}. We prove two key properties of the optimal gflow and SSF of the reduced open graph. This highlights the recursive structures of the gflow and SSF leading to the possibility of extending these notions to new domains \footnote{For example, the authors are currently exploring this structure to define the concept of partial flow, for patterns with no deterministic computation.}. In Section \ref{sec_optimality_induction} we put the pieces together, by showing that the previous properties imply that reduced gflow (implicitly also optimal gflow and SSF) layers are equal to the original gflow layers from layer 1 onward. This allows us to construct a recursive proof for Theorem \ref{thm_optimality}, which we present in Section \ref{sec_optimality_proof}.

\subsubsection{The last two layers}
\label{sec_optimality_pen_layers}

The equality of the last layers of an SSF and optimal gflow follows from Lemma \ref{lem_last_layer} and Proposition \ref{prop_ssf} -- the last layer of an optimal gflow and an SSF is always the set of output vertices.
What is left to prove is that the penultimate layers are also equal, for doing so we need the following properties of open graphs with SSF.
An illustration of the property proven in the first of the two lemmas is shown in Figure \ref{fig_strict_subset}.

\begin{lemma}
	\label{lem_strict_subset}
	Let $(G, I, O$) be an open graph with flow $(f, \prec_f)$ and corresponding SSF $(s, \prec_s)$.
	If $i \in O^C$ then for every strict subset $S$ of $s(i)$ containing $f(i)$ there must exist a non-output vertex $v$ that is oddly connected to $S$ such that $f(v) \in s(i) \setminus S$, \emph{i.e.}
	\begin{align*}
		\forall i \in O^C \quad
		\forall S \subset s(i) \quad s.t. \quad f(i) \in S \quad
		\exists v \in Odd(S) \quad s.t. \quad f(v) \in s(i) \setminus S 
	\end{align*}
\end{lemma}
\begin{proof}
	If $s(i) = \lbrace f(i) \rbrace$ the lemma holds trivially, as there does not exist any nonempty strict subsets of $s(i)$.
	Consider the case where $s(i)$ contains more than one element and $S$ is a strict subset of $s(i)$.
	Then we select any vertex $j \notin S$ from $s(i)$ and look at the stepwise influencing paths from $i$ to $j$.
	Note that there might be more than one such path.
	We move backwards from $j$ towards $i$ over the stepwise influencing paths in the following way:
	\begin{enumerate}
		\item Move by two vertices
		\subitem{1.1} If possible, choose any stepwise influencing path where the previous even-placed element is not in $S$ and move to that element.
		\subitem{1.2} If the previous even-placed elements in all the stepwise influencing paths from $i$ to $j$ are in $S$, then stop.
		\item Repeat step 1.
	\end{enumerate}
	Let $u$ be the vertex to where we moved using the above process, $u$ has to exist because of the way we initially selected $j$.
	There are a couple of other observations that we can make about $u$.
	First, $u\in s(i)\setminus S$, because of the selection of $j$ and the way we moved on the paths.
	Second, $u$ cannot be the first even placed vertex on a stepwise influencing path from $i$ to $u$ because the first element is $f(i) \in S$ (according to Definition \ref{def_path}).
	Third, for every stepwise influencing path ending in $u$, the previous even-placed vertex has to be in $S$ as otherwise we could have moved one more step towards $i$.
	
	Considering the previous three observations we can show that the vertex $v = f^{-1}(u)$ must be oddly connected to $S$.
	We begin by noting that $v$ cannot be connected to any vertex $k \in s(i) \setminus (S \cup \lbrace f(v) \rbrace$.
	Otherwise, according to Lemma \ref{lem_ext_path}, we could extend any stepwise influencing path ending at $k$ with $v$ and $f(v)$. Hence
	$k \notin S \cup \lbrace f(v) \rbrace$ would then be an even-placed vertex on a stepwise influencing path from $i$ to $f(v)$.
	In particular, $k$ would be the second to last even-placed vertex on a stepwise influencing path from $i$ to $f(v) = u$
	Every such vertex, except $f(v)$ itself, is in $S$ as mentioned before.
	Because, according to Lemma \ref{lem_evenconnections}, $v$ has to be evenly connected to $s(i)$, it has to be oddly connected to $S$ and Lemma \ref{lem_strict_subset} holds.
\end{proof}

Next we need to show that every non-input vertex $i$ has a corresponding unique vertex $f^{-1}(i)$, this is only true for those graphs with $|I|=|O|$.

\begin{lemma}
	\label{lem_io_size}
	If $(f,\prec_f)$  is a flow on an open graph (G,I,O), then $|I| = |O|$ iff for every $j \in I^C$ there exists $f^{-1}(j)$.
\end{lemma}
\begin{proof}
	First, if $|I| = |O|$ then also $|I^C| = |O^C|$.
	The flow definition uniquely defines $f(i)$ for every $i \in O^C$ and therefore $f^{-1}(j)$ is uniquely defined for some, but not necessarily for all, vertices $j \in I^C$.
	The number of vertices for which $f$ is defined must equal the number of vertices for which $f^{-1}$ is defined and because $|I^C| = |O^C|$, $f^{-1}$ must be defined for every element in $I^C$.
	
	Second, Let us consider the case when for every $j \in I^C$ there exists $f^{-1}(j)$.
	The number of elements for which $f^{-1}$ is defined equals the number of elements $f$ is defined for.
	$f$ is by Definition \ref{def_flow} defined for every element in $O^C$.
	Hence $|I^C| = |O^C|$ which implies that $|I| = |O|$.
\end{proof}

Note that the above requirement, \emph{i.e.} the existence of $f^{-1}(i)$, is the only reason why our proof of Theorem \ref{thm_optimality} fails if $|I| \neq |O|$. We conjecture that by padding the input with necessary ancilla qubits without changing the underlying computation we could extend the above theorem to the general graphs. However the proof of such result is outside of the scope of this paper and not relevant for the optimisation of quantum circuit.

Note that because of Definition \ref{def_delayed_gflow} if a gflow is not optimal, its penultimate layer has to either be equal to the penultimate layer of the optimal gflow or there exists a vertex in the penultimate layer of optimal gflow that is not included in the penultimate layer of the other gflow. In the proof of the main result we assume that the penultimate layers are not equal, hence we could choose a vertex with particular properties (described in the next two lemmas) to derive a contradiction. 

\begin{figure}[h]
	\begin{minipage}[t]{0.40\linewidth}
	\begin{center}
		\includegraphics[width=\textwidth]{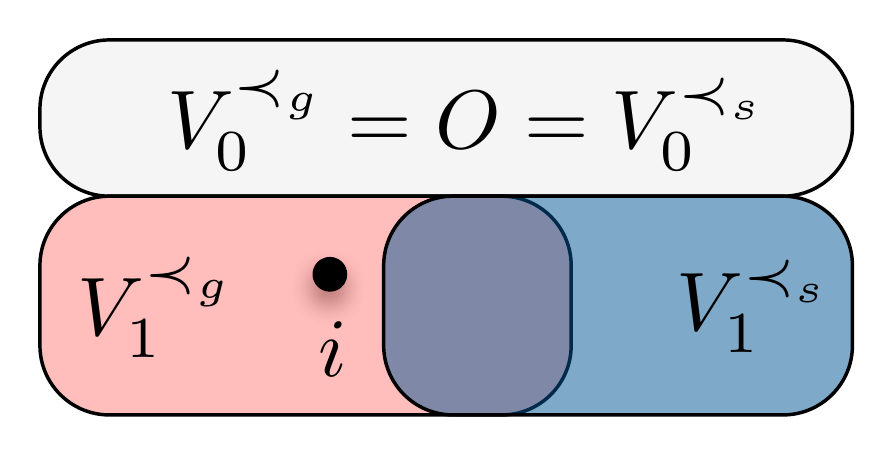}
	\end{center}
	\caption{
		The initial conditions required for Lemma \ref{lem_max_base}.}
	\label{fig_max_base_initial}
	\end{minipage}
	\begin{minipage}[t]{0.1\linewidth}
		\hspace{\textwidth}
	\end{minipage}
	\begin{minipage}[t]{0.44\linewidth}
	\begin{center}
		\includegraphics[width=\textwidth]{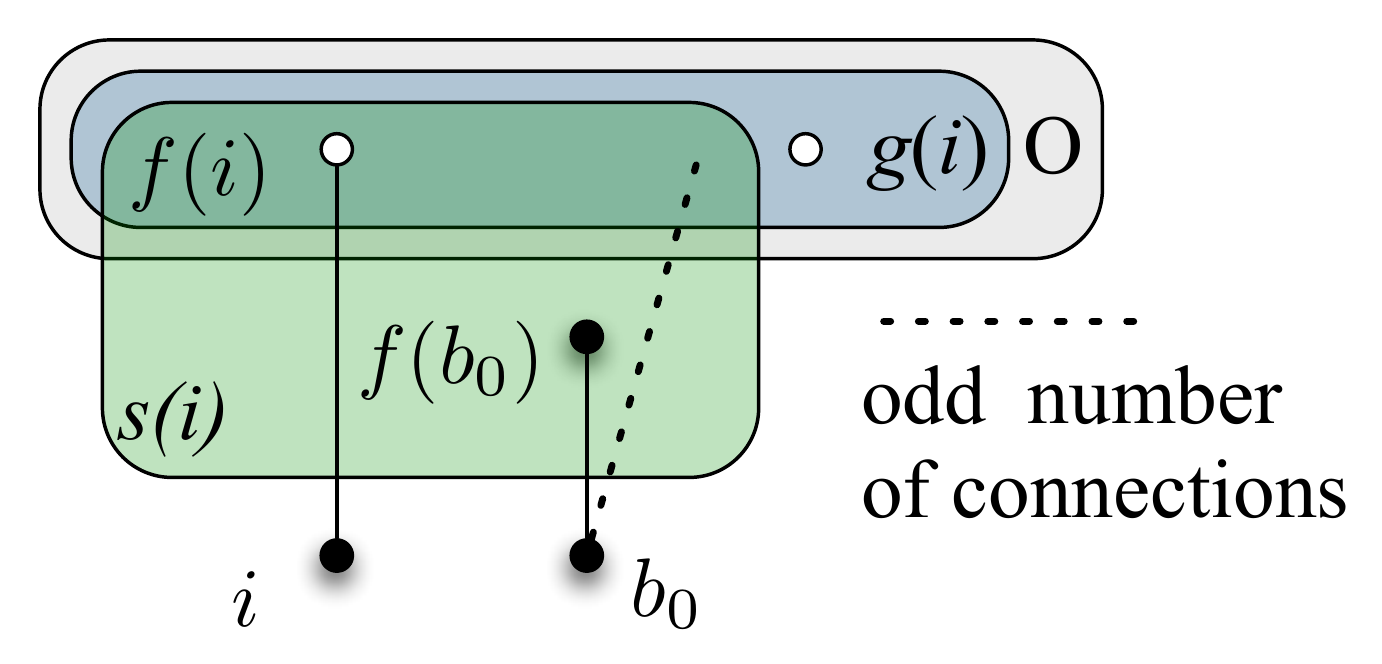}
	\end{center}
	\caption{
		The final conditions proved in Lemma \ref{lem_max_base}.}
	\label{fig_max_base_final}
	\end{minipage}
\end{figure}

\begin{lemma}
	\label{lem_max_base}
	Let $(G, I, O)$ be an open graph where $|I| = |O|$ with flow $(f, \prec_f)$, corresponding SSF $(s, \prec_s)$ and a gflow $(g, \prec_g)$ such that $V_0^{\prec_g} = O$. Assume there exists a vertex $i \in V_1^{\prec_g} \setminus V_1^{\prec_s}$, then
	\begin{itemize}
		\item $g(i) \subseteq O$
		\item $g(i) \cap s(i) \subset s(i)$
		\item $f(i) \in g(i)$
	\end{itemize}
	and there exists a vertex $b_0$ such that
	\begin{itemize}
		\item $b_0 \in Odd(g(i) \cap s(i))$
		\item $f(b_0) \in s(i) \setminus g(i)$
	\end{itemize}
\end{lemma}
\begin{proof}
	Because $i$ is in $V_1^{\prec_g}$ the set $g(i)$ must be a subset of $V_0^{\prec_g} = O$ according to Definition \ref{def_gflow_depth}.
	Proposition \ref{prop_ssf} implies that $V_0^{\prec_s} = O$.
	This and the fact that $i \notin V_1^{\prec_s}$ implies that $s(i)$ is not a subset of the output vertices $O = V_0^{\prec_s}$.
	Therefore there must exist a non-output vertex in $s(i)$ and, because $g(i) \subseteq O$, this vertex cannot be contained in $g(i)$.
	Thus the intersection of $s(i)$ and $g(i)$ cannot be equal to $s(i)$ and $g(i) \cap s(i) \subset s(i)$.
	
	We now show that $f(i) \in g(i)$.
	Let us assume that $f(i) \notin g(i)$, and choose a vertex $a_1 \in g(i)$ connected to $i$, such a vertex has to exist because the gflow definition says that $i$ is oddly connected to $g(i)$.
	As $a_1 \in g(i)$ then by the gflow definition $a_1$ cannot be an input qubit.
	According to Lemma \ref{lem_io_size}, there must exist a vertex $f^{-1}(a_1)$ to which $a_1$ is connected to.
	By the definition of flow, $f^{-1}(a_1)$ cannot be an output vertex and thus is not in layer $V_0^{\prec_g}$.
	As $g(i) \subseteq O$ this also means $f^{-1}(a_1) \notin g(i)$.
	On the other hand $f^{-1}(a_1)$ is connected to $a_1 \in g(i)$.
	Because $i \in V_1^{\prec_g}$ and $f^{-1}(a_1) \notin V_0^{\prec_g}$ we know from Definition \ref{def_gflow_depth} that $i \not\prec_g f^{-1}(a_1)$.
	As $f^{-1}(a_1)$ is connected to $g(i)$ we can conclude from the gflow definition that $f^{-1}(a_1)$ has to be evenly connected to $g(i)$ and therefore has at least one more connection to a vertex $a_2 \in g(i)$.

	Using the same argument for $a_2$ as for $a_1$ we can say that there must exist $f^{-1}(a_2) \notin g(i)$ to which $a_2$ is connected to.
	Let us assume that $f^{-1}(a_2)$ is not connected to $a_1$.
	This means it has only one connection to the set $A_2 = \lbrace a_1, a_2\rbrace \subseteq g(i)$ and is therefore oddly connected to it.
	We can continue this procedure of selecting vertices from $g(i)$ until we select a vertex $a_n$ such that $f^{-1}(a_n)$ is connected to at least one vertex $a_j$ in $A_{n-1} = \lbrace a_1, \dots a_{n-1}\rbrace \subseteq g(i)$.
	If this happens we can no longer say with certainty that $f^{-1}(a_n)$ is oddly connected to $A_n \subseteq g(i)$, which means we cannot select any more elements from $g(i)$ using this method.
	Because $(G, I, O)$ is a finite open graph we must find this $a_n$ in finite number of steps.

	We created the set $A_n$ in such a way that:
	\begin{align*}
		\forall j \in \lbrace 1, 2, \dots, n-1 \rbrace \quad f^{-1}(a_j) \in N(a_{j+1}) = N(f(f^{-1}(a_{j+1})))
	\end{align*}
	Hence we have a $Z$-correction from every $f^{-1}(a_{j+1})$ to $f^{-1}(a_j)$ and thus there exists a $Z$-path from $f^{-1}(a_n)$ to every $f^{-1}(a_j)$ such that $a_j \in A_{n-1}$ and , because of Lemma \ref{lem_path}, $f^{-1}(a_n)$ cannot be connected to any vertex in $A_{n-1}$.
	This leads to a contradiction with the assumption that it is connected to at least one vertex in $A_{n-1}$.
	Therefore our initial assumption that $f(i) \notin g(i)$ must be false and $g(i)$ must contain $f(i)$.

	From the definition of SSF we have that $f(i) \in s(i)$ and therefore also $f(i) \in g(i) \cap s(i)$. Now we know that $g(i) \cap s(i)$ is a strict subset of $s(i)$ containing $f(i)$; the existence of $b_0$ follows from Lemma \ref{lem_strict_subset}.
\end{proof}

Now we prove that if we have a vertex with the same properties as $b_0$ in Lemma \ref{lem_max_base} and a (possibly empty) subset $A$ of vertices with particular properties (which will be defined in the next lemma) we can always increase the size of $A$ and find another vertex with properties of $b_0$. This would imply the possibility of increasing the size of $A$ to infinity and will give us the contradiction we need.

\begin{lemma}
	\label{lem_max_step}
	Let $(G, I, O)$ be an open graph where $|I| = |O|$ with flow $(f, \prec_f)$, corresponding SSF $(s, \prec_s)$ and a gflow $(g, \prec_g)$.
	If we have a vertex $i$ in the open graph such that
	\begin{itemize}
		\item $g(i) \subseteq O$
		\item $g(i) \cap s(i) \subset s(i)$
		\item $f(i) \in g(i)$
	\end{itemize}
	and if we have a subset $A \subseteq g(i)$ and another vertex $b_0$ such that
	\begin{itemize}
		\item $b_0 \in Odd(g(i) \cap s(i))$
		\item $f(b_0) \in s(i) \setminus g(i)$
		\item $\forall j \in A \quad \exists \quad b_0 \buildrel Z \over \longrightarrow f^{-1}(j)$
	\end{itemize}
	then there exists another vertex $c_o$ and a non empty set $B \subseteq g(i)$ such that
	\begin{itemize}
		\item $B \neq \emptyset$
		\item $B \cap A = \emptyset$
		\item $c_0 \in Odd(g(i) \cap s(i))$
		\item $f(c_0) \in s(i) \setminus g(i)$
		\item $\forall j \in A \cup B \quad \exists \quad c_0 \buildrel Z \over \longrightarrow f^{-1}(j)$
	\end{itemize}
\end{lemma}
\begin{proof}
	The proof consists of three steps:
	we start by constructing the set $B$;
	we proceed with finding the vertex $c_0$;
	and finally we prove that $c_0$ has the required properties.

	Define $S = g(i) \cap s(i)$, since $f(b_0)$ exists hence $b_0$ cannot be an output vertex.
	Also since $g(i) \subseteq O$ therefore $b_0$ is not in $g(i)$. As $b_0 \notin O = V_0^{\prec_g}$ and $g(i) \subseteq O$ we can conclude from Definition \ref{def_gflow_depth} that $i \in V_1^{\prec_g}$ and $i \not\prec_g b_0$.
	Therefore according to the gflow definition, $b_0$ must  be in the even neighbourhood of $g(i)$.
	We also know from the initial conditions of this lemma that $b_0$ is in the odd neighbourhood of $g(i) \cap s(i)$.
	Thus there has to exist a vertex $v_1$ in $g(i)$ to which $b_0$ is connected to, but which is not included in $g(i) \cap s(i)$, \emph{i.e.} $v_1 \in g(i) \setminus s(i)$.
	As $g : O^C \rightarrow P^{I^C}$, $v_1 \in g(i)$ cannot be an input qubit and because $f^{-1}$ exists for every non-input qubit according to Lemma \ref{lem_io_size}, there must exist a vertex $f^{-1}(v_1) = b_1$.
	It is also important for the later part of the proof to note that $f(b_1) = v_1 \notin A$. This is due to Lemma \ref{lem_zpath_conn}, which implies that $b_0$ cannot be connected to any vertex in $A$.

	Define $B_0 = S$ and consider the case when $b_1$ is evenly connected to $B_0$.
	Remember that the flow property (F3) says that there is always an edge between $b_1$ and $f(b_1)$.
	This means that $b_1$ is oddly connected to $B_1 = \lbrace f(b_1) \rbrace \cup B_0$ which is a subset of $g(i)$.
	But again because of the gflow property (G2) we have that $b_1$ must be evenly connected to $g(i)$.
	Thus there must exist another vertex $b_2$ such that $b_1$ is connected to $f(b_2) \in g(i) \setminus B_1$, otherwise $b_1$ could not be in the even neighbourhood of $g(i)$.
	If $b_2$ is evenly connected to $B_1$, it must be oddly connected to $B_2 = \lbrace f(b_2) \rbrace \cup B_1$ which is again a subset of $g(i)$.
	If $b_2$ is oddly connected to $B_2$ there must exist a vertex $b_3$ such that $b_3$ is connected to $f(b_3) \in g(i) \setminus B_2$, otherwise $b_2$ could not be in the even neighbourhood of $g(i)$.
	We can continue this scheme until we get to vertex $b_n$ that is oddly connected to $B_{n-1}$.
	As $B_n = \lbrace f(b_n) \rbrace \cup B_{n-1}$ and there exists an edge between $b_n$ and $f(b_{n})$ we get that $b_n$ must be evenly connected to $B_n$.
	Such vertex $b_n$ must exist, otherwise we could continue selecting elements from $g(i)$ infinitely, but $(G, I, O)$ is a finite open graph.
	We select $B = B_n \setminus S$. Recall that $f(b_1)$ must exist, therefore $B$ must have at least on element. 
	
	Next we show $b_n$ is oddly connected to $S$. We note that we have the following:
	\begin{align*}
		\forall j \in \lbrace 1, 2, \dots, n \rbrace \quad b_j \in N(f(b_j)) \quad  \land \quad b_{j-1} \in N(f(b_j))
	\end{align*}
	Corollary \ref{cor_zcorr} implies that for every $j > 0$ there exists a $Z$-correction from $b_j$ to $b_{j-1}$.
	Thus we have a $Z$-path from $b_n$ to every other $b_j$ where $j < n$, hence from Lemma \ref{lem_zpath_conn} we conclude $b_n$ cannot be connected to any vertex $f(b_j) \in B_{n-1}$ where $j < n$. The number of edges that connect the vertices in $B_{n-1}$ to vertex $b_n$ has to be the same as the number of edges between vertices of $S$ and $b_n$, because $B_{n-1} = \lbrace f(b_1), f(b_2), \dots, f(b_{n-1}) \rbrace \cup S$.	As $b_n$ was oddly connected to $B_{n-1}$, it must also be oddly connected to $S$. Note that however $b_n$ does not have the required properties for $c_0$, but will be used to find such a vertex.
	
	The gflow definition says that $b_n$ must be evenly connected to $s(i)$. It is also oddly connected to $s(i) \cap g(i)$ hence there must exist a vertex $c \in s(i) \setminus g(i)$ to which $b_n$ is connected to.
	According to Lemma \ref{lem_path} there exists a stepwise influencing path $\wp_i(c)$ and due to Definition \ref{def_path}, $f(i)$ has to be on on this path. Therefore there exists at least one element in $\wp_i(c)$ that is in $S$.
	Let $f(a_0)$ be the last element of the path $\wp_i(c)$ in $S$.
	
	Define $a_1$ to be the vertex in $\wp_i(c)$ that comes after $f(a_0)$.	We know that $a_1$ has odd many $Z$-paths from $i$ because Definition \ref{def_path} implies that $f(a_1) \in s(i)$.
	If $a_1$ is already oddly connected to $S$, then we are done and $a_1 = c_0$.
	If $a_1$ is evenly connected to $S \subset s(i)$, then we know that it must be oddly connected to $S \cup \lbrace f(a_1) \rbrace \subseteq s(i)$.
	There must exist another vertex $f(a_2) \in s(i) \setminus (S \cup \lbrace f(a_1) \rbrace)$ to which $a_1$ is connected to for it to be evenly connected to $s(i)$ as is required by Lemma \ref{lem_evenconnections}.
	Because $f(a_2) \in s(i)$ we know there exists a stepwise influencing path $\wp_i(f(a_2))$ (Lemma \ref{lem_path}) and we can extend that path by $a_1$ and $f(a_1)$ as was proven in Lemma \ref{lem_ext_path}.
	We move backward on this path and find the element $a_2$. If $a_2$ is oddly connected to $S$, we are done and set $c_0 = a_2$. Otherwise we can continue as was the case for $a_1$ until we find an element $a_m$ that is oddly connected to $S$.
	This element must exist since graph is finite and the $Z$ corrections do not create any loops.
	We select $c_0 = a_m$. Note that $a_m$ cannot be $i$ because $f(i) \in S = s(i) \cap g(i)$ but $f(a_m) \notin g(i)$.
	
	There is a $Z$-path from $a_m = c_0$ to $a_1$ (we moved backwards along this path to find $a_m$) and from $a_1$ to $b_n$ because of the way we selected $a_1$.
	There also exists a $Z$-path from $b_n$ to every other $b_j$ such that $0 \leq j < n$, thus $a_m$ will also have a $Z$-path to every $b_j$ in $\lbrace b_1, b_2, \dots, b_n \rbrace$.
	Even more, because:
	\begin{align*}
		&a_m \buildrel Z \over \longrightarrow b_n
		\quad \land \quad
		b_n \buildrel Z \over \longrightarrow b_0
		\quad \land \quad
		\forall j \in A \quad \quad b_0 \buildrel Z \over \longrightarrow f^{-1}(j)
		\quad \Rightarrow \\
		\Rightarrow \quad
		&\forall j \in A \quad \quad a_m \buildrel Z \over \longrightarrow f^{-1}(j)
	\end{align*}
	This completes the proof.
\end{proof}

Finally we could put together Lemmas \ref{lem_max_base} and \ref{lem_max_step}.

\begin{lemma}[\textbf{Equality of the penultimate SSF and optimal gflow layer}]
	\label{lem_equal_layer}
	Let $(G, I, O)$ be an open graph with flow $(f, \prec_f)$, corresponding SSF $(s, \prec_s)$ and optimal gflow $(g, \prec_g)$ such that $|I| = |O|$. Then $V_1^{\prec_s} =  V_1^{\prec_g}$. 
\end{lemma}
\begin{proof}
Assume  $V_1^{\prec_s} \neq V_1^{\prec_g}$ we show how we can choose infinitely many different vertices from $V(G)$. Due to Definition \ref{def_delayed_gflow} we have $|V_1^{\prec_s}| \leq |V_1^{\prec_g}|$ and since $V_1^{\prec_s} \neq V_1^{\prec_g}$ hence trivially  $V_1^{\prec_g} \not\subset V_1^{\prec_s}$ and there must exist a vertex $i$ in $V_1^{\prec_g} \setminus V_1^{\prec_s}$. Then from Lemma \ref{lem_last_layer} we have $V^0_g = O$ and using Lemma \ref{lem_max_base} we obtain the following:
	\begin{itemize}
		\item $g(i) \subseteq O$
		\item $f(i) \in g(i)$
		\item $g(i) \cap s(i) \subset s(i)$
	\end{itemize}
	and that there exists another vertex $b_0$ such that
	\begin{itemize}
		\item $b_0 \in Odd(g(i) \cap s(i))$
		\item $f(b_0) \in s(i) \setminus g(i)$
	\end{itemize}
	These constraints together with an empty set $A$ allow us to apply Lemma \ref{lem_max_step}. Lemma \ref{lem_max_step} is constructed in such a way that whenever we can apply it to a  (possibly empty) set $A$, it proves the existence of another set $B$ such that that $|A| < |A \cup B|$ and Lemma \ref{lem_max_step} is applicable to the new set $A \cup B$.
	Thus it is possible to apply Lemma \ref{lem_max_step} infinitely many times and construct a subset of $V(G)$ containing infinitely many vertices. This leads to a contradiction as $G$ is a finite graph.
\end{proof}

\subsubsection{Reducing the open graph}
\label{sec_optimality_reduced_og}

The equality of penultimate layers of SSF and gflow might suggest that one could prove the equality of other layers simply by removing the last layer from the open graph and reapply the lemmas from the last section. However this would fail as the vertices in any layers can also use the output vertices in their correcting sets. Therefore we need to be careful which vertices we remove such that the reduced graph still have a gflow.

\begin{definition}
	\label{def_reduced_og}
	If $(G, I, O)$ is an open graph with flow $(f, \prec_f)$ and corresponding SSF $(s, \prec_s)$ then we call the open graph $(G', I, O')$ a \emph{reduced open graph according to} $(s, \prec_s)$, where
	\begin{itemize}
		\item $R = \lbrace v \in O \; | \; f^{-1}(v) \in V_1^{\prec_s} \rbrace$ is the set of removed vertices.
		\item $G' = (V', E')$ where
			\subitem $V' = V \setminus R$
			\subitem $E' = E \setminus (V \times R)$
		\item $O' = (V_1^{\prec_s} \cup O) \setminus R$
	\end{itemize}
\end{definition}

We will omit "according to..." and call $(G', I, O')$ just reduced open graph when it is clear from the text which SSF is used for constructing it.
An example of a reduced open graph is shown in Figure \ref{fig_reduced_og}

\begin{figure}[h]
	\begin{center}
		\resizebox{\hsize}{!}{\includegraphics{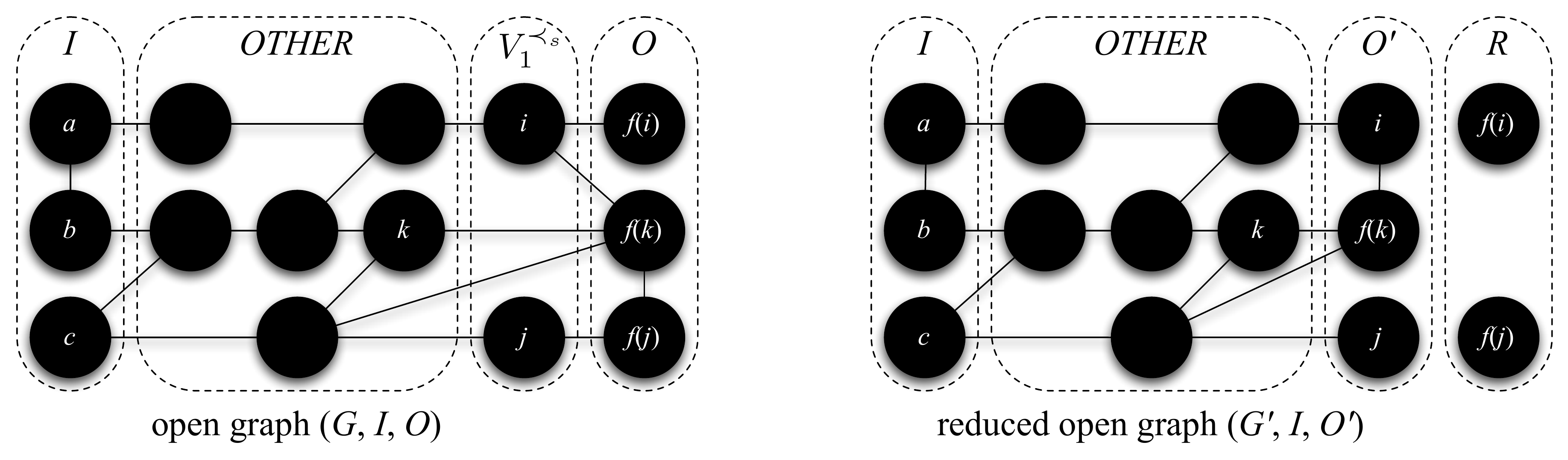}}
	\end{center}
	\caption{An example of an SSF reduced open graph (right) together with the original open graph (left).}
	\label{fig_reduced_og}
\end{figure}

As we saw in the previous section, we needed the fact that $|I|=|O|$ to be able to prove that the penultimate layers of SSF and optimal gflow are equal.
If we want to apply the same lemmas to the new reduced open graph, we need to guarantee that if we start with a graph where input size equals output size, the same holds for the reduced open graph.

\begin{lemma}
	\label{lem_reduced_og_output}
	Let $(G', I, O')$ be a reduced open graph of the open graph $(G, I, O)$, then $|O| = |O'|$.
\end{lemma}
\begin{proof}
	Let $R$ be the set of vertices removed from $G$, then for every vertex $i \in V_1^{\prec_s}$ we have a corresponding unique vertex $f(i)$ in $R$ since Proposition \ref{prop_ssf} implies that $s(i) \subseteq O$ and $f(i) \in s(i)$. On the other hand, for every vertex in $R$ there exists a corresponding vertex in $V_1^{\prec_s}$ from the definition of $R$. Therefore for every vertex $v \in R$ that we remove from $O$ when constructing $O' = (V_1^{\prec_s} \cup O) \setminus R$ we add another vertex $f^{-1}(v) \in V_1^{\prec_s}$ and it must hold that $|O| = |O'|$.
\end{proof}

The next lemma is used later to construct a gflow of the reduced open graph from the gflow of the original open graph.

\begin{lemma}
	\label{lem_set_property}
	Let $(G, I, O)$ be an open graph and $A$ and $B$ two sets in $O$ such that $Odd(B) \cap O^C = \emptyset$.
	Then $Odd((A \cup B) \setminus (A \cap B)) \cap O^C = Odd(A) \cap O^C$.
\end{lemma}
\begin{proof}
	There are altogether four different possibilities for a vertex $v \in O^C$ to be connected to the sets $A$  and $B$ satisfying $Odd(B) \cap O^C = \emptyset$ as shown in Figure \ref{fig_set_property}:
	\begin{align*}
		1) \quad &v \in Even(A) \cap Odd(A \setminus B) \Rightarrow
		v \in Odd(A \cap B) \Rightarrow
		v \in Odd(B \setminus A) \Rightarrow \\
		\Rightarrow \quad
		&v \in Even((A \setminus B) \cup (B \setminus A)) \\
		2) \quad &v \in Even(A) \cap Even(A \setminus B) \Rightarrow
		v \in Even(A \cap B) \Rightarrow
		v \in Even(B \setminus A) \Rightarrow \\
		\Rightarrow \quad
		&v \in Even((A \setminus B) \cup (B \setminus A)) \\
		3) \quad &v \in Odd(A) \cap Odd(A \setminus B) \Rightarrow
		v \in Even(A \cap B) \Rightarrow
		v \in Even(B \setminus A) \Rightarrow \\
		\Rightarrow \quad
		&v \in Odd((A \setminus B) \cup (B \setminus A)) \\
		4) \quad &v \in Odd(A) \cap Even(A \setminus B)\Rightarrow
		v \in Odd(A \cap B) \Rightarrow
		v \in Odd(B \setminus A) \Rightarrow \\
		\Rightarrow \quad
		&v \in Odd((A \setminus B) \cup (B \setminus A))
	\end{align*}
	We see that every time $v$ is evenly connected to $A$ it is also evenly connected to $(A \setminus B) \cup (B \setminus A)$ and every time $v$ is oddly connected to $A$ it is also oddly connected to $(A \setminus B) \cup (B \setminus A)$.
	Because $(A \setminus B) \cup (B \setminus A) = (A \cup B) \setminus (A \cap B)$ ans $v$ is in $O^C$ it must hold that $Odd((A \cup B) \setminus (A \cap B)) \cap O^C = Odd(A) \cap O^C$.
\end{proof}

\begin{figure}
	\begin{center}
		\resizebox{\hsize}{!}{\includegraphics{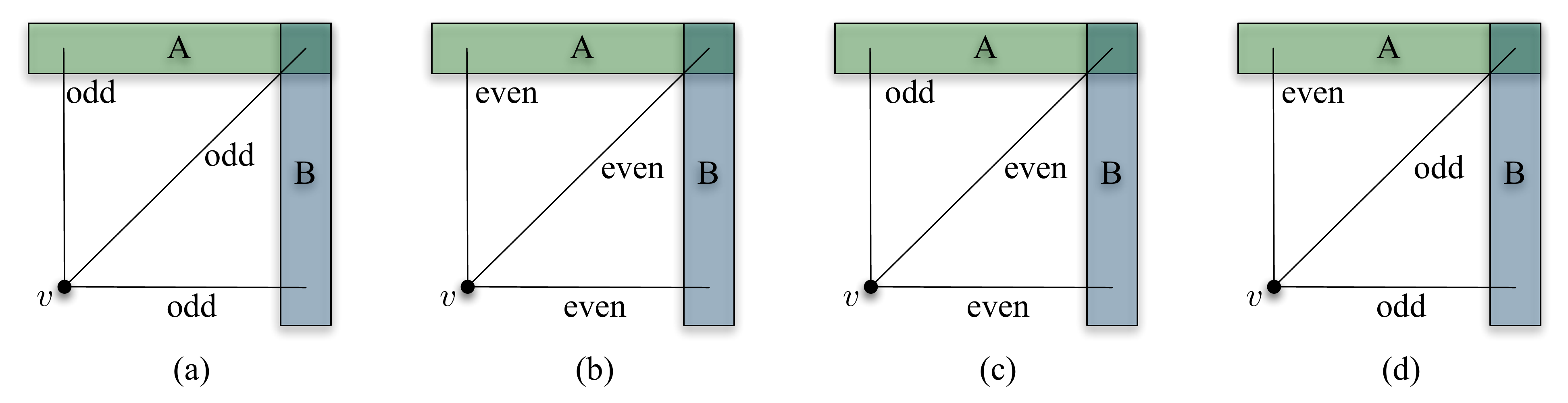}}
	\end{center}
	\caption{The four possibilities for a vertex $v \in O^C$ to be connected connected to the sets $A$ and $B$. The vertex $v$ can be either oddly (a) or evenly (b) connected the sets $A \setminus B$, $B \setminus A$ and $A \cap B$; oddly connected to $A \setminus B$ and evenly to $B \setminus A$ and $A \cap B$ (c); or evenly connected to $A \setminus B$ and oddly to $B \setminus A$ and $A \cap B$ (d).}
	\label{fig_set_property}
\end{figure}

We start by creating a function that will be proven to have the required properties of the gflow.

\begin{lemma}[\textbf{Finding the reduced gflow function}]
	\label{lem_reduced_gflow_existance}
	Let $(G, I, O)$ be an open graph with flow $(f, \prec_f)$, SSF $(s, \prec_s)$ and optimal gflow $(g, \prec_g)$ such that $|I| = |O|$.
	Let $(G', I, O')$ be the SSF reduced open graph of $(G, I, O)$ with the removed vertices set $R$, then  there exists a function $g' : O'^C \rightarrow P^{I^C \cap V(G')}$ such that:
	\begin{enumerate}
		\item $\forall i \in O'^C \quad g'(i) \cap O'^C = g(i) \cap O'^C$
		\item $\forall i \in O'^C \quad Odd(g'(i)) \cap O'^C = Odd(g(i)) \cap O'^C$
	\end{enumerate}
\end{lemma}
\begin{proof}
	We start by noting that according to Lemma \ref{lem_penultimate_layer} we can create an optimal gflow $(\tilde{g}, \prec_{\tilde{g}})$ of the open graph $(G, I, O \cup V^{\prec_g}_1) = (G, I, O' \cup R)$ by restricting $g$ to $V(G) \setminus (V^{\prec_g}_1 \cup V^{\prec_g}_0) = O'^C$ and setting $\prec_{\tilde{g}} = \prec_g \setminus V^{\prec_g}_1 \times V^{\prec_g}_0$. We construct our desired $g'$ function from $\tilde{g}$.
	
	We consider $i \in O'^C$, if there exists a vertex $j \in R \cap \tilde{g}(i)$ then from the reduced open graph definition we have $f^{-1}(j) \in V^{\prec_g}_1$. Also from Lemma \ref{lem_equal_layer} we have $V^{\prec_g}_1 = V^{\prec_s}_1$ and thus $f^{-1}(j) \in V^{\prec_s}_1$.
	According to Proposition \ref{prop_ssf} this means that $s(f^{-1}(j)) \subseteq O$. We have $Odd(s(f^{-1}(j)) \cap O'^C = \emptyset$ since the only odd neighbours of $s(f^{-1}(j))$ are either output vertices or the vertex $f^{-1}(j) \in V_1^{\prec_g} \subseteq O'$.

	Now we define $g'(i) = (\tilde{g}(i) \cup s(f^{-1}(j))) \setminus (\tilde{g}(i) \cap s(f^{-1}(j)))$, hence $j\not \in g'(i)$. Moreover Lemma \ref{lem_set_property} implies that $Odd(g'(i)) \cap O'^C = Odd(\tilde{g}(i)) \cap O'^C$. Also $g'(i) \cap O'^C = \tilde{g}(i) \cap O'^C$ since $s(f^{-1}(j)) \subseteq O'$.
	Note that, since the new set will be constructed via a union of two sets we might add another vertex $k \in R$ to the set $g'(i)$.
	However, we can remove any such vertex $k$ added to $g'(i)$ by applying the same procedure recursively.
	For every such vertex $k$, it must hold that $f^{-1}(j) \prec_f f^{-1}(k)$ since $k \in s(f^{-1}(j))$ and Proposition \ref{prop_ssf} implies the existence of a $Z$-path from $f^{-1}(j)$ to $f^{-1}(k)$. Now we remove $k$ via the above procedure \emph{i.e.} defining $g'(i) = (g'(i) \cup s(f^{-1}(k))) \setminus (g'(i) \cap s(f^{-1}(k)))$. If this would add vertex $j$ again to $g'(i)$, hence there exists a $Z$-path from $f^{-1}(k)$ to $f^{-1}(j)$ and $f^{-1}(k) \prec_f f^{-1}(j)$ which contradicts the previous relation.
	This procedure will eventually terminate and remove all undesired vertices $k \in R$ since in the above procedure we never create any $Z$-path loops.
\end{proof}

We call a function which satisfies properties (1) and (2) of Lemma \ref{lem_reduced_gflow_existance} the \emph{reduced gflow function} of $g$.
We can interpret these properties as saying that the gflow function $g'$ differs from the gflow function $g$ only by the vertices in $O'$, \emph{i.e.} the other elements in the correcting set are left unchanged.
As a gflow consists of a function and a partial order, we still need to define a valid partial order.
The one that is most useful to us is such that it preserves as much relations as possible from the original gflow, hence the layering structures remain similar.

\begin{lemma}[\textbf{Constructing the reduced gflow}]
	\label{lem_ordering}
	Let $(G, I, O)$ be an open graph with SSF $(s, \prec_s)$, gflow $(g, \prec_g)$ and $(G', I, O')$ a reduced open graph of $(G', I, O')$.
	If $g'$ is a reduced gflow function of $(g, \prec_g)$, then $(g', \prec_{g'})$ is a gflow of $(G', I, O')$, where
	\begin{align*}
		&\forall i, j \in O'^C \quad i \prec_g j \Leftrightarrow i \prec_{g'} j \\
		&\forall i \in O'^C, \forall j \in O' \cap g'(i) \quad \Rightarrow \quad i \prec_{g'} j
	\end{align*}
\end{lemma}
\begin{proof}
	We will show that $(g', \prec_{g'})$ satisfies the three gflow properties (G1) - (G3) in Definition \ref{def_gflow}. First property requires that if $j \in g'(i)$, then $i \prec_{g'} j$.
	This is obviously true if $j \in O'$.
	If $j \in g'(i)\cap O'^C$, from Lemma \ref{lem_reduced_gflow_existance} we have $j \in g(i)$ which implies that $i \prec_g j$ because $(g, \prec_g)$ is a gflow.
	Now according the definition of $\prec_{g'}$ it must also hold, that $i \prec_{g'} j$.
	
	Now we consider the gflow property (G2).
	For every $j \in Odd(g'(i))$ it must be that $j = i$ or $i \prec_{g'} j$.
	If $j \in O'$, then again this is obviously true because of the definition of $\prec_{g'}$.
	If $j \in Odd(g'(i)) \cap O'^C$ then we know that $j \in Odd(g(i))$ and $j = i$ or $i \prec_g j$.
	According to the definition of $\prec_{g'}$, $i \prec_g j$ implies that $i \prec_{g'} j$ and we have that if $j \in Odd(g'(i))$ then either $i = j$ or $i \prec_{g'} j$. Thus the gflow property (G2) is satisfied.
	
	Finally, we require for gflow property (G3) that $i \in Odd(g'(i))$ and as $i \in O'^C$ this is true because of the properties of $g'$.
\end{proof}

We call the gflow $(g', \prec_{g'})$ from Lemma \ref{lem_ordering} the \emph{reduced gflow} of $(g, \prec_g)$. Similarly we can construct the SSF of the reduced open graph. Note that an SSF can only exist if the reduced open graph has flow.
Thus arises the need to prove the existence of a flow on the reduced open graph, as is done in the next lemma.

\begin{lemma}
	\label{lem_restricted_flow}
	If $(G, I, O)$ is an open graph with flow $(f, \prec_f)$ and if $(G', I, O')$ is the reduced open graph described in Definition \ref{def_reduced_og}, then $(f', \prec_{f'})$, where
	\begin{itemize}
		\item $\forall i \in O'^C \quad f'(i) = f(i)$
		\item $\prec_{f'} = \prec_f \setminus [(V \times R) \cup (V_1^{\prec_s} \times O)]$
	\end{itemize}
	is a flow of $(G', I, O')$.
\end{lemma}
\begin{proof}
	It is sufficient to show that $(f', \prec_{f'})$ satisfies the flow properties (F1) - (F3) in Definition \ref{def_flow} and that $f'$ is a function from $O'^C$ to $I^C$.
	It is easy to see that $f': O'^C \rightarrow V' \setminus I$.
	$f'$ acts by definition on $O'^C$ and
	\begin{align*}
		\forall i \in O'^C \quad f'(i) = f(i) \in V \setminus I.
	\end{align*}
	The graph $G'$ has fewer vertices than $G$, therefore we need to show that all the vertices required according to the flow function $f'$ are included in $G'$, \emph{i.e.} $\forall i \in O'^C$ it must hold that $f'(i) \in V'$.
	According to Definition \ref{def_reduced_og} every vertex $v$ removed from the initial open graph $(G, I, O)$ is chosen such that $f^{-1}(v) \in O' $.
	Therefore it must be that every vertex $j \in V'$ such that $f(j) \notin V'$ must be an output vertex in $(G', I, O'^C)$.
	Because $f'(j)$ is not defined for outputs the vertices removed from the original graph $G$ are not needed for $f'$ and $f': O'^C \rightarrow V' \setminus I$. Hence we have that $f(i) = f'(i) \in I^C$ for every vertex $i \in O'^C$ and $f': O'^C \rightarrow I^C$.
	
	Let $R$ be the set of removed vertices as defined in Definition \ref{def_reduced_og}.
	The flow property (F1) states that $i \prec_{f'} f'(i)$ and holds because:
	\begin{align*}
		&\forall i \in O'^C \subseteq O^C \quad i \prec f(i) = f'(i) \quad \Rightarrow \\
		\Rightarrow \quad &f'(i) \notin R \quad \land \quad (i, f'(i)) \in \; \prec_f \quad \Rightarrow \\
		\Rightarrow \quad &(i, f'(i)) \in \; \prec_f \setminus V \times R = \; \prec_{f'} \quad \Rightarrow\\
		\Rightarrow \quad & i \prec_{f'} f'(i)
	\end{align*}
	
	To prove that $(f', \prec_{f'})$ satisfied flow property (F2) we need to show that for every $j \in V'$ if $j \in N(f'(i))$ then either $j = i$ or $i \prec_{f'} j$.
	\begin{align*}
		&j \in N(f'(i)) \quad \Rightarrow \quad j \in N(f(i)) \quad \Rightarrow\\
		\Rightarrow \quad &j = i \quad \lor \quad i \prec_f j \quad \Rightarrow \\
		\Rightarrow \quad &j = i \quad \lor \quad (i, j) \in \prec_f \quad \land \quad j \in V' = V \setminus R \quad \Rightarrow \\
		\Rightarrow \quad &j = i \quad \lor \quad (i, j) \in \prec_f \setminus V \times R = \; \prec_{f'} \quad \Rightarrow \\
		\Rightarrow \quad &j = i \quad \lor \quad i \prec_{f'} j
	\end{align*}
	
	Finally the flow property (F3) $i \in N(f'(i))$ holds almost trivially:
	\begin{align*}
		i \in O'^C \subseteq O^C \quad \Rightarrow \quad i \in N(f(i)) = N(f'(i))
	\end{align*}
\end{proof}

Next we prove that the reduced gflow of an SSF is also an SFF.

\begin{lemma}[\textbf{Constructing the reduced SSF}]
	\label{lem_restricted_ssf}
	Let $(G, I, O)$ be an open graph with flow $(f, \prec_f)$ and SSF $(s, \prec_s)$.
	If $(G', I, O')$ is the reduced open graph, according to $(s, \prec_s)$, then there exists an SSF $(s', \prec_{s'})$ of $(G', I, O')$ such that $(s', \prec_{s'})$ is the reduced gflow of $(s, \prec_s)$.
\end{lemma}
\begin{proof}
	Let $R$ be the set of vertices removed from $(G, I, O)$ to get $(G', I, O')$. The reduced flow $(f', \prec_{f'})$ exists because of Lemma \ref{lem_restricted_flow}. Define $(s', \prec_{s'})$ to be the the SSF derived from this reduced flow. Assume $(s', \prec_{s'})$ is not a reduced gflow of $(s, \prec_s)$, then one of the properties of Lemma \ref{lem_reduced_gflow_existance} should not hold, We show a contradiction in both cases. 
	
	If the first property does not hold then
	\begin{align*}
		&&\exists i \in O'^C \quad s.t. \quad s'(i) \cap O'^C \neq s(i) \cap O'^C \Rightarrow \\ 
		&&\exists j \in O'^C \cap [(s(i) \setminus s'(i)) \cup (s'(i) \setminus s(i))] \Rightarrow \\ 
		&& (parity(\zeta_i^{s}(f^{-1}(j))) = 1 \land parity(\zeta_i^{s'}(f^{-1}(j))) = 0) \lor \\ 
		&&(parity(\zeta_i^{s'}(f^{-1}(j))) = 1 \land parity(\zeta_i^{g}(f^{-1}(j))) = 0) \Rightarrow \\ && parity(\zeta_i^{s}(j)) \neq parity(\zeta_i^{s'}(j))
	\end{align*}
	Hence by removing vertices and edges from the open graph $(G, I, O)$ we must have changed $\zeta_i^s(f^{-1}(j))$ by an odd number to get $\zeta_i^{s'}(f^{-1}(j))$.

	We look at how removing the vertices in $R$ from the open graph $(G, I, O)$ changes $\zeta_i^s(f^{-1}(j))$.
	Removing a vertex $v$ changes the number of $Z$-paths from $i$ to $f^{-1}(j)$ if by removing it we also remove an edge in the $Z$-correction graph $G_Z$.
	Let this removed edge be $(k, l)$, then Corollary \ref{cor_zcorr} implies that $l \in N(f(k))$ and $v$ has to be either $k$, $l$ or $f(k)$.
	Corollary \ref{cor_zcorr} also implies that for $v$ to have an outgoing edge in $G_Z$, $f(v)$ has to be defined. Since $f$ is not defined for output vertices and $v \in R \subseteq O$, there cannot be any outgoing edges from $v$.
	Therefore $v$ cannot be $k$ as there is an edge from $k$ to $l$ in $G_Z$.
	Also $v$ cannot be $l$ since again $v$ cannot have an outgoing edge in $G_Z$, hence $v$ would have to be the last element on the $Z$-path from $i$ to $f^{-1}(j)$, which is $f^{-1}(j)$.
	This is impossible, as $f^{-1}(j)$ cannot be an output vertex.
	Therefore the only possibility is that $v = f(k)$.
	
	Let $v$ be the first vertex removed from $G$, such that $\zeta_i^s(f^{-1}(j))$ changes by an odd number.
	Hence all the paths from $i$ to $f^{-1}(j)$ that disappear due to removal of $v$ have to go through $f^{-1}(v)$. Therefore there must also exist an odd number of paths from $f^{-1}(v)$ to $f^{-1}(j)$.
	We know that because of Proposition \ref{prop_ssf}, $j \in s(f^{-1}(v))$ and $f^{-1}(v) \prec_s j$.
	On the other hand because of Definition \ref{def_reduced_og} it must also hold that $f^{-1}(v) \in V_1^{\prec_s}$, which together with Definition \ref{def_gflow_depth} implies that $j \in V_0^{\prec_s} = O$.
	This leads to a contradiction, because $j$ has to be in $O'^C \subseteq O^C$ and cannot be in $O$.
	Therefore property $(1)$ must be true for $(s', \prec_{s'})$.
	
	Now we show that property $(2)$ has to hold.
	According to Lemma \ref{lem_evenconnections}:
	\begin{align*}
		Odd(s'(i)) \cap O'^C = \lbrace i \rbrace = Odd(s(i)) \cap O^C
	\end{align*}
	Because $i \in O'^C \subseteq O^C$, it must also hold that
	\begin{align*}
		Odd(s'(i)) \cap O'^C = \lbrace i \rbrace = Odd(s(i)) \cap O'^C
	\end{align*}
	and property $(2)$ has to be true for $(s', \prec_{s'})$.
\end{proof}

In the specific case where $|I|=|O|$, it will follow from Lemma \ref{lem_restricted_ssf} that the unique SSF of a reduced open graph is the reduced gflow of the original SSF.

\begin{corollary}
	\label{cor_restricted_ssf}
	Let $(G, I, O)$ be an open graph with an SSF $(s, \prec_s)$ such that $|I|=|O|$.
	If $(G', I, O')$ is the reduced open graph, according to $(s, \prec_s)$, of $(G, I, O)$, then the unique SSF $(s', \prec_{s'})$ of $(G', I, O')$ has the following properties:
	\begin{enumerate}
		\item $\forall i \in O'^C \quad s'(i) \cap O'^C = s(i) \cap O'^C$
		\item $\forall i \in O'^C \quad Odd(s'(i)) \cap O'^C = Odd(s(i)) \cap O'^C$
	\end{enumerate}
\end{corollary}
\begin{proof}
	Because of Lemma \ref{lem_restricted_flow} $(G', I, O')$ has a flow.
	Since $|I| = |O|$, then according to Lemma \ref{lem_reduced_og_output} $|I| = |O'|$ and hence $(G'. I, O')$ has a unique flow \cite{Beaudrap06a}.
	Flow is required for the existence of an SSF according to Proposition \ref{prop_ssf}, therefore there can exist only one SSF and because of Lemma \ref{lem_strict_subset} this SSF has to satisfy properties (1) and (2).
\end{proof}

\subsubsection{Moving back}
\label{sec_optimality_induction}

We have proven that the penultimate layers of SSF and optimal gflow are equal if $|I|=|O|$.
Then we showed how to remove some vertices from the open graph and construct an SSF and optimal gflow on the new reduced graph.
Both of them are reduced gflows, a property which we will use in this section to show that they preserve the layering of the gflows they were derived from.

\begin{lemma}
	\label{lem_layering}
	Let $(G, I, O)$ be an open graph with SSF $(s, \prec_s)$ and gflow $(g, \prec_g)$ such that $(G', I, O')$ is the reduced open graph of $(G, I, O)$ with the removed vertices set $R$. For every reduced gflow $(g', \prec_{g'})$ of  $(G', I, O')$ such that $V_0^{\prec_s} = V_0^{\prec_g} = O$ and $V_1^{\prec_s} = V_1^{\prec_g}$
	it must hold that
	\begin{align}
		\label{eq_layering}
		\forall n \ge 0 \quad
		\cup_{k = 0}^n V_k^{\prec_{g'}} =
		\cup_{k = 0}^{n+1} V_k^{\prec_{g}} \setminus R
	\end{align}
\end{lemma}
\begin{proof}
	We prove Lemma \ref{lem_layering} by induction and showing first that Equation \ref{eq_layering} holds if $n = 0$, \emph{i.e.} we need to prove that
	\begin{align*}
		V_0^{\prec_{g'}} = (V_1^{\prec_g} \cup V_0^{\prec_g}) \setminus R
	\end{align*}
	Lemma \ref{lem_last_layer} tells that $V_0^{\prec_{g'}} = O'$ and $V_0^{\prec_g} = O$.
	Because the penultimate layers of SSF $(s, \prec_s)$ and gflow $(g, \prec_g)$ are equal we also have that $V_1^{\prec_s} = V_1^{\prec_g}$.
	Now we take the definition of $O'$ from Definition \ref{def_reduced_og} of the reduced open graph and substitute the appropriate sets:
	\begin{align*}
		O' = (V_1^{\prec_s} \cup O) \setminus R \quad
		\Rightarrow
		V_0^{\prec_{g'}} = (V_1^{\prec_g} \cup V_0^{\prec_g}) \setminus R
	\end{align*}
	Thus Equation \ref{eq_layering} holds for $n = 0$.
	For the induction step we assume that Equation \ref{eq_layering} holds for $n = m-1$, \emph{i.e}
	\begin{align}
		\cup_{k = 0}^{m-1} V_k^{\prec_{g'}} =
		\cup_{k = 0}^{m} V_k^{\prec_{g}} \setminus R
	\end{align}
	and show that it holds for $n = m$.
	We use contradiction and assume that
	\begin{align*}
		\cup_{k = 0}^m V_k^{\prec_{g'}} \neq
		\cup_{k = 0}^{m+1} V_k^{\prec_{g}} \setminus R
	\end{align*}
	There are two possibilities: either $\exists i \in V^{\prec_{g'}}_m \setminus  V^{\prec_g}_{m+1}$ or $\exists i \in V^{\prec_{g}}_{m+1} \setminus  V^{\prec_g'}_{m}$.
	We note that according to Lemma \ref{lem_ordering} $i \prec_g j \Leftrightarrow i \prec_{g'} j$ if $i \in O'^C$ and $j \in O'^C$.
	Because $V^{\prec_{g'}}_m \subseteq O'^C$ for every $m > 0$ we have that
	\begin{align*}
		\exists i \in V^{\prec_{g'}}_m \setminus  V^{\prec_g}_{m+1}
		&\Rightarrow \quad
		\exists j \in V^{\prec_g}_{m+1} \cap g'(i) \quad s.t. \quad i \prec_g j
		\quad \Rightarrow \\
		&\Rightarrow \quad
		i \prec_{g'} j 
		\quad \Rightarrow \quad
		j \in \cup_{k = 0}^{m-1} \; V^{\prec_{g'}}_k = \cup_{k = 0}^m \; V^{\prec_{g}}_k \\
		\exists i \in V^{\prec_{g}}_{m+1} \setminus  V^{\prec_{g'}}_{m}
		\quad &\Rightarrow \quad
		\exists j \in V^{\prec_{g'}}_{m} \cap g(i) \quad s.t. \quad i \prec_{g'} j
		\quad \Rightarrow \\
		&\Rightarrow \quad
		i \prec_{g} j 
		\quad \Rightarrow \quad
		j \in \cup_{k = 0}^m \; V^{\prec_{g}}_k = \cup_{k = 0}^{m-1} \; V^{\prec_{g'}}_k
	\end{align*}
	As can be seen above, both of the possible cases lead to a contradiction and hence it must hold that 
	\begin{align*}
		\cup_{k = 0}^m V_k^{\prec_{g'}} =
		\cup_{k = 0}^{m+1} V_k^{\prec_{g}} \setminus R
	\end{align*}
	This completes the induction step and the proof itself.
\end{proof}

From the previous lemma we can construct a proof saying that every layer of a reduced gflow starting from the second to last one is equal to a layer of the original gflow.

\begin{corollary}
	\label{cor_layering}
	Let $(G, I, O)$ be an open graph with SSF $(s, \prec_s)$ and gflow $(g, \prec_g)$ such that $(G', I, O')$ is the reduced open graph of $(G, I, O)$ with the removed vertices set $R$. For every reduced gflow $(g', \prec_{g'})$ of  $(G', I, O')$ such that $V_0^{\prec_s} = V_0^{\prec_g} = O$ and $V_1^{\prec_s} = V_1^{\prec_g}$ it must hold that
	\begin{align*}
		\forall n > 0 \quad V^{\prec_{g'}}_n = V^{\prec_g}_{n+1} 
	\end{align*}
\end{corollary}
\begin{proof}
	This follows trivially from Lemma \ref{lem_layering}. We have that if $n > 0$:
	\begin{align*}
		\cup_{k = 0}^n V_k^{\prec_{g'}} =
		\cup_{k = 0}^{n+1} V_k^{\prec_{g}} \setminus R \\
		\cup_{k = 0}^{n-1} V_k^{\prec_{g'}} =
		\cup_{k = 0}^{n} V_k^{\prec_{g}} \setminus R
	\end{align*}
	We can now subtract the elements of the second set from the first.
	\begin{align*}
		&\cup_{k = 0}^n V_k^{\prec_{g'}} \setminus 
		\cup_{k = 0}^{n-1} V_k^{\prec_{g'}} =
		(\cup_{k = 0}^{n+1} V_k^{\prec_{g}} \setminus R) \setminus
		(\cup_{k = 0}^{n} V_k^{\prec_{g}} \setminus R)
		\quad \Rightarrow \quad
		V_n^{\prec_{g'}} = V_{n+1}^{\prec_{g}} \setminus R
	\end{align*}
	Because Definition \ref{def_reduced_og} of SSF reduced open graph we know that $R \subseteq O$.
	Lemma \ref{lem_last_layer} says that $O = V_0^{\prec_g}$, hence we know that no element in $R$ can be included in $V_{n+1}^{\prec_g}$ for $n \ge 0$ and $V^{\prec_{g'}}_n = V^{\prec_g}_{n+1}$.
\end{proof}

It turns out, that if the gflow we have for the original open graph is the optimal one, then the reduced gflow will be optimal for the reduced open graph.

\begin{lemma}[\textbf{Constructing the optimal reduced gflow}]
	\label{lem_reduced_optimal_gflow}
	Let $(G, I, O)$ be an open graph with SSF $(s, \prec_s)$ and optimal gflow $(g, \prec_g)$. If $(G', I, O')$ is the reduced open graph of $(G, I, O)$ then the reduced gflow $(g', \prec_{g'})$ of $(g, \prec_g)$ is the optimal gflow of $(G', I, O')$.
\end{lemma}
\begin{proof}
	First, because of Lemma \ref{lem_ordering} $(g', \prec{g'})$ has to be a gflow of $(G', I, O')$.
	Let us assume that $(g', \prec_{g'})$ is not the optimal gflow of $(G', I, O')$ and let $(d, \prec_d)$ be the optimal one.
	Then according to Definition \ref{def_delayed_gflow}
	\begin{align*}
		\exists n > 0, i \in O'^C \quad s.t. \quad
		i \in V_n^{\prec_d} \setminus V_n^{\prec_{g'}} \quad \land \quad
		\forall k < n \quad V_k^{\prec_{g'}} = V_k^{\prec_d}
	\end{align*}
	and since $d(i) \in \cup_{k = 0}^{n-1} V_k^{\prec_d} = \cup_{k = 0}^{n-1} V_k^{\prec_{g'}} $ from Lemma \ref{lem_layering} we obtain $d(i) \in \cup_{k = 0}^{n} V_k^{\prec_g}\setminus R$, where the $R$ is the set of vertices removed from the original graph.
	Now we know from Definition \ref{def_delayed_gflow} that $i$ is in $V_{n+1}^{\prec_g}$ which according to Corollary \ref{cor_layering} must be equal to $V_n^{\prec_{g'}}$.
	This leads to a contradiction because $i \in V_n^{\prec_d} \setminus V_n^{\prec_{g'}}$ and thus $(g', \prec_{g'})$ has to be the optimal gflow of $(G', I, O')$.
\end{proof}

\subsubsection{Proof of the optimality theorem}
\label{sec_optimality_proof}

We can now prove Theorem \ref{thm_optimality} by showing that the vertex layering of any SSF and an optimal gflow is exactly the same.
Let $(G, I, O)$ be an open graph with flow $(f, \prec_f)$ such that $|I| = |O|$.
Let $(s, \prec_s)$ be the SSF obtained from $(f, \prec_f)$ according to Proposition \ref{prop_ssf} and $(g, \prec_g)$ the optimal gflow of $(G, I, O)$.
According to Proposition \ref{prop_ssf} the last layer of any SSF is the set of output vertices.
Lemma \ref{lem_last_layer} says that this is also true for the last layer of an optimal gflow, therefore $V^{\prec_s}_0 = V^{\prec_g}_0 = O$.
The layers $V^{\prec_s}_1$ and $V^{\prec_g}_1$ are equal because of Lemma \ref{lem_equal_layer}.

Now we need to show that layers $V^{\prec_s}_n$ and $V^{\prec_g}_n$ are equal for $n > 1$.
We can construct a reduced open graph (Definition \ref{def_reduced_og}) $(G', I, O')$ from $(G, I, O)$.
We now consider the unique SSF $(s', \prec_{s'})$ and reduced gflow $(g', \prec_{g'})$ of $(G', I, O')$, which according to Lemma \ref{lem_reduced_optimal_gflow} is optimal.
According to Lemma \ref{lem_layering}, $V^{\prec_{g'}}_n =  V^{\prec_g}_{n+1}$ for every $n > 0$ and because SSF is by Theorem \ref{theorem_ssfisgflow} a gflow the same lemma also implies that $V^{\prec_{s'}}_n = V^{\prec_s}_{n+1}$.
Because of the way a reduced open graph is defined, we know that $|I| = |O'|$ (Lemma \ref{lem_reduced_og_output}).
Thus we can again use Lemma \ref{lem_equal_layer} to say that $V^{\prec_{g}}_2 = V^{\prec_{g'}}_1 = V^{\prec_{s'}}_1 = V^{\prec_{s}}_2$.
We can now take $(G', I, O')$ and find its reduced open graph to show using the same technique that $V^{\prec_{g}}_3 = V^{\prec_{s}}_3$.
This can be continued until we reach the empty layers, in which case we have considered all the layers according to $(s, \prec_s)$ and $(g, \prec_g)$.
As every layer of $(s, \prec_s)$ and $(g, \prec_g)$ will be equal and $(g, \prec_g)$ is the optimal gflow, the SSF of a flow of an open graph $(G, I, O)$ is an optimal gflow if $|I|=|O|$, which proves Theorem \ref{thm_optimality}.

\section{Compact circuits from signal shifted flow} \label{sec_compactification}

In the last section we presented an automated parallelisation technique for measurement patterns. However, when we translate those parallel measurement patterns back to the circuit model using the method described in \cite{BroadbentK09}, we end up with quantum circuits with many ancilla qubits. More specifically, the new circuits will have the same number of qubits as there are vertices in the associated MBQC graph. Our next main result of the paper is a new scheme that explore the notion of circuit compactification introduced in \cite{daSilvaG12} to remove all ancilla qubits introduced by the back-and-forth translation between the two models. We start by reviewing the notion of extended circuits, which is basically a re-interpretation of measurement patterns using circuit notation. Then we derive a set of rewrite procedures that combine the rewrite rules introduced in \cite{daSilvaG12} in such a way that the SSF layering function (and, consequently, the optimised depth) does not change in the process of removing ancilla qubits from an SSF extended circuit. Finally, we introduce the algorithm that make use of rewrite procedures to completely rewrite an extended circuit until all ancilla qubits are removed.

\subsection{Extended translation} 
\label{sec_ext}

A straightforward translation method for measurement patterns, which we refer to as \emph{extended translation}, was introduced in \cite{BroadbentK09}. This translation is inefficient, in the sense that it gives as many circuit wires as vertices in the original pattern (instead of inputs only). However its importance comes from the fact that its very easy to implement since the procedure to obtain an extended circuit is just a re-interpretation of the measurement pattern using the quantum circuit notation. Moreover, it will serve as a starting point to obtain more compact circuits for patterns with signal-shifted flow. 

\begin{figure}
 \center
\includegraphics[scale=0.66]{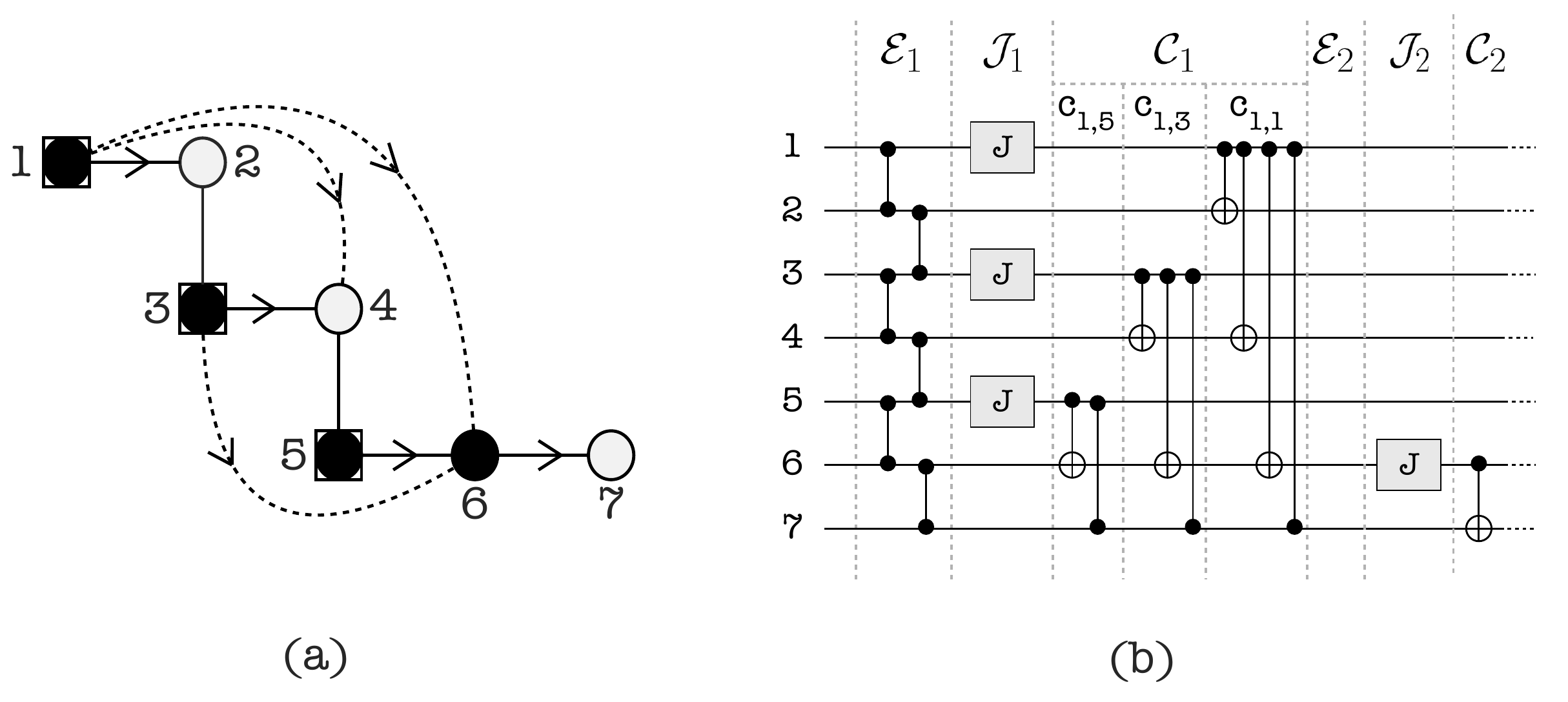}
\caption{Graph associated to the measurement pattern in Equation \ref{eq_examplessf} and the corresponding extended circuit.}
\label{fig_examplessf}
\end{figure}

\begin{definition} \label{def_ext} Given a signal shifted measurement pattern with computational space~$(V,I,O)$ and underlying geometry~$(G,I,O)$ with SSF $(s, \prec_s)$. The corresponding \emph{extended circuit}~$C$ with~$|I|$ input qubits and~$|V\setminus I|$ ancilla qubits, is constructed in the
following steps:

\begin{enumerate}
\item Each vertex on the graph is translated as a circuit wire.
\item The wires corresponding to $I^C$ are prepared in the sate $\ket +$.
\item Each edge linking vertices $i$ and $j$ on the graph (command $E_{ij}$) is translated as a $CZ_{ij}$ in the beginning of the circuit.
\item Each dependent measurement $[M_{i}^{\theta_i}]^{s(i)}$ is translated as a gate $J({- \theta_i})$ in wire $i$ followed by Controlled-$X$ operator with qubit $i$ as control and $j\in s(i)$ as the target. The layering respects the $\prec_s$ by replacing the control just after the $J$-gate and the target just before the $J$-gate of the next measurement command.
\item Each correction on the output qubits $C_j^{s_i}$ ($C=$ Pauli -X or -Z)  is translated as a Controlled-$C$ gates at with qubit $i$ as control and $j$ as target. The layering again respects $\prec_s$ by putting the control right after the $J$-gate and all the corresponding $CX$ (introduced in Step 4) acting on qubit $i$.
\item All the qubits in $O^C$ will be measured in the computational basis.
\end{enumerate}
The obtained layering structure is referred to as $\mathcal{E}_n$, $\mathcal{J}_n$, and $\mathcal{C}_n$, each containing only entangling gates, $J$-gates, and correction gates, respectively. We also divide slices $\mathcal{C}_n$ into $s_n$ many smaller slices $c_{n,i}$, where $s_n$ is the total number of $J$-gates in slice $\mathcal{J}_n$. Each slice $c_{n,i}$ contains all correction gates with control on qubit $i$ s.t. $J_i$ is in $\mathcal{J}_n$.
\end{definition}
 
It is easy to verify that the above circuit implements the same operator as the measurement pattern (see also \cite{BroadbentK09}). For clarity, in what follows we will refer to a $CZ_{ij}$ created in Step 3 above as $E_{ij}$ while keeping the notation $CZ_{ij}$ for those created in Step 5. Later on we will use the fact that, for a signal shifted pattern (Equation \ref{eq_ssf_pattern}), a $CX_{ij}$ or $CZ_{ik}$ will be created in the corresponding extended circuit if and only if $j \in s(i)$ or $k \in Odd(s(i))$, respectively. Note that by construction, all gates associated to operators $E_{ij}$ are initially in slice $\mathcal{E}_1$ (Step 3 in Definition \ref{def_ext}), with $\mathcal{E}_2, ..., \mathcal{E}_n$ all empty. However during the compactification procedure while we rewrite the  circuit new gates will be added to these empty slices. Figure \ref{fig_examplessf}-b shows the extended circuit of the following signal shifted measurement pattern with associated graph given in Figure \ref{fig_examplessf}-a.
\begin{equation} \label{eq_examplessf}
Z^{s_1+s_3+s_5}_7X_7^{s_6}M_{6}^{\theta_6}X_6^{s_1+s_3+s_5}M_{5}^{\theta_5}X_4^{s_1+s_3}M_{3}^{\theta_3}X_2^{s_1}M_{1}^{\theta_1}E_{67}E_{56}E_{45}E_{34}E_{23}E_{12}.
\end{equation} 

\subsection{Compactification procedures} 
\label{sec_rps}

Compactification procedures can be described as a way of globally rewriting extended circuits in order to remove ancilla (non-input) wires. One way to achieve this is to rewrite the circuit to create \textit{$J$-blocks}, defined as follows. 

\begin{definition}\label{def_jblock} Consider a measurement pattern with computational space~$(V,I,O)$ and underlying geometry~$(G,I,O)$ with flow $(f, \prec_f)$ and corresponding extended circuit $C$. We say there is a \emph{$J$-block} in wires $i$ and $f(i)$ if the following set of conditions are satisfied (see Figure \ref{fig_jgate}-a):
\begin{enumerate}
\item The initial state of wire $f(i)$ is $|+\rangle$.
\item The gates sequence ($E_{if(i)}$, $J_i(\theta_i)$, $CX_{if(i)}$) appears in $C$.
\item The only gate acting on the wire $f(i)$ before $CX_{if(i)}$ is $E_{if(i)}$.
\item The only gates acting on wire $i$ after $E_{if(i)}$ are $J_i(\theta_i)$ and $CX_{if(i)}$. 
\item After $CX_{if(i)}$ gate, the qubit $i$ is measured in the $Z$ basis.
\end{enumerate}
\end{definition}

Once a $J$-block is created (via circuit rewriting), one can use the identity in Figure \ref{fig_jgate} ($J$-gate identity) to remove one wire from the circuit. In general, extended circuits are not prepared for direct applications of the $J$-gate identity. In Definition \ref{def_ext}, all corrections $C_j^{s_i}$ are translated as controlled-gates with control placed after $J_i$ gate. Hence Condition 4 in Definition \ref{def_jblock} is not satisfied in general. Moreover, since all $E$ gates are initially placed in slice $\mathcal{E}_1$, Condition 3 is not satisfied either for any open graph with more than two non-output qubits (see example in Figure \ref{fig_examplessf}-b). In order to create $J$-blocks in SSF extended circuits, we explore the relation between the $E$ gates and the correcting gates $C$, since the latter are defined accordingly to the former through the stabilizer formalism. In other words, there is a direct relation between the gates in slice $\mathcal{E}_1$ and all other two-qubit gates in the rest of the extended circuit. In the case where we succeed in removing as many wires as there are non-output qubits in the graph, we say that the resulting circuit is in a \emph{compact form}. We will refer to circuits in the compact form as \textit{compact circuits}.

\begin{figure}[ht]
\center
\includegraphics[scale=1]{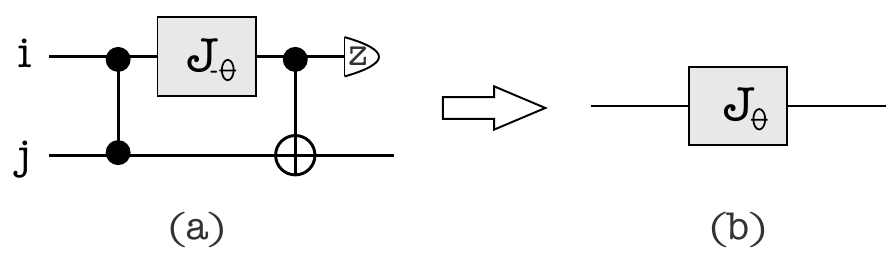}
\caption{This \textit{$J$-gate identity} will be used repeatedly to simplify generic extended circuits. Note that the $J$-gate angles differ from each other by a minus sign.}
\label{fig_jgate}
\end{figure}

\begin{definition} \label{def_compactform} Let~$C$ be the extended circuit of a measurement pattern with computational space $(V,I,O)$.
We say that $C$ can be put into a compact form if there exists a sequence of circuit rewriting equations such that the $J$-gate identity (Figure \ref{fig_jgate}) can be applied $|O^C|$ times. \end{definition}

\begin{figure}
\center
\includegraphics[scale=.7]{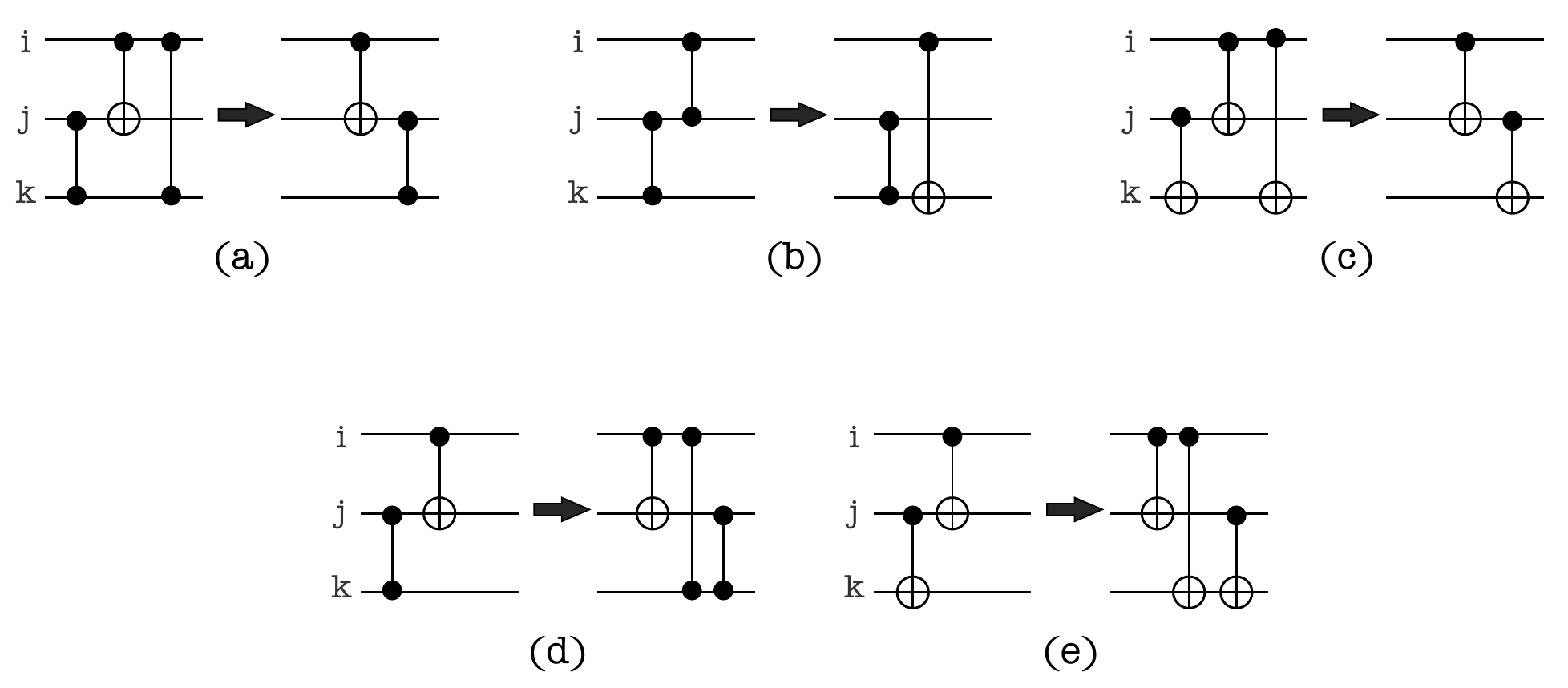}
\caption{Circuit identities in figures (a) to (c) are the rewrite rules from \cite{daSilvaG12} (the circuit identity in Figure (b) is true only if qubit $k$ is in the $|+\rangle$ state). The circuit identity in Figure (d) (Figure (e)) is obtained by multiplying $CZ_{ik}$ ($CX_{ik}$) in both sides of the identity in Figure (a) (Figure (c)).}
\label{fig_rules}
\end{figure}

A collection of circuit identities with the purpose of exploring the aforementioned relation between gates in extended circuits to create $J$-blocks was introduced in \cite{daSilvaG12} (Figure \ref{fig_rules}). A compactification procedure for graphs with flow was provided and some simple examples of graphs with gflow explored. In this section we present a novel algorithm able to rewrite SSF extended circuits to put it into a compact form. In what follows we present the circuit identities that will be used in the algorithm, which we will refer to as Rewrite Procedures (RPs). We refer to the $i$ wire of each RP as the \emph{target wire},  $j_1, ..., j_n$ as the \emph{correcting wires} and finally $k$ as the \emph{neighbour wire}. Moreover, when we need to emphasize which RP we are referring to we also add a superscript to the wire label; for instance $k^{(2)}$ indicates the neighbour wire of RP2. 

\begin{itemize}
\item \textbf{Rewrite Procedure 1.}  The circuit identity in Figure \ref{fig_rp1} moves gates $(E_{kj_1}, ..., E_{kj_m})$ past gates $(CX_{ij_1}, ..., CX_{ij_m})$, adding $m$ many gates $CZ_{ik}$ to slice $\mathcal{C}$ in the process. Using the rewrite rule in Figure \ref{fig_rules}-d each of those $E$ gates can be moved past gates $(CX_{ij_1}, ..., CX_{ij_m})$ in $\mathcal{C}$, creating a new $CZ_{ik}$ each time the rule is applied (Figure \ref{fig_rp1}-c). \\

\begin{figure}[ht]
\center
\includegraphics[scale=.56]{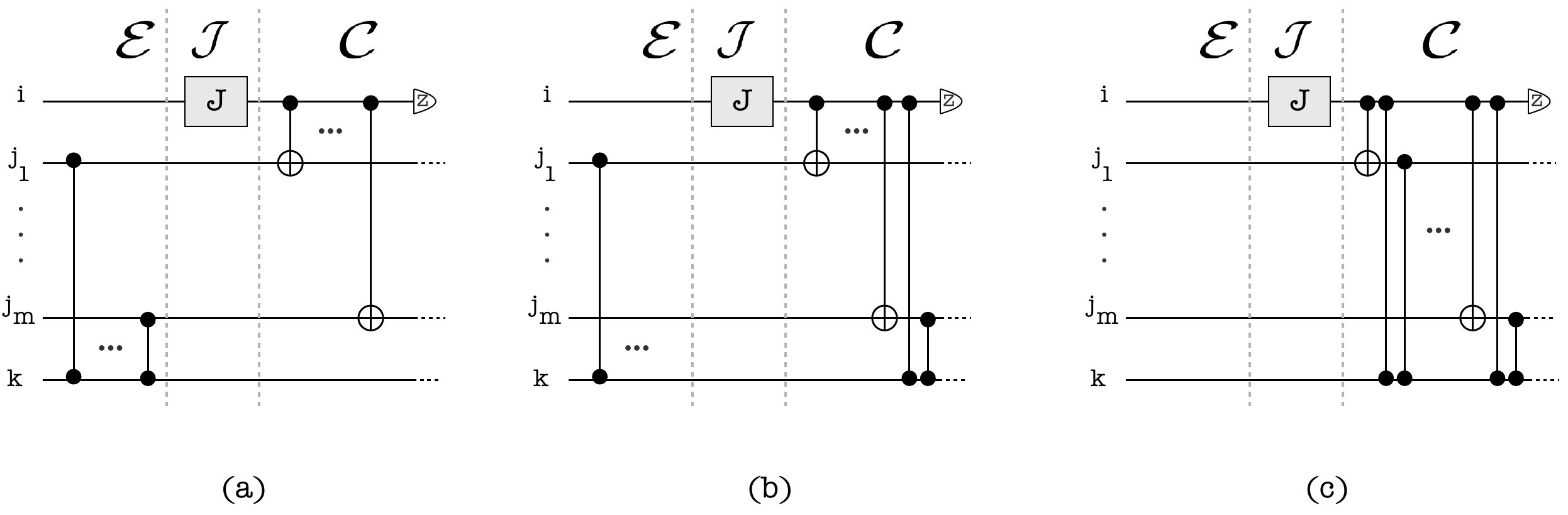}
\caption{Rewrite Procedure 1. This RP is basically several applications of the rewrite rule in Figure \ref{fig_rules}-d. Although all $E$ gates are drawn in the $\mathcal{E}$ slice, it is sufficient that those $E$ gates are placed just before the corresponding $CX$ gates (that is, with no gate in between)).}
\label{fig_rp1}
\end{figure}

\item \textbf{Rewrite Procedure 2.} The circuit identity in Figure \ref{fig_RP2} replaces gates $(E_{kj_1}, ...,E_{kj_{n-1}})$ in slice $\mathcal{E}$ with $(CX_{j_1j_n}, CX_{j_2j_n},$ $ ...,CX_{j_{n-1}j_n})$ in slice $\mathcal{C}$, removing gate $CX_{ij_n}$ in the process. We use the rewrite rule in Figure \ref{fig_rules}-b for each pair $\{(E_{j_1k}E_{kj_n}), ..., (E_{j_{n-1}k}E_{kj_n})\}$ transforming gates $(E_{kj_1}, ..., E_{kj_{n-1}})$ into $(CX_{j_1j_n}, ..., CX_{j_{n-1}j_n})$, as depicted in Figure \ref{fig_RP2}-b. The new $CX$ gates can be pushed forward to the beginning of slice $\mathcal{C}$, since it commutes trivially with $CX_{kj_n}$. Using the rewrite rule in Figure \ref{fig_rules}-e we can commute $(CX_{j_1j_n}, ..., CX_{j_{n-1}j_n})$ past $(CX_{ij_1}, ..., CX_{ij_{n-1}})$ creating $(n-1)$ many new $CX_{ij_n}$ in the process, which together with the pre-existing $CX_{ij_n}$ in $\mathcal{C}$ will cancel out, resulting in the circuit depicted in Figure \ref{fig_RP2}-c.\\

\begin{figure}[ht]
\center
\includegraphics[scale=.56]{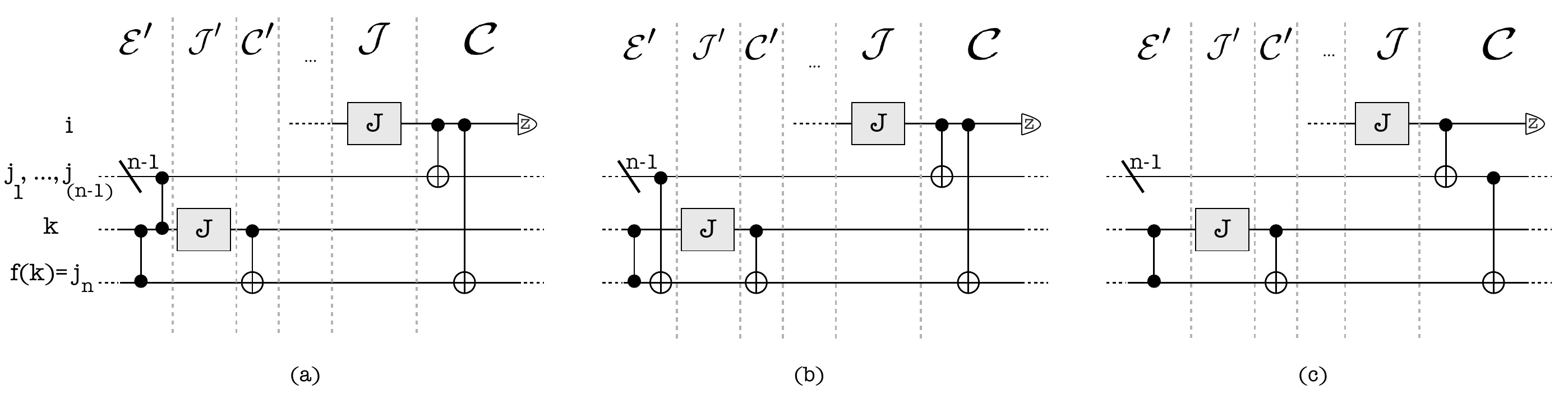}
\caption{Rewrite Procedure 2. If $L_s(k)=L_s(i)$, slices $(\mathcal{E},\mathcal{J},\mathcal{C})$ and $(\mathcal{E'},\mathcal{J'},\mathcal{C'})$ become the same; The rewrite procedure remains exactly the same.} 
\label{fig_RP2}
\end{figure}

\item \textbf{Rewrite Procedure 3.} The circuit identity in Figure \ref{fig_RP3} replaces gates $(E_{kj_1}, ...,E_{kj_n})$ in slice $\mathcal{E}'$ with $(CX_{j_1f(k)}, ...,$ $CX_{j_nf(k)})$ in slice $\mathcal{C}$. We use the rewrite rule in Figure \ref{fig_rules}-b for each pair $\{(E_{j_1k}E_{kf(k)}), ...,$ $ (E_{j_nk}E_{kf(k)})\}$ transforming gates $(E_{j_1k}, ..., E_{j_nk})$ into $(CX_{j_1f(k)}, ..., CX_{j_nf(k)})$ (Figure \ref{fig_RP3}-b). The new $CX$ gates can be pushed forward to the beginning of slice $\mathcal{C}$, since it commutes trivially with $CX_{kf(k)}$. Using the rewrite rule in Figure \ref{fig_rules}-e we can commute $(CX_{j_1f(k)}, ..., CX_{j_nf(k)})$ past $(CX_{ij_1}, ..., CX_{ij_n})$, creating $n$ many $CX_{if(k)}$ in the process. Since $n$ is even, all those $CX$ gates will cancel, resulting in the circuit depicted in Figure \ref{fig_RP3}-c.
\end{itemize}

\begin{figure}[ht]
\center
\includegraphics[scale=.6]{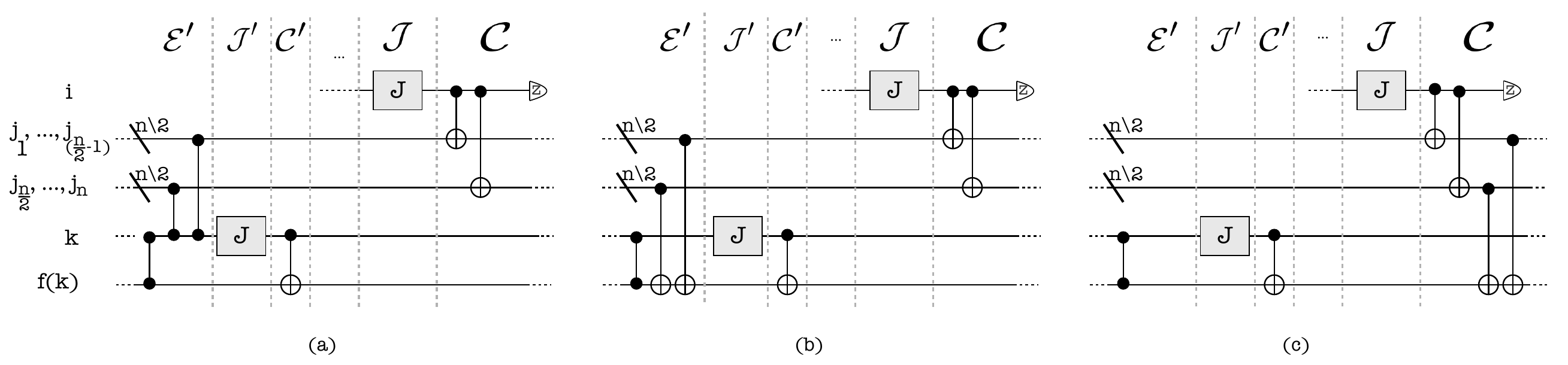}
\caption{Rewrite Procedure 3. If $L_s(k)=L_s(i)$, slices $(\mathcal{E},\mathcal{J},\mathcal{C})$ and $(\mathcal{E'},\mathcal{J'},\mathcal{C'})$ become the same; The rewrite procedure remains exactly the same.} 
\label{fig_RP3}
\end{figure}


In order to apply the $J$-gate identity for the pair of wires $(i,f(i))$ of a SSF extended circuit, we need to rewrite it until all conditions in Definition \ref{def_jblock} are satisfied. In a SSF extended circuit, the first two conditions are trivially satisfied for any pair of wires $(i,f(i))$ and, therefore, we need to rewrite the circuit to satisfy the other conditions. To do so we analyse each qubit $k$ in the neighbourhood of $s(i)$, classifying it according to three different cases: (i) $L_{s}(k) < L_{s}(i)$, (ii) $L_{s}(k) \geq L_{s}(i)$ and $f(k) \in s(i)$ and (iii) $L_{s}(k) \geq L_{s}(i)$ and $f(k) \notin s(i)$. This separation into cases is necessary for two reasons: First, the distinction between $L_{s}(k) < L_{s}(i)$ and $L_{s}(k) \geq L_{s}(i)$ is necessary because we are interested in keeping the $J$-gate parallelization introduced by signal shifting and hence we need a different procedure to deal with each case. Secondly, in the case where $L_{s}(k) \geq L_{s}(i)$, we use the rewrite rule in Figure \ref{fig_rules}-b which deletes $E$ gates. Since condition 2 in Definition \ref{def_jblock} requires the existence of gates of form $E_{kf(k)}$, we will treat differently cases where $f(k) \in s(i)$ and $f(k) \notin s(i)$ to guarantee those $E$ gates will not be removed from the circuit. As we show next, for each case one of the RPs can be applied if a set of prior conditions are satisfied.

\begin{proposition} \label{prop_rp1} Let $i$ be a wire in a SSF extended circuit s.t. $J_i$ is in some slice $\mathcal{J}_n$. If there exists a wire $k$ s.t. (i) $L_{s}(k) < L_{s}(i)$ and (ii) the set of gates $\{E_{kj_1}, ..., E_{kj_m}\}$ (with $j_1, ..., j_m \in s(i)$) can be pushed to slice $\mathcal{J}_n$, then RP1 (Figure \ref{fig_rp1}) can be applied.
\end{proposition}
\begin{proof}
Since $L_{s}(k) < L_{s}(i)$, the gate $J_k$ belongs to a future slice $\mathcal{J}_m$ ($m>n$). Also, $j_1, ..., j_m \in s(i)$ implies $i \prec_s \{j_1, ..., j_m\}$ and hence the gates $J_{j_1}, ..., J_{j_m}$ are in slices after $\mathcal{J}_n$ as well. Moreover, $j_1, ..., j_m \in s(i)$ implies there exist operators $X_{j_1}^{s_i}, ..., X_{j_m}^{s_i}$ in the measurement pattern, which are translated to the extended circuit as a $CX_{ij_1}, ..., CX_{ij_m}$ in slice $c_{n,i}$ (according to Definition \ref{def_ext}). Thus, if every gate $E_{kj}$, $j \in s(i)$, can be trivially pushed to slice $\mathcal{J}_n$, we have exactly the scenario depicted in Figure \ref{fig_rp1}-a. Therefore, the circuit identity in Figure \ref{fig_rp1} can be applied.
\end{proof}

For reasons that will become clear in the section \ref{sec_alg}, where the algorithm to obtain compact circuits from SSF extended circuits is introduced, we need to consider a case which is slightly different from the scenario described in Proposition \ref{prop_rp1}. In this case, condition (ii) in Proposition \ref{prop_rp1} is not satisfied because some of the $E$ gates are in slice $c_{n,i}$ but cannot be pushed trivially back to $\mathcal{J}_n$. The only scenario where that could happen is if there exists $CX_{i,k}$ in slice $c_{n,i}$ and some of the $E_{kj}$ gates are placed past it (and hence cannot be pushed trivially to $\mathcal{J}_n$). In this scenario, as we show next, RP1 can also be applied.

\begin{proposition} \label{prop_rp1b} Let $i$ be a wire in a SSF extended circuit s.t. $J_i$ is in some slice $\mathcal{J}_n$. If there exists a wire $k$ s.t. (i) $k \in s(i)$ and (ii) gates $\{E_{kj_1}, ..., E_{kj_m}\}$ (with $j_1, ..., j_m \in s(i)$) can either be all pushed to slice $\mathcal{J}_n$ or just some can be pushed  to slice $\mathcal{J}_n$ and the other $E$ gates are placed in $c_{n,i}$ just after $CX_{ik}$, then RP1 (Figure \ref{fig_rp1}) can be applied.
\end{proposition}
\begin{proof}
The proof of this proposition is trivial due to its similarity to Proposition \ref{prop_rp1}. First, note that $k \in s(i)$ implies $L_{s}(k) < L_{s}(i)$ and hence condition (i) in this proposition is equivalent to the one in Proposition \ref{prop_rp1}. Therefore, if all gates in the set $\{E_{kj_1}, ..., E_{kj_m}\}$ can be pushed to slice $\mathcal{J}_n$ we have exactly the conditions in Proposition \ref{prop_rp1} and there is nothing to prove. The other possibility is when a subset of the set $\{E_{kj_1}, ..., E_{kj_m}\}$ can be pushed to slice $\mathcal{J}_n$ but gates in the complementary subset are placed in slice $c_{n,i}$, after gate $CX_{ik}$ (which exists since $k \in s(i)$). It is easy to note that it will not prevent the application of RP1, since the only gate in $c_{n,i}$ that could be placed in between the $E$ and $CX$ gates used in RP1, namely $CX_{ik}$, is assumed in condition (ii) to be placed before the gates of the form $E_{kj}$.
\end{proof}

In the next Lemma we show the interesting effect of applying RP1 to SSF extended circuits.


\begin{lemma} \label{lem_rp1}
The application of RP1 to a pair of target and neighbour wires $(i,k)$ of a SSF extended circuit removes all $CZ_{ik}$ from the circuit.
\end{lemma}
\begin{proof}
According to Definition \ref{def_ext}, all $CZ$  gates with control in wire $i$ are placed in slice $c_{n,i}$ and hence we only need to show that all $CZ_{ik}$ in that slice are removed. Let us divide the analysis into two cases: (i) $k \in Odd(s(i))$ and (ii) $k \notin Odd(s(i))$. Consider the first case. Since
$k \in Odd(s(i))$, there exists a $Z_k^{s_i}$ in the measurement pattern and hence there exists a gate $CZ_{ik}$ in slice $c_{n,i}$. For this case, the index $m$ in RP1 is an odd number and therefore the application of RP1 creates odd many $CZ_{ik}$ gates in slice $c_{n,i}$. Note that those $CZ_{ik}$ gates can be created in different parts of $c_{n,i}$ (depending whether Proposition \ref{prop_rp1} or \ref{prop_rp1b} is satisfied) such that they cannot be grouped together trivially (like in Figure \ref{fig_lemczcx}-a). It is easy to verify that the only possible scenario for that to happen is when there exists $CX_{ik}$ in $c_{n,i}$ in between the created $CZ_{ik}$ gates. However, those $CZ_{ik}$ gates can be grouped together by moving all $CZ_{ik}$ to one side of the troublesome $CX_{ik}$ using the identity

\begin{equation*}
CZ_{ik}CX_{ik}|\psi\rangle_i |\phi\rangle_k = CX_{ik}CZ_{ik}(Z_i|\psi\rangle_i)|\phi\rangle_k
\end{equation*}
where $|\psi\rangle$ and $|\phi\rangle$ are arbitrary quantum sates and $Z$ is the Pauli operator (see Figure \ref{fig_lemczcx}-b). This way, since there exist even many $CZ_{ik}$ in sequence and even many $CZ$ gates equal the identity, all $CZ_{ik}$ can be simply removed from the circuit. The created single-qubit gate $Z_i$ can be pushed forward to the end of the $i$ wire, since it commutes with all gates in $c_{n,i}$. Moreover, since the final measurement is onto the $Z$-basis, the operator $Z_i$ has no effect in the measurement statistics and therefore can be removed from the circuit without changing the computation being implemented by the circuit (Figure \ref{fig_lemczcx}-c). An equivalent analysis applies for the second case, where $k \notin Odd(s(i))$. For this case RP1 create even many $CZ_{ik}$ in slice $c_{n,i}$ and there is no pre-existing $CZ_{ik}$ in that slice. The same identity can be used to group together all created $CZ_{ik}$, which cancel out since there are even many of those. Since both created and pre-existing $CZ_{ik}$ are removed from the circuit by the application of RP1, the Lemma holds.

\end{proof}

\begin{figure}
\center
\includegraphics[scale=.85]{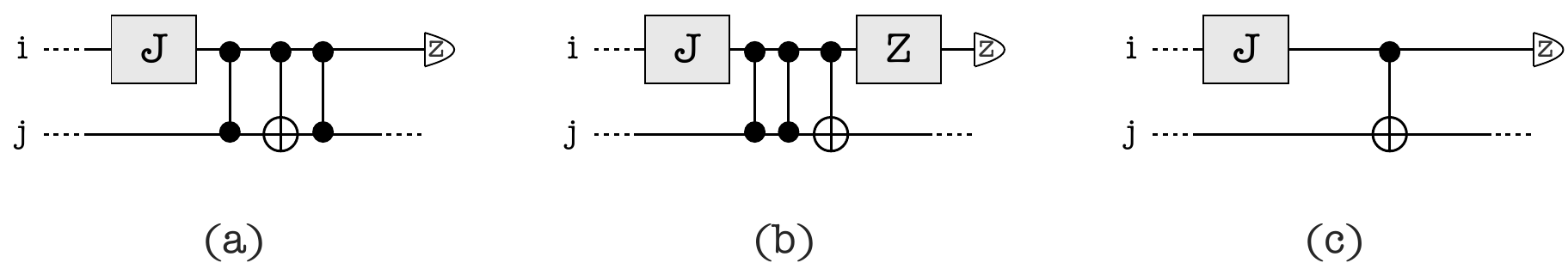}
\caption{Removing even many $CZ_{ij}$ from a SSF extended circuit. See Lemma \ref{lem_rp1} for more information.} 
\label{fig_lemczcx}
\end{figure}

Now we analyze the cases where $L_{s}(k) \geq L_{s}(i)$, that is, with $J_k$ gate blocking all gates of form $E_{kj}$ to be pushed past the correction gates (necessary to satisfy condition 3 in Definition \ref{def_jblock}). Since in those cases the corresponding RP will delete $E$ gates and we do not want to delete E gates of the form $E_{i,f(i)}$ (see Definition \ref{def_jblock}), we analyze separately the case where $f(k) \in s(i)$ (Proposition \ref{prop_rp2}) and the case where $f(k) \notin s(i)$ (Proposition \ref{prop_rp3}).

\begin{proposition} \label{prop_rp2} Let $i$ be a wire in a SSF extended circuit s.t. $J_i$ is in slice $\mathcal{J}_n$. If there exists a wire $k$ such that the following conditions are satisfied: (i) $L_{s}(k) \geq L_{s}(i)$; (ii) $f(k) \in s(i)$; (iii) gates $E_{kj_1}, ..., E_{kj_n}$ (with $j_1, ..., j_n \in s(i)$) can be trivially pushed to some slice $\mathcal{J}_{n'}$ (containing $J_k$); and (iv) the $E_{kf(k)}$ gate is the first gate acting on wire $f(k)$, then RP2 (Figure \ref{fig_RP2}) can be applied.\end{proposition}

\begin{proof}
Since $L_{s}(k) \geq L_{s}(i)$, $n' \leq n$. In slice $\mathcal{C}_{n'}$ there is a $CX_{kf(k)}$; since $f(k) \in s(i)$, we define $j_n = f(k)$. Also, $j_1, ..., j_n \in s(i)$ implies that $i \prec_s \{j_1, ..., j_n\}$ and hence the gates $J_{j_1}, ..., J_{j_n}$ are in some future slices (compared to $n$). Moreover, it also implies the existence of the operators $X_{j_1}^{s_i}, ..., X_{j_n}^{s_i}$ in the measurement pattern, which are translated to the extended circuit as $CX_{ij_1}, ..., CX_{ij_n}$ in slice $c_{n,i}$ (containing all the $CX$ and $CZ$ with control on qubit $i$). On the other hand, $L_{s}(k) \geq L_{s}(i)$ implies both $k \notin Odd(s(i))$ and $k \notin s(i)$ and therefore there is no $CZ_{ik}$ or $CX_{ik}$ in the circuit. Thus, if every gate $E_{kj_1}, ..., E_{kj_n}$ can be trivially pushed to slice $\mathcal{J}_{n'}$ and $E_{kf(k)}$ is the first gate acting on wire $j_n$, we have exactly the scenario depicted in Figure \ref{fig_RP2}-a. Therefore, the circuit identity in Figure \ref{fig_RP2} can be applied.\end{proof}

\begin{proposition} \label{prop_rp3} Let $i$ be a wire in a SSF extended circuit s.t. $J_i$ is in slice $\mathcal{J}_n$. If there exists a wire $k$ such that the following conditions are satisfied: (i) $L_{s}(k) \geq L_{s}(i)$; (ii) $f(k) \notin s(i)$; (iii) gates $E_{kj_1}, ..., E_{kj_n}$ (with $j_1, ..., j_n \in s(i)$) can be trivially pushed to slice $\mathcal{J}_{n'}$ (containing $J_k$); and (iv) the $E_{kf(k)}$ gate is the first gate acting on wire $f(k)$, then RP3 (Figure \ref{fig_RP3}) can be applied.\end{proposition}
\begin{proof}
Since $L_{s}(k) \geq L_{s}(i)$, then $n' \leq n$. In slice $\mathcal{C}_{n'}$ there is a $CX_{kf(k)}$, with $f(k) \notin s(i)$. Also, $j_1, ..., j_n \in s(i)$ implies that $i \prec_s \{j_1, ..., j_n\}$ and hence the gates $J_{j_1}, ..., J_{j_n}$ are also in some future slices (compared to $n$). Moreover, it also implies the existence of the operators $X_{j_1}^{s_i}, ..., X_{j_n}^{s_i}$ in the measurement pattern, which are translated to the extended circuit as $CX_{ij_1}, ..., CX_{ij_n}$ in slice $c_{n,i}$ (containing all the $CX$ and $CZ$ with control on qubit $i$). On the other hand, $L_{s}(k) \geq L_{s}(i)$ implies both $k \notin Odd(s(i))$ and $k \notin s(i)$ and therefore there is no $CZ_{ik}$ or $CX_{ik}$ in the circuit. Thus, if every gate $E_{kj_1}, ..., E_{kj_n}$ and $E_{kf(k)}$ can be trivially pushed to slice $\mathcal{J}_{n'}$ and $E_{kf(k)}$ is the first gate acting on wire $f(k)$, we have exactly the scenario depicted in Figure \ref{fig_RP3}-a. Therefore, the circuit identity in Figure \ref{fig_RP3} can be applied.
\end{proof}


In what follows we analyse some properties regarding the interplay between the RPs that will be crucial for the compactification algorithm for SSF extended circuits. 

\begin{lemma} \label{lem_czdeleted}
Let $C$ be an extended circuit obtained from a signal shifted measurement pattern. Then, if a rewrite procedure removes a gate $E_{jk}$ from the circuit, it will not be required for the application of any other rewrite procedure.
\end{lemma}

\begin{proof}
Note that only RP2 and RP3 can delete a $E_{jk}$ gate from the circuit.  The proof is divided into three parts: (i) we show that the $E_{jk}$ is never a flow edge and, therefore, it will not be required for the application of a $J$-gate identity; (ii) it will not be needed for the application of RP1 and (iii) it will not be needed to create a $CX$ (by using RP2 or RP3) in the next layers. Suppose $E_{jk}$ is a flow edge such that $k=f(j)$, with $j$ being any of qubits $(j_1, ..., j_{n-1})$ in RP2 or $(j_1, ..., j_n)$ in RP3. The other way around, that is, $j = f(k)$, is not possible since $f(k)$ is separately identified in both RP2 and RP3. Since $f(j) \in s(j)$ for all $j \in O^C$, we have $j \prec_s k$. By construction, RP2 and RP3 are applied if and only if $k\preccurlyeq_si$. Moreover, since $j \in s(i)$ in both RP2 and RP3, it holds that $i \prec_s j$. Putting everything together we have $i\prec_s j \prec_s k \preccurlyeq_s i$, which is a contradiction. Therefore, $E_{jk}$ is never a flow edge.

In this second part of the proof, we show that if $E_{jk}$ is deleted by RP2 or RP3, it will not be required for the application of RP1 in any future step. Using the notation defined earlier, the fact that RP1 is applied after RP2 or RP3 can be written as  $i^{(2)} \prec_s i^{(1)}$ or $i^{(3)} \prec_s i^{(1)}$, where $i$ is the target wire of the corresponding RP. This relation will be used in the rest of the proof. We show that $k^{(2)}$ and $k^{(3)}$ cannot be any of wires $j^{(1)}_1, ...,  j^{(1)}_n$ (denoted simply by $j^{(1)}$ in the rest of this proof) or $k^{(1)}$ and hence the gates $E_{j^{(2)}k^{(2)}}$ or $E_{j^{(3)}k^{(3)}}$ cannot be the same as $E_{j^{(1)}k^{(1)}}$, which is what we want to prove. Let us show for $k^{(2)}$. Since $k^{(2)} \preccurlyeq_s i^{(2)}$ and $i^{(1)} \prec_s k^{(1)}$ (true for both Propositions \ref{prop_rp1} and \ref{prop_rp1b}, which state the scenarios where RP1 is applied) we have $k^{(2)} \preccurlyeq_s i^{(2)} \prec_s i^{(1)} \prec_s k^{(1)}$ and therefore $k^{(2)}$ cannot be the same wire as $k^{(1)}$. Now assume wire $j^{(1)}$ is $k^{(2)}$. In RP2 we have $k^{(2)} \preccurlyeq_s i^{(2)}$ (by construction) and in RP1 it holds that $i^{(1)} \prec_s j^{(1)}$, since the target of a correcting gate is always placed before the $J$-gate acting on that same wire (Definition \ref{def_ext}). Putting everything together gives $k^{(2)}\preccurlyeq_s i^{(2)} \prec_s i^{(1)} \prec_s j^{(1)}$. But since by assumption $j^{(1)}$ is $k^{(2)}$, it gives $k^{(2)} \prec_s k^{(2)}$ which is a contradiction. The proof is the same for $k^{(3)}$. Therefore, if $E_{jk}$ is ``consumed'' by RP2 or RP3, it will not be required for the application of RP1 in any future step.

Finally for the third part note that an $E$ gate will be removed to create a $CX$ gate only through identity in Figure \ref{fig_rules}-b. The pair of such $E$ gates is always associated with a flow and a non-flow edge when used in RP2 or RP3. Therefore, there are two possible pairs of $E$ gates using the non-flow edge $E_{jk}$: $E_{jk}E_{kf(k)}$ (which would create $CX_{jf(k)}$) and $E_{jk}E_{jf(j)}$ (which would create $CX_{kf(j)}$). We want to show that whenever we ``consume'' $E_{jk}$, it will not be required to create a different $CX$ than the one already created. Suppose we use $E_{xy}$ to create $CX_{xf(y)}$. This is the case when $y\preccurlyeq_s i$ and $i\prec_sx$, if $x$ is a correcting wire and $y$ is a neighbour wire in either RP2 or RP3. Similarly, if $x$ is the neighbour wire and $y$ is the correcting one then $x\preccurlyeq_s i$ and $y\succ_s i$ (easily obtained by relabelling the wires in RP2 and RP3). Since the first pair of conditions is inconsistent with the last pair of conditions ($j\preccurlyeq_si$ and $k\succ_si$), we conclude that once a given $E$ is used to create a $CX$, that $E$ would not be required to create a different $CX$ in any further step of the algorithm since the partial order $\prec_s$ is not changed by the application of any RP.
\end{proof}

\begin{corollary} \label{cor_czcx}
Let $CX_{ij}$ be a gate in an extended circuit created by the application of a rewrite procedure. The $E$ gate ``consumed'' to create this $CX_{ij}$ can be univocally determined.
\end{corollary}
\begin{proof} It follows from the third part of the proof of Lemma \ref{lem_czdeleted} that if a gate of the form $E_{jk}$ is deleted in RP2 or RP3 we will have one of the following cases: the gates $E_{jk}E_{kf(k)}$ are used to create $CX_{jf(k)}$ or the pair $E_{jk}E_{jf(j)}$ are used to create $CX_{kf(j)}$. The $CX$ that is created by ``consuming'' the $E_{jk}$ depends exclusively on the relation between vertices $j$, $k$ and $f(j)$ in respect to the partial order $\prec_s$. Since the partial order $\prec_s$ is not changed by the application of any RP, there is a one-to-one correspondence between the deleted $E$ gate and the newly created $CX$ gate.
\end{proof}


\begin{algorithm}
	\caption{This algorithm decides which Rewrite Procedure must be applied by analysing how a given qubit $k$ is connected with the correcting set of another qubit $i$}
	\label{alg_chooserp}

	\KwIn{Wire labels $i$ and $k$.}
	
	\KwOut{A description of which Rewrite Procedure must be applied.}
	
	\Begin
	{
	\textbf{Read} $i$ and $k$.\\
\If{all conditions in Proposition \ref{prop_rp1}  \textbf{or} Proposition \ref{prop_rp1b} are satisfied}{
Apply RP1;\\
}

\If{all conditions in Proposition \ref{prop_rp2} are satisfied \textbf{or} conditions (iii) and/or (iv) in Proposition \ref{prop_rp2} are not satisfied but there exist a sequence of $(CX_{j_1j_n}, ..., CX_{j_{n-1}j_n})$ gates that can be pushed trivially to $c_{n,i}$ (slice containing all the $CX$ and $CZ$ with control on qubit $l$)}{
Apply RP2;\\
}
\If{all conditions in Proposition \ref{prop_rp3} are satisfied \textbf{or} conditions (iii) and/or (iv) in Proposition \ref{prop_rp3} are not satisfied but there exist a sequence $(CX_{j_1f(k)}, ..., CX_{j_nf(k)})$ gates that can be pushed trivially to $c_{n,i}$ (slice containing all the $CX$ and $CZ$ with control on qubit $i$)}{
Apply RP3;\\
}

{\bf otherwise} Abort
}
\end{algorithm}

The process of choosing the correct RP to be applied is summarised in Algorithm \ref{alg_chooserp}, which will be used as a subroutine in the compactification algorithm for SSF extended circuits (Algorithm \ref{alg_full}). The modified conditions in Lines 5 and 7 are based on the fact that if a $CZ$ required for the application of RP2 or RP3 is not in the circuit, the $CX$ that it would create will be (see Lemma \ref{lem_czdeleted}), allowing the application of future RPs.
In all RPs we assumed that some $E$ gates could be moved to the beginning of a given slice $\mathcal{J}$ in a trivial way. It will not be true in general for SSF extended circuits, since there might exist several other gates in the extended circuit such that the aforementioned $E$ gates could not be moved trivially to $\mathcal{J}$. In other words, the conditions in Propositions \ref{prop_rp1} to \ref{prop_rp3} would not be satisfied. The algorithm in the next section address exactly this problem: it provides an ordering where Algorithm \ref{alg_chooserp} never aborts. This ordering is the \textit{global} structure coming into play, since it is related to SSF and flow of the graph.

\subsection{Obtaining optimised compact circuits} \label{sec_alg}

In this section we explain how the compactification algorithm for SSF extended circuits (Algorithm \ref{alg_full}) works and show some examples. The goal of Algorithm \ref{alg_full} is to create $|O^C|$ many $J$-blocks (see Definition \ref{def_jblock}) in a SSF extended circuit and then apply the $J$-gate identity for all $J$-block, removing $|O^C|$ many wires from the extended circuit. Therefore, the output of Algorithm \ref{alg_full} is the compact form of the inputed SSF extended circuit. One of the main differences between SSF and other gflows is the notions of stepwise influencing path, introduced in Section \ref{sec_ip}. The next lemma use some properties of stepwise influencing path to relate the SSF correcting set to the partial order of flow. This relation will play an important role in Algorithm \ref{alg_full} since it gives an appropriate ordering for the application of the RPs. 

\begin{lemma} \label{lem_flowssf}
Let $(G,I,O)$ be an open graph with flow $(f,\prec_f)$ and SSF $(s,\prec_s)$. Then, for all $v \in N(j)\setminus\{i\}$, where $j \in s(i)$, it holds that $v \succ_f i$.
\end{lemma}

\begin{proof}First suppose $j=f(i)$. Then, by flow definition (Definition \ref{def_flow}), $v \succ_f i$. Now suppose $j \neq f(i)$ then from Definition \ref{def_flow}, $v \succ_f f^{-1}(j)$. By Lemma \ref{lem_path}, there exists a step-wise influencing path passing through $f(i)$ and $f^{-1}(j)$, namely $\wp_i(j)$. It follows from Corollary \ref{cor_neighbour}, that there exists $f^{-1}(k)$ such that $k \in s(i)$ and $f^{-1}(j) \in N(k)$ and, consequently, $f^{-1}(j) \succ_f f^{-1}(k)$. Lemma 2.14 allows us to repeat this process to find the previous vertices in the path $\wp_i(j)$ until we reach vertex $i$. Hence we conclude that $f^{-1}(j) \succ_i i$ and therefore $v \succ_f f^{-1}(j) \succ_f i$.
\end{proof}

Before running Algorithm \ref{alg_full}, the slices $c_{n,i}$ must be arranged in the extended circuit from right to left respecting the order imposed by $\prec_f$ (see for instance Figure \ref{fig_examplessf}). This can always be done since the set of gates in a given slice $c_{n,i}$ commutes with the gates in an other slice $c_{n,j}$, for any valid $j$. The algorithm starts with a $\mathtt{for}$ loop that runs for one SSF layer at a time, starting with the input layer and moving onwards until the last layer of non-output qubit is considered. At each iteration $n$ of the $\mathtt{for}$ loop, a rewrite procedure is applied to each pair of qubits $(i,k)$ s.t. $J_i$ is in $\mathcal{J}_n$ and $k \in N(s(i))$.
The two $\mathtt{foreach}$ loops and the two $\mathtt{while}$ loops define the order in which Algorithm \ref{alg_chooserp} will be called. This ordering, which is the inverse of the order given by $\prec_f$, assures that Algorithm \ref{alg_chooserp} will never abort when it is called by Algorithm \ref{alg_full}.

\begin{algorithm} [ht]
	\caption{Transforms a SSF extended circuit into a compact circuit (Definition \ref{def_compactform}).}
	\label{alg_full}

	\KwIn{SSF extended circuit.}
	
	\KwOut{Compact form of the SSF extended circuit.}
	
	\Begin
	{
	$D \leftarrow \max_{i \in G} \{L_{s}(i)$\}\\
	
	\For{ $n = 1$ to $D$}{
\textbf{Let} $W_n = \{a_{n,1}, ..., a_{n,m_n}\}$ be the set of the $m_n$ wires with $J$-gate in slice $\mathcal{J}_n$ \\
\textbf{Let} $S_{max}^{n} \leftarrow \max_{i \in W_n}\{L_f(i)\}$ \\
\textbf{Let} $S_{min}^{n} \leftarrow \min_{i \in W_n}\{L_f(i)\}$ \\
\While{ $S_{max}^{n} \geq S_{min}^{n}$}{ 
\ForEach{ $a_{n,i} \in W_n$ such that $L_f(a_{n,i}) = S_{max}^{n}$}{
\textbf{Let} $NSTEP \leftarrow \max_{k \in \{N(s(a_{n,i}))\setminus \{a_{n,i}\}\}}\{L_f(k)\}$ \\
\While{ $NSTEP > S_{max}^{n}$}{ 
\CommentSty{// That is, for all $k \in \{N(s(a_{n,i}))\setminus \{a_{n,i}\}\}$ (Due to Lemma \ref{lem_flowssf}) }\\
\ForEach{ $k \in \{N(s(a_{n,i}))\setminus \{a_{n,i}\}\}$ such that $L_f(k) = NSTEP$}{
$i \leftarrow a_{n,i}$;\\
\textbf{Run} Algorithm \ref{alg_chooserp}\\
}
$NSTEP \leftarrow (NSTEP-1)$\\
}
}
$S_{max} \leftarrow (S_{max}-1)$\\
}
}	
Remove from the circuit all gates $CX_{ij}$ placed in any slice $\mathcal{C}_n$ such that $j \neq f(i)$.\\
Apply the $J$-gate identity for all pairs of wires $\{i,f(i)\}$ such that $i \in O^C$.
}
\end{algorithm}

\begin{lemma} \label{lem_algneveraborts}
Algorithm \ref{alg_chooserp} never aborts when called by Algorithm \ref{alg_full}.
\end{lemma}
\begin{proof}
Let us start by showing that the algorithm works for the first iteration of the $\mathtt{for}$ loop (\textit{i.e.}, $n=1$) and then we explain why it will work for all other layers. In what follows we use the notation $a_{n,i}$ to represent a target wire $i$ in layer $n$; note that there might exist more than one such wire in the same layer. Let $S^{n}_{max}$ ($S^{n}_{min}$) be the maximum (minimum) value of $L_f(a_{n,i})$ for $a_{n,i} \in W_n$ (as defined in Lines 5 and 6 in the algorithm), where $W_n$ is the set of the $m_n$ wires with $J$-gate in slice $\mathcal{J}_n$. 
According to Lemma \ref{lem_flowssf}, for any $a_{1,i}$ such that $L_f(a_{1,i})= S^{1}_{max}$, and any qubit $k$ connected to a qubit in $s(a_{1,i})$ it holds that $k \notin W_1$. In this case, condition (i) in Proposition \ref{prop_rp1} is satisfied. Since we are in the first SSF layer, all $E$ gates required for the application of any of those RPs can be trivially pushed forward to slice $\mathcal{J}_1$, hence all conditions in Proposition \ref{prop_rp1} are initially satisfied. After the application of one or more RP1, it might happen that some $E$ gates get stuck in between two $CX$ gates in slice $c_{1,i}$ (which happens when a neighbour wire $k$ is itself in both $N(s(a_{1,i}))$ and $s(a_{1,i})$ sets). In this case the conditions in Proposition \ref{prop_rp1b} are the ones satisfied. In both cases, only RP1 can be selected by Algorithm \ref{alg_chooserp}. Therefore, Algorithm \ref{alg_chooserp} would not abort for all target qubits $a_{1,i}$ such that $L_{f}(a_{1,i})= S^{1}_{max}$. As a consequence, since only RP1 is applied, all gate $E_{kp}$, such that $p \in s(a_{1,i})$, are moved past the correction slice $c_{1,i}$, for all $a_{1,i}$ such that $L_{f}(a_{1,i})= S^{1}_{max}$. 

Now consider qubits $a_{1,i}$ such that $L_f(a_{1,i})= S^{1}_{max} - 1$, that is, the second iteration of the first $\mathtt{while}$ loop. For $k \in N(s(a_{1,i}))$, Lemma \ref{lem_flowssf} implies that either $k \notin W_1$ or $\{(k \in W_1) \land [L_f(k)= S^{1}_{max}]\}$ is true. For $k \notin W_1$, the procedure to be applied is RP1, similarly and by the same reasons as in the previous iteration. For the other case the applicable rewrite procedures are either RP2 or RP3, since $J_k$ in $\mathcal{J}_1$ implies $L_f(k) = L_f(a_{1,i})$. Note that in the previous iteration (qubits $a_{1,i}$ s.t. $L_{f}(a_{1,i})= S^{1}_{max} $, like the wire $k$ being analysed now), all $E_{f(k)v}$, for any $v \neq k$ were moved past the corresponding correction slice and hence $E_{f(k)k}$ is the first gate acting on the $f(k)$ wire (as required by both Propositions \ref{prop_rp2} and \ref{prop_rp3}). Moreover, when the circuit identity in Figure \ref{fig_rules}-b creates the new $CX$ gates (first step in both RP2 and RP3), it can be pushed forward to $c_{1,i}$ because the only gate in between is $CX_{kf(k)}$, which commutes with the $CX$ gates we want to push forward. Therefore, since all conditions in either Proposition \ref{prop_rp2} or Proposition \ref{prop_rp3} would be satisfied for all qubits $a_{1,i}$ such that $L_f(a_{1,i})= S^{1}_{max} - 1$, the corresponding $RP$ can be successfully applied.

To complete the proof by induction for the first iteration of the $\mathtt{for}$ loop, assume Algorithm \ref{alg_chooserp} has not aborted in the first $q$ iteration of the first $\mathtt{while}$ loop. Then in the $(q+1)^{th}$ iteration, qubits $a_{1,i}$ such that $L_f(a_{1,i})= (S^{1}_{max} - q) \geq S^{1}_{min}$, where $S^{1}_{min} \equiv \min_{a_{1,i} \in W_1}\{L_f(a_{1,i})\}$, are considered. If a given qubit $k$ is neighbour of a qubit in $s(a_{1,i})$, then Lemma \ref{lem_flowssf} implies that either $k \notin W_1$ or $\{(k \in W_1) \land [(S^{1}_{max} - q) < L_f(k)\}$ is true. Here the same analysis as before applies; First, if $k\notin W_1$, Algorithm \ref{alg_chooserp} will apply RP1, moving the corresponding $E$ gates past slice $c_{1,i}$. On the other hand, if $\{(k \in W_1) \land [(S^{1}_{max} - q) < L_f(k)\}$, either the conditions in Proposition \ref{prop_rp2} or Proposition \ref{prop_rp3} will be satisfied,  since all $E$ gates acting on qubits in the sets $s(a_{1,v})$, for all $a_{1,v} \in W_1$ s.t. $L_f(a_{1,v}) > (S^{1}_{max} - q)$, were moved past slice $c_{1,v}$ in previous iterations (with the obvious exception of gates of the form $E_{vf(v)}$). The procedure continues until there is no qubit left in $W_1$ to be considered. Therefore, Algorithm \ref{alg_chooserp} never aborts when called during the first iteration ($n=1$).

Now let us analyse how the $E$ gates placement changed in the circuit after the $n^{th}$ iteration of the $\mathtt{for}$ loop where Algorithm \ref{alg_chooserp} applied a RP each time it was called by Algorithm \ref{alg_full}.  First, note that all $E$ gates acting on qubits in $s(i)$, for $i \in W_n$, were moved to $\mathcal{E}_{n+1}$ (via RP1) or transformed into a $CX$ (via RP2 or RP3) and then moved to the same slice, with the exception of gates of the form $E_{if(i)}$ which remains in $\mathcal{E}_n$ and will later be used for the application of the $J$-gate identity. Moreover, Lemma \ref{lem_czdeleted} guarantees that if an $E$ gate is deleted in a given iteration of the algorithm, it will not be required in any future step of the algorithm.

Note also that those $E$ and $CX$ gates were moved to the  $\mathcal{E}_{n+1}$ slice in a specific order, namely the inverse of the order induced by $\prec_f$. Since this ordering is a property of the associated graph, it does not change during the run of the algorithm. It means that for any given iteration $n$ of the $\mathtt{for}$ loop, the ordering of the two-qubit gates required for the application of a RP will be in agreement with the order Algorithm \ref{alg_chooserp} will be called by Algorithm \ref{alg_full}. Hence the condition (required for the application of any RP) that the required two-qubit gates can be pushed trivially to a given slice is always satisfied. Therefore,  at each iteration of the $\mathtt{for}$ loop, it arranges all two-qubit gates necessary for the next iteration in the same order it will be called. 

To complete the proof by induction, let us assume that in the first $(p-1)$ iterations of the $\mathtt{for}$ loop Algorithm \ref{alg_chooserp} has not aborted, and then we show that it will not abort in the $p^{th}$ iteration. Let us analyse the first iteration of the first $\mathtt{while}$ loop for $n=p$. According to Lemma \ref{lem_flowssf}, for any $a_{p,i}$ such that $L_f(a_{p,i})= S^{p}_{max}$, if a given qubit $k \in N(s(a_{p,i}))$ then either $k \in W_m$ (if there exists $m<p$ s.t. $S^{m}_{max} > S^{p}_{max}$) or the $J_k$ gate belongs to a layer $\mathcal{J}_{m'}$ s.t. $m'> p$. If the latter condition is the case, then by definition only RP1 is applied to $a_{p,i}$. On the other hand, if $k \in W_m$, RP2 or RP3 must be applied. Since by assumption Algorithm \ref{alg_chooserp} has not aborted in any previous iteration, the conditions for the application of the aforementioned RPs are satisfied and Algorithm 2 will not abort in the current iteration.

Now assume Algorithm \ref{alg_chooserp} has not aborted in the first $q$ iteration of the first $\mathtt{while}$ loop for $n=p$. In the $(q+1)^{th}$ iteration of this loop, qubits $a_{p,i}$ such that $L_f(a_{p,i})= (S^{p}_{max} - q) \geq S^{p}_{min}$, where $S^{p}_{min} \equiv \min_{a_{p,i} \in W_p}\{L_f(a_{p,i})\}$, are considered. For a given $k$ neighbour of a qubit in $s(a_{p,i})$, Lemma \ref{lem_flowssf} implies that one of the following statement holds: (i) $k \in W_m \land [S^m_{max}>(S^{p}_{max} - q) ]$ for some $m<p$, (ii) $(k \in W_p) \land [(S^{p}_{max} - q) < L_f(k)]$ or (iii) $k \in W_{m'}$, for $m'>p$. Here the same analysis as before applies; For cases (i) and (ii) RP2 or RP3 is applied and for (iii) RP1 will move the corresponding $E$ gates past slice $c_{p,i}$. Since the same order is respected throughout the algorithm (the inverse of the order defined by flow), the conditions for the application of the aforementioned RPs are satisfied and Algorithm 2 will not abort in the current iteration. This procedure is repeated until there is no more elements in the set $W_p$, which concludes our proof by induction and proves the Lemma.

\end{proof}

\begin{theorem}\label{the_compact}
Algorithm \ref{alg_full} outputs a compact circuit.
\end{theorem}
\begin{proof}
We show that Algorithm \ref{alg_full} creates $|O^C|$ many $J$-gate blocks (Definition \ref{def_compactform}) and then removes $|O^C|$ many wires from the circuit (using the $J$-gate identity), yielding a compact circuit. The first two conditions in Definition \ref{def_jblock} are trivially satisfied for all SSF extended circuit. In Lemma \ref{lem_algneveraborts} we proved that Algorithm \ref{alg_full} moves to a future slice (or removes) all gates $E$ acting on qubits $j \in s(i)$ ($\forall i \in O^C$) except gate $E_{if(i)}$ and hence condition 3 in Definition \ref{def_jblock} is also satisfied. Now we prove that condition 4 is satisfied. Note that all non-input qubit are initialised in state $|+\rangle $. According to SSF construction (Proposition \ref{prop_ssf}), for all qubit $j \in s(i)$, there exists $f^{-1}(j)$. By definition $f: O^C \rightarrow I^C$ and therefore every qubits in $s(i)$, for all $i \in O^C$, starts in the $|+\rangle$ state. Therefore, all remaining $CX_{ij}$ such that $j \neq f(i)$ can me removed from the circuit without changing the computation, since $CX_{ij}|+\rangle_j = |+\rangle_j$ (remember that the only gate $E$ acting on wires $j \in s(i)$ left behind by the algorithm is $E_{if(i)}$). Therefore, after the step in Line 17, condition 4 in Definition \ref{def_compactform} is also satisfied. Since all non-output wires in an extended circuit are measured in the Z basis by construction, condition 5 is also satisfied and hence Algorithm \ref{alg_full} has created $|O^C|$ many $J$-blocks after the step in Line 17. Finally, in Line 18 the $J$-gate identity is applied to all $J$-blocks, resulting in a compact circuit.
\end{proof}

We complete this section with two examples using the above algorithm. In the first one we start with a generic quantum circuit, translate it to the MBQC model, find the SSF extended circuit of it and then apply Algorithm \ref{alg_full} to put the circuit in the compact form. In the second example, we analyse an example with Pauli measurements and show that if we want to keep the parallelization introduced by it we can no longer find the compact form of the extended circuit.

\subsubsection*{Example 1}
\label{sec_example1}

 \begin{figure}
\center
\includegraphics[scale=.5]{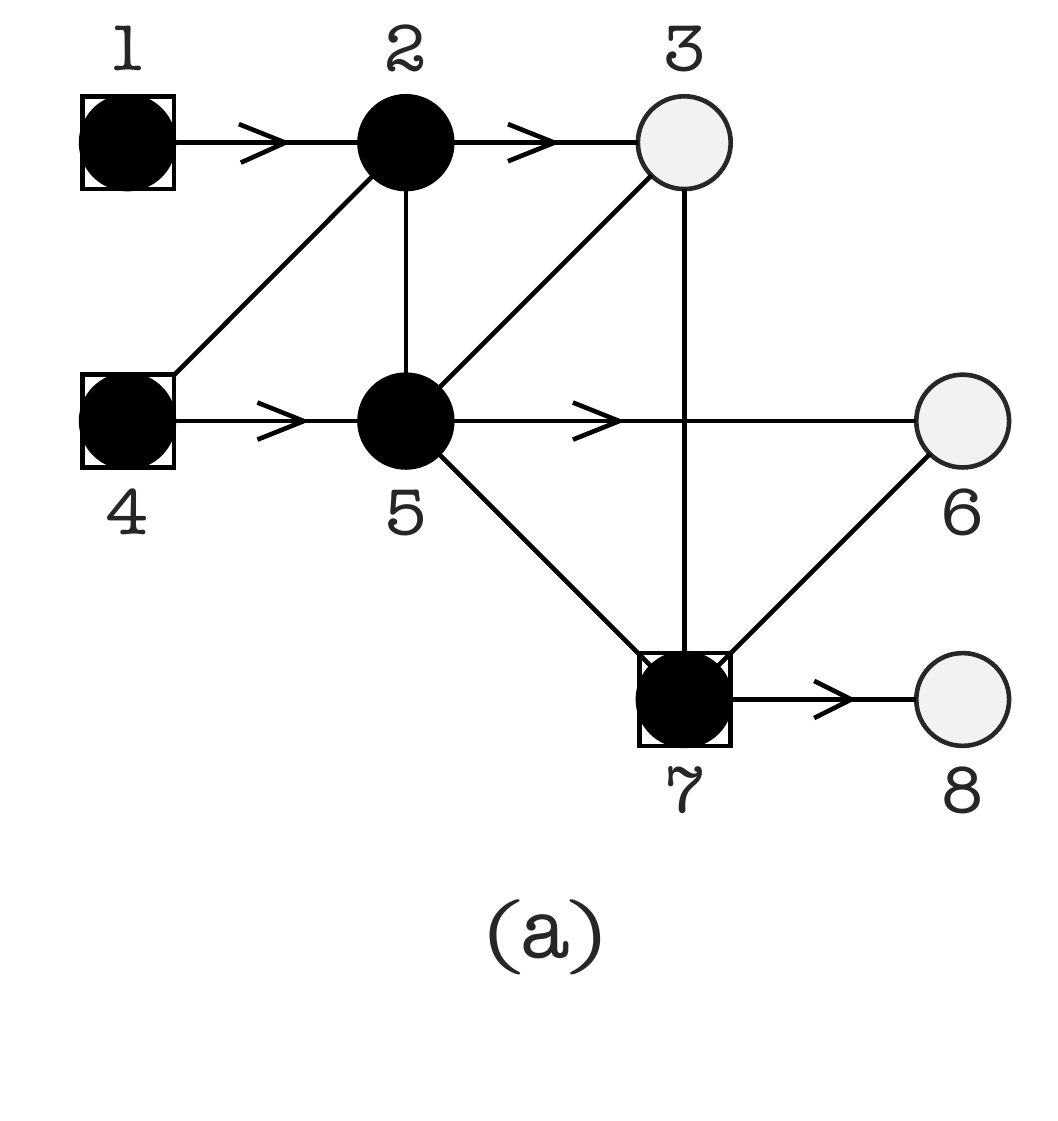}
\caption{Open graph associated to the example in Section \ref{sec_alg}}
\label{fig_examplegraphs}
\end{figure}

In this first example we apply our method to optimise the circuit in Figure \ref{fig_example1}-a. Using the method from \cite{BroadbentK09}, we translate the circuit in Figure \ref{fig_example1}-a to the following measurement pattern, resulting in the graph in Figure \ref{fig_examplegraphs}-a. The difference in depth between the flow and SSF is shown in Table \ref{table_ex1}.

\begin{center} \label{table_ex1}
    \begin{tabular}{ | l | l | l |}
    \hline
\textbf{Method} & \textbf{depth} & \textbf{partial order} \\ \hline
Flow  &5& $1 \prec_f 4 \prec_f 2 \prec_f 5 \prec_f 7$\\ \hline
SSF & 2&$1,4,7 \prec_s 2,5$ \\ \hline
    \end{tabular}
\end{center}

\begin{align} \label{eq_ex1ssf}
X^{s_6}_8Z^{s_5}_6X^{s_2}_6X^{s_2}_3M_5^{\theta_5}M_2^{\theta_2}
X^{s_7}_8
Z^{s_4}_6Z^{s_4}_3X^{s_4}_6X^{s_4}_3X^{s_4}_5
Z^{s_1}_8X^{s_1}_5X^{s_1}_3X^{s_1}_2M_7^{\theta_7}M_4^{\theta_4}M_1^{\theta_1}\nonumber \\
E_{78}E_{67}E_{57}E_{56}E_{45}E_{37}E_{35}E_{25}E_{24}E_{23}E_{12}
\end{align}

 \begin{figure}
\includegraphics[scale=.9]{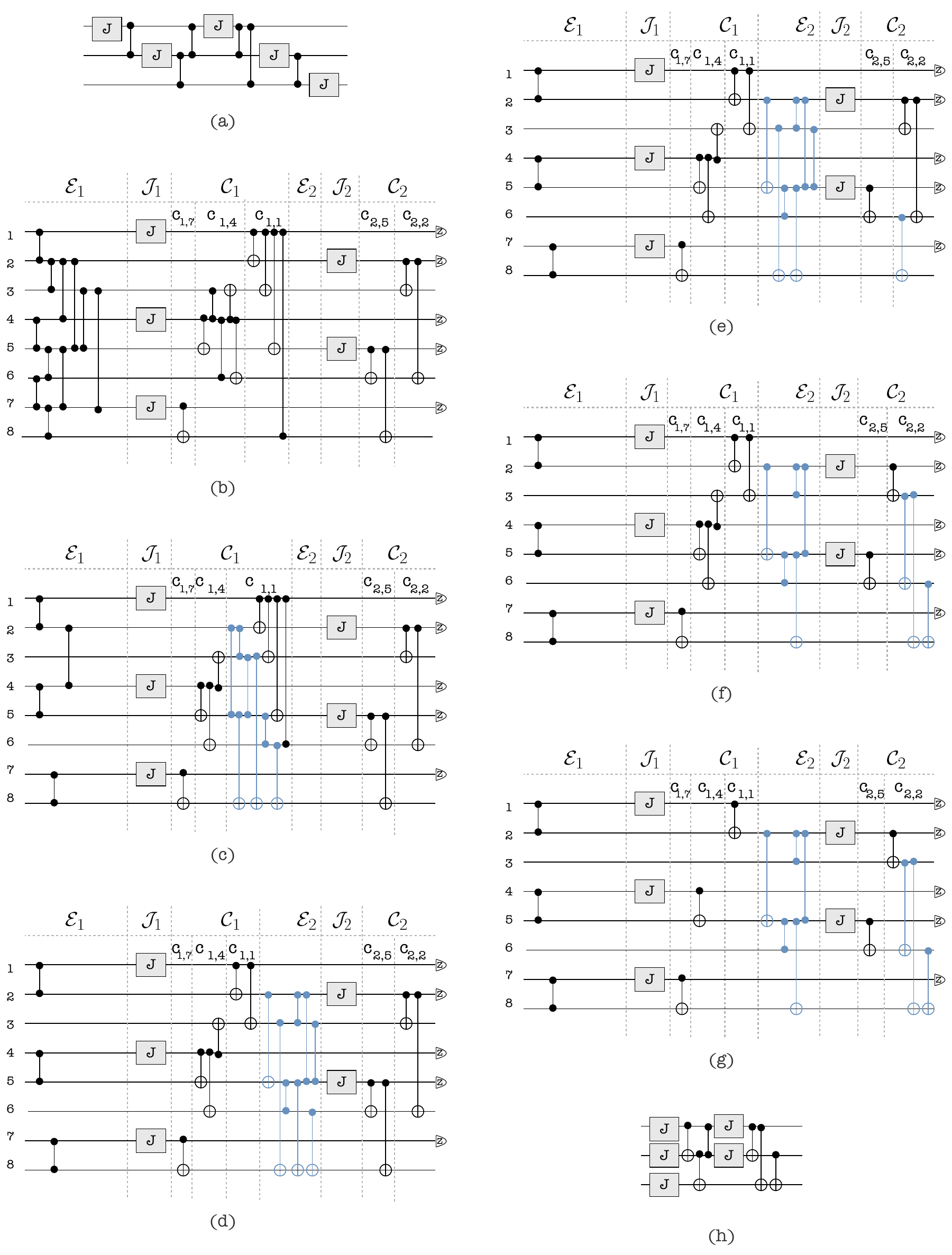}
\caption{A complete example of global circuit optimization. In each figure the shaded (blue online) gates are the new gates obtained through the RPs (see Section \ref{sec_example1} for more information.)}
\label{fig_example1}
\end{figure}

The extended translation (Definition \ref{def_ext}) of the measurement pattern in Equation \ref{eq_ex1ssf} gives the circuit in Figure \ref{fig_example1}-b. Let us apply Algorithm \ref{alg_full} to find the compact form of that extended circuit. We start with the first SSF layer of circuit in Figure \ref{fig_example1}-a. For $n=1$ we have $W_1 = \{1,4,7\}$, $S^{1}_{max}=L_f(7)=5$ and $S^{1}_{min}=L_f(1)=1$. Thus, we start with qubit 7 with $N(s(7))\backslash \{7\}= \{\emptyset\}$ and therefore no rewrite procedure is applied. The value of $S^{1}_{max}$ is decreased to 2 (there is no qubit $i$ s.t. $L_f(i) = 3$ or $4$). Now we analyse qubit 4 because $L_f(4)=2$ with $N(s(4))\backslash \{4\}= \{2, 3, 5, 6, 7\}$. For qubits in this set it holds that $2 \prec_f 5 \prec_f 7 \prec_f 3,6$. The maximal value of $L_f$ is stored in $NSTEP$. For both qubits $3$ and $6$, Algorithm \ref{alg_full} applies RP1, moving $CZ_{65}$ and $CZ_{35}$ past $CX_{45}$. 

Now $NSTEP \leftarrow NSTEP -1$ and the qubit considered is qubit $7$. For qubit $7$ the algorithm applies RP2, transforming $CZ_{67}$, $CZ_{57}$ and $CZ_{37}$ into $CX_{68}$, $CX_{58}$ and $CX_{38}$ respectively and moving those new $CXs$ past the correction structure, removing $CX_{48}$ from the circuit. Now consider qubit $5$ and move $CZ_{56}$ and $CZ_{53}$ past the correction structure, using RP2. The same rewrite procedure is applied for qubit $2$, moving $CZ_{23}$ and $CZ_{25}$ past the correction structure. The rewritten circuit is depicted in Figure \ref{fig_example1}-c.

In the next step, the value of $S^1_{max}$ is decreased to 1, where the qubit to be considered is qubit $1$. $N(s(1))\backslash \{1\}= \{2, 3, 4, 5, 6, 7\}$. For qubits in this set it holds that $4 \prec_f 2 \prec_f 5 \prec_f 7 \prec_f 3,6$. For $6$, Proposition \ref{prop_rp1} is satisfied and hence Algorithm \ref{alg_full} applies RP1, moving $CZ_{65}$ past $CX_{15}$. Then for $k=3$, Algorithm \ref{alg_full} applies RP1 (Proposition \ref{prop_rp1}), moving $E_{35}$ and $E_{23}$ past $CX_{12}$ and $CX_{15}$. Next, RP2 is applied to $k=7$, moving $CX_{38}$ and $CX_{58}$ past $CX_{13}$ and $CX_{15}$. For $k=5$, we have that $k \in s(i)$ (condition (i) in Proposition \ref{prop_rp1b}) and $E_{35}$ in slice $c_{1,1}$ after $CX_{15}$ gate ((condition (ii) in Proposition \ref{prop_rp1b}). Therefore, for this case RP1 moves $E_{35}$ and $E_{25}$ past $CX_{12}$ and $CX_{13}$, creating two $CZ_{15}$ separated by $CX_{15}$. By Lemma \ref{lem_rp1}, we know those two $CZ_{15}$  can simply be removed from the circuit without changing the computation being implemented. For $k=2$, Proposition \ref{prop_rp1} is satisfied and hence RP1 moves $E_{25}$ and $E_{23}$ past $CX_{13}$ and $CX_{15}$. Finally, for qubit $4$, RP3 transforms $E_{24}$ into $CX_{25}$ and move it past $CX_{12}$ and $CX_{15}$, removing the latter in the process, resulting in Figure \ref{fig_example1}-d.

In the second SSF layer we have $n=2$, $W_2 = \{2,5\}$, $S^{2}_{max}=L_f(2)=3$ and $S^{2}_{min}=L_f(5)=4$. Thus, we start with qubit 5. Since $N(s(5))\backslash \{5\}=\{7\}$, Algorithm \ref{alg_full} applies RP2, moving $CX_{68}$ past $CX_{56}$ and $CX_{58}$ while removing $CX_{58}$ in the process. After that, the value of $NSTEP$ decreases and qubit $2$ is considered. $N(s(2))\backslash \{2\}=\{5,7\}$ and since $L_f(5) <L_f(7)$, qubit $7$ is considered first. In this case Algorithm \ref{alg_full} applies RP3 moving gates $CX_{38}$ and $CX_{68}$ past the correction structure. Now for qubit $5$,  RP2 is applied transforming $CZ_{53}$ into $CX_{36}$. The created $CX$ is moved forward past the correction structure, removing $CX_{26}$ from the circuit. The resulting circuit is shown in Figure \ref{fig_example1}-f. The command in Line 17 of Algorithm \ref{alg_full} removes the remanning undesired gates using the identity trivial $CX_{ij}|+\rangle_j=|+\rangle_j$. The resulting circuit is depicted in Figure \ref{fig_example1}-g. Finally, the $J$-gate identity is applied for  every qubit $i \in O^C$ resulting in the optimised compact circuit shown in Figure \ref{fig_example1}-h.

\subsubsection*{Example 2}
In this example we explore what is the role of Pauli measurements in the context of compactification procedures. Let us start by reviewing how the correction operators are modified when applied to a qubit measured with $0$ and $\pi/2$ angles (corresponding to Pauli measurements):
\begin{eqnarray}
M^{\frac{\pi}{2}}_i X^{s}_i &=& M^{\frac{\pi}{2}}_i Z^{s}_i \label{eq_pauliy}\\
M^{0}_i X^{s}_i &=& M^{0}_i \label{eq_paulix}
\end{eqnarray}
Note that in the MBQC context, both substitutions clearly might reduce the depth of the computation: Equation \ref{eq_pauliy} substitute $Z^{s}_i$ for $X^{s}_i$, which can then be removed using signal-shifting, and Equation \ref{eq_paulix} simply deletes the existing $X^{s}_i$ dependency. However, in the associated extended circuit, these optimisations can not be implemented alongside the compactification procedure developed in this section (Algorithm \ref{alg_full}), forcing one to choose between optimization in time or memory. To see why, note that both Equations \ref{eq_pauliy} and \ref{eq_paulix} change the correcting structure for every qubit encoded in the index $s$. As a consequence, the step-wise influencing path $\wp$ (Definition \ref{def_path}), which is the backbone of the compactification protocol, can not be defined anymore for these qubits.

In summary, if one wants to save memory in the circuit model, the substitutions in Equations \ref{eq_pauliy} and \ref{eq_paulix} must be avoided, allowing the identification of all step-wise influencing paths and, consequently, the removal of auxiliary qubits. Conversely, if $J$-gate parallelization is the goal, Algorithm \ref{alg_full} can be easily adapted to create just $|O^C|-p$ many $J$-blocks in the extended circuit (instead of $|O^C|$ many), where $p$ is the number of Pauli measurements in the associated measurement pattern.

Consider the circuit shown in Figure \ref{fig_pauliopt}-a. It has two $J$ gates with arbitrary angles (arbitrary angles are omitted in the Figure) and two Pauli angles, namely $0$ and $\frac{\pi}{2}$. Using the method in \cite{BroadbentK09}, we can translate this circuit to the MBQC model; the associated graph is shown in Figure \ref{fig_pauliopt}-b and the flow measurement pattern is given by:

\begin{equation}
Z^{s_4}_6Z^{s_1}_3X^{s_5}_6X^{s_2}_3M_5^{0}M_2^{\pi/2}Z^{s_4}_2X^{s_4}_5M_4^{\theta_4}Z^{s_1}_5Z^{s_1}_4X^{s_1}_2M_1^{\theta_1}E_G
\end{equation}
where, $E_G = E_{56}E_{45}E_{25}E_{24}E_{23}E_{12}$. Using the signal shifting technique explored in this paper, we obtain the following SSF measurement pattern:
\begin{equation} \label{eq_ex2ssf}
Z^{s_1+s_4}_6Z^{s_1}_3X^{s_5+s_1}_6X^{s_1+s_2+s_4}_3M_5^{0}M_2^{\pi/2}X^{s_1}_2X^{s_1+s_4}_5M_4^{\theta_4}M_1^{\theta_1}E_G
\end{equation}
Now let us see what happens when we use the fact that some measurements are Pauli measurements. Using Equation \ref{eq_pauliy}, we can optimise it by removing the dependency between measurements $1$ and $2$, obtaining the following measurement pattern:\begin{equation}
Z^{s_1+s_4}_6Z^{s_1}_3X^{s_5+s_1}_6X^{s_2+s_4}_3M_5^{0}X^{s_1+s_4}_5M_2^{\pi/2}M_4^{\theta_4}M_1^{\theta_1}E_G
\end{equation}
Alternatively, the measurement on qubit $5$ can also be parallelised with the ones in $1$ and $4$, using Equation \ref{eq_paulix}. The optimized measurement pattern is the following:
\begin{equation}
Z^{s_1+s_4}_6Z^{s_1}_3X^{s_5+s_1}_6X^{s_1+s_2+s_4}_3M_2^{\pi/2}X^{s_1}_2M_5^{0}M_4^{\theta_4}M_1^{\theta_1}E_G
\end{equation}
Finally, both optimisations can be considered together. In this case, all measurements can be performed at once:
\begin{equation}  \label{eq_ex2allpauli}
Z^{s_1+s_4}_6Z^{s_1}_3X^{s_5+s_1}_6X^{s_1+s_2+s_4}_3M_2^{\pi/2}M_5^{0}M_4^{\theta_4}M_1^{\theta_1}E_G
\end{equation}

\begin{figure}
\center
\includegraphics[scale=.7]{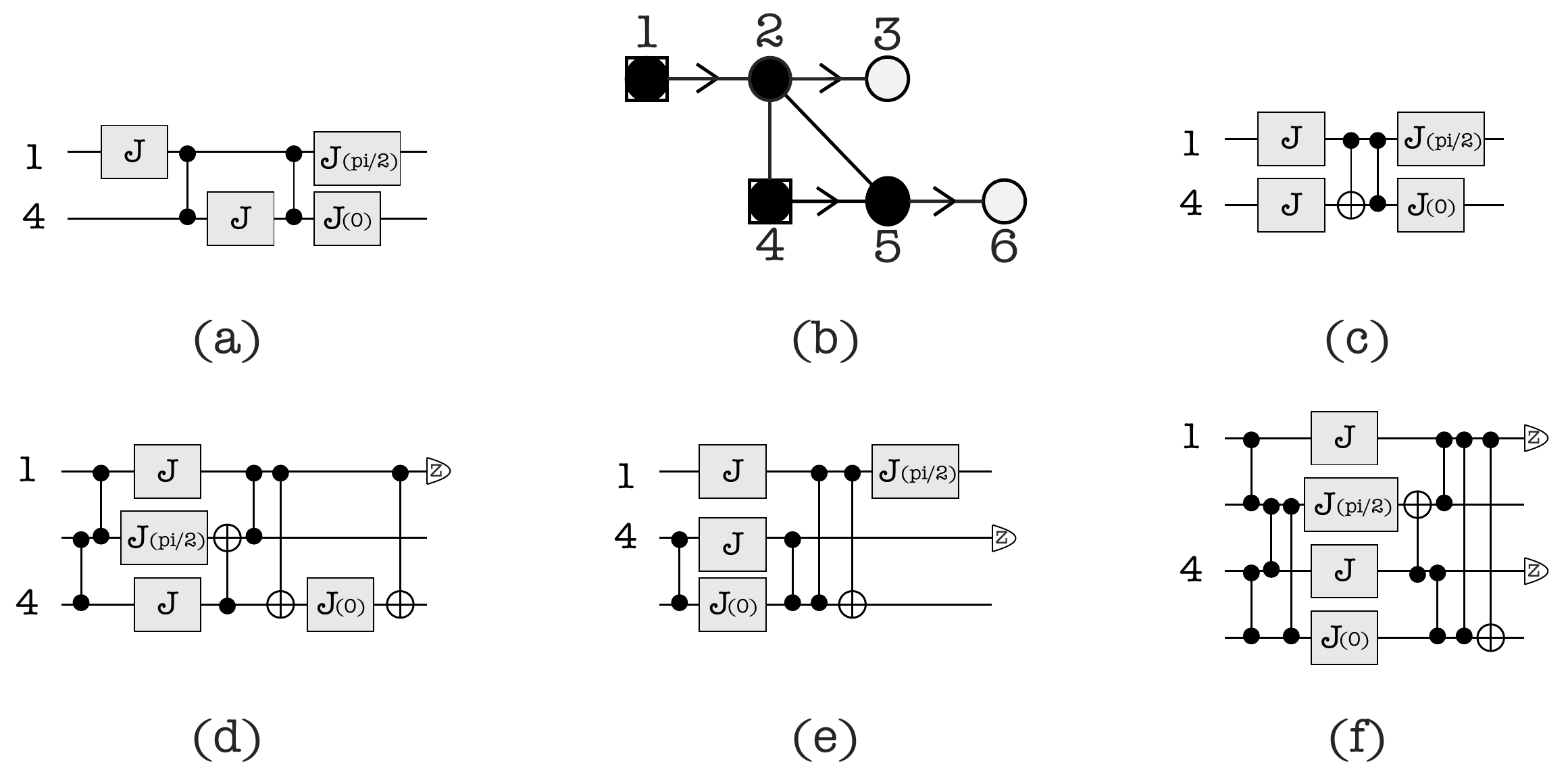}
\caption{Using compactification procedures to obtain $J$-gate parallelization for circuits with Pauli angles. See example 2 for more information.}
\label{fig_pauliopt}
\end{figure}

With few adaptations\footnote{Due to the modifications introduced by equations \ref{eq_pauliy} and \ref{eq_paulix}, it is not possible to create $J$-blocks for $p$ many wires, where $p$ is the number of times equations \ref{eq_pauliy} or \ref{eq_paulix} were used in the original measurement pattern.}, Algorithm \ref{alg_full} can be applied to the extended circuits associated to the measurement patterns in equations \ref{eq_ex2ssf} to \ref{eq_ex2allpauli}. The obtained circuits are shown in Figures \ref{fig_pauliopt}-c to \ref{fig_pauliopt}-f, respectively. Note that the optimization given by equations \ref{eq_pauliy} and \ref{eq_paulix} in the MBQC model becomes $J$-gate parallelization in the circuit model. However, it is not clear in which cases the $J$-gate parallelization implies time optimization. An in-depth analysis of the time-memory tradeoff for extended circuits with pauli angles remains as an interesting open question.
\section{Complexity analysis}

In this section we analyse the space, \emph{i.e.} the number of wires, and depth complexity of our proposed optimization scheme.  We show that our procedure can give a more optimal circuit in terms of both space and depth compared with the general method obtained in \cite{BroadbentK09}. An overview of this analysis is shown in Figure \ref{fig_complexity}.

Let $C$ be our initial quantum circuit with $n$ wires, $m = \mathtt{poly}(n)$ J gates and depth $d_C = \mathtt{poly}(n)$.
We can translate the computation implemented by $C$ to an MBQC pattern $P$ with $m + n$ qubits, out of which $m$ are measured qubits, and depth $d_P \leq d$ \cite{BroadbentK09}.
This pattern can be further optimised to depth $d_{\text{SS}} \leq d_P$ by performing signal shifting \cite{BroadbentK09} and obtaining a new pattern $P_{\text{SS}}$.
The extended quantum circuit $C_{\text{SS}}$, corresponding to the signal shifted pattern $P_{\text{SS}}$ can be created via the method given in Definition \ref{def_ext}.
This new circuit $C_\text{SS}$ performs exactly the same computation as the initial circuit $C$ and has the same number of $J$-gates, $m + n$ wires and depth $d_{\text{SS}} \cdot O(m)$.
One could apply the parallelisation method of \cite{MooreN01} to create a parallelised circuit $C_\text{Par}$ with depth $d_\text{SS} \cdot O(\log n)$ and size $O(m^2)$ \cite{BroadbentK09}.
Depending on the value of $d_\text{SS}$, the depth of $C_\text{Par}$ can be smaller than that of the original circuit $C$, but this could increase the space used, which can increase considerably, \emph{i.e} from $n$ to $O(m^2)$.

\begin{figure}[h]
        \begin{center}
                \resizebox{\hsize}{!}{\includegraphics{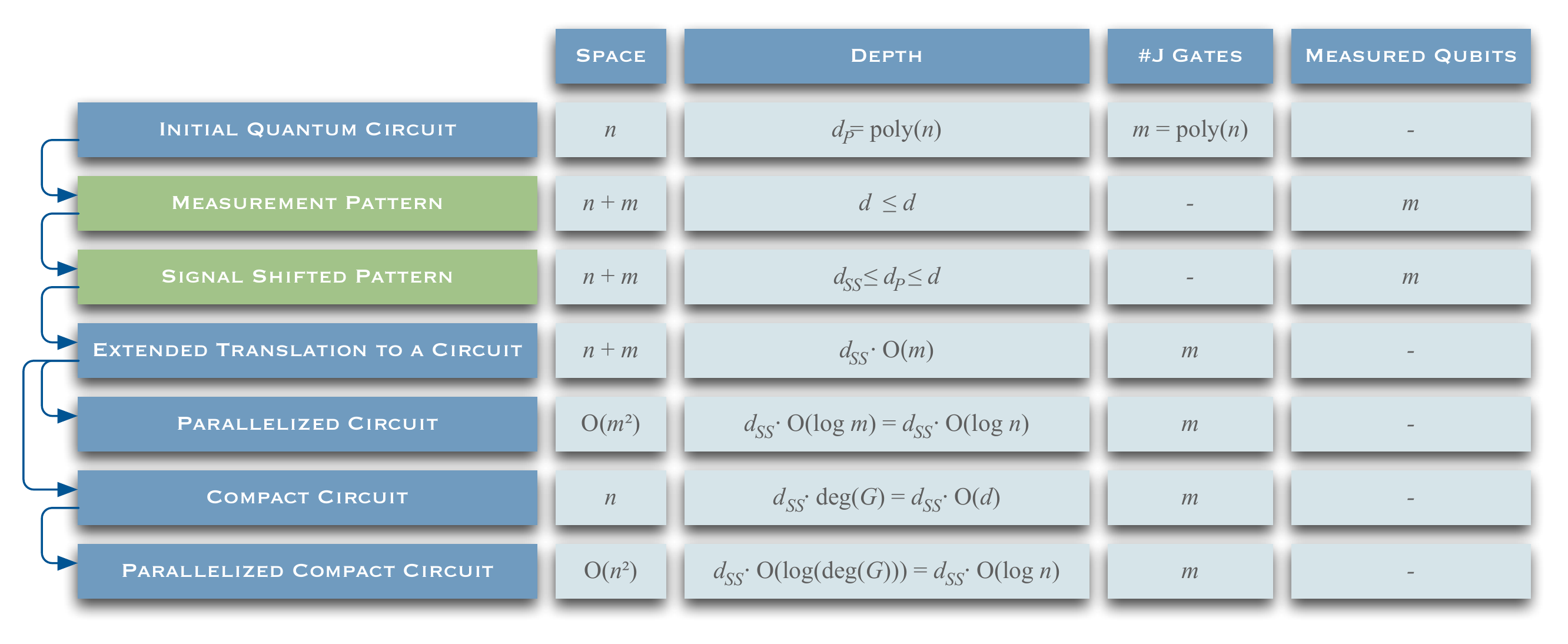}}
        \end{center}
        \caption{A summary of the optimisation procedure using extended translation and compactification. The parallelisation procedure described in \cite{BroadbentK09} (without Pauli simplification) ends with the ``Parallelised Circuit'' whereas our optimisation can end with a ``Compactified Circuit'' or a ``Parallelised Compactified Circuit'' depending the goal of the optimization.}
        \label{fig_complexity}
\end{figure}

To resolve this problem, instead of parallelising the circuit $C_\text{SS}$ we apply our compactification method. This will give us the compact circuit $C_{Com}$ with $n$ wires and $d_\text{SS} \cdot O(deg(G))$ where $deg(G)$ is the degree of the open graph $(G, I, O)$ of $P_{\text{SS}}$.

\begin{lemma}
        Let $C_\text{SS}$ be the quantum circuit obtained through extended translation from a signal shifted measurement pattern $P_\text{SS}$ on an open graph $(G, I, O)$ with depth $d_\text{SS}$. The application of the compactification Algorithm \ref{alg_full} to $C_\text{SS}$ results in a circuit with depth $d_\text{SS} \cdot O(deg(G))$.
\end{lemma}
\begin{proof}
        Algorithm \ref{alg_full} is the application of Rewrite Procedures (1) - (5) as defined in Section \ref{sec_rps} to the extended circuit of $P_\text{SS}$.
        Note that the number of $J$-gate layers in $C_{SS}$ is equal to the depth of $P_\text{SS}$ \cite{BroadbentK09}.
        Since the rewrite procedures do not change the $J$-gate layers, the compactified circuit $C_\text{Com}$ will also have $d_\text{SS}$ $J$-gate layers.

        Now we will count the number of gates in-between the $J$-gates.
        We do not need to be concerned with gates in the  $\mathcal{C}$ layers, since Algorithm \ref{alg_full} removes all of them.
        It is possible that the rewrite procedures leave some $E_{ij} = CZ_{ij}$ gates in front of qubit $i$.
        As all of these gates correspond to an edge in the graph $G$, the maximum number of such gates has to be $O(deg(G))$.
        According to the rewrite procedures, none of the $CX_{k, i}$ gates created will be moved past the $J$-gates.
        As can be seen from the rewrite procedures and Algorithm \ref{alg_full}, these will be created because of the existence of some $E_{f^{-1}(i) k}$ gates in $C_\text{SS}$.
        Since these correspond to the edges connected to the vertex $f^{-1}(i)$ there can only exist $O(deg(G))$ many $CX$ gates between any two $J$-gates.
        Now we know that there can be a total of $O(deg(G))$ two qubit gates in-between any two $J$-gates and hence also between the $J$-gate layers.
        Therefore the total depth of the circuit $C_\text{Com}$ will be $d_\text{SS} \cdot O(deg(G))$
\end{proof}

Already we can see, that if $deg(G) < \log n$ the $C_\text{Com}$ will have smaller depth than $C_\text{SS}$ while the corresponding space will be considerably smaller.
We can further decreased the depth of $C_\text{Com}$ by applying the parallelisation method from \cite{MooreN01}.
The depth of the new circuit $C_\text{ComPar}$ will be $d_\text{SS} \cdot O(\log (deg G))$ and size $O(n^2)$.
Note that because of the way $C$ is translated into $P$, the maximum number of edges connected to a vertex in $G$ will be $O(d)$, where $d = \text{poly}(n)$.
Therefore the depth of $C_\text{ComPar}$ can be written as $d_\text{SS} \cdot O(\log n)$.
Hence the compactification procedure together with the parallelisation method from \cite{MooreN01} will in the worst case give us the same depth as the method from \cite{BroadbentK09}, without Pauli Simplifications, but uses considerably less qubits ($n^2$ vs $m^2$, where $m = poly(n)$).

\section{Discussion and summary of results}
      
Initially MBQC was proposed as an alternative architecture for the implementation of quantum computing. However from early on the distinct parallel power of the model attracted researchers to explore further this unique feature of the model. In a series of results the key concepts of flow, signal shifted flow, gflow, maximally delayed gflow, focused gflow and information preserving flow were introduced \cite{DanosK06,DanosKP07,BrowneKMP07,MhallaP07, MhallaMPST11}. They address the general question of determinism in MBQC while shedding light on the parallelism as well. Although it is now known that the MBQC parallel power is equivalent to the quantum circuit with unbounded fanout \cite{BKP11}, further investigation is required to fully take advantage of this extra power. In this paper we continue this line of research by first presenting a surprising link between signal shifted flow  and maximally delayed gflow. The surprise comes from the fact that the former is obtained via a simple rewrite rules of pushing the $Z$ dependencies of a pattern to the end of the computation, while the latter is constructed directly from the stabilisers of the underlying graph. This leads to a new efficient procedure for finding the optimal gflow of graphs with flow as we discussed in the last section.

Moreover the link between signal shifted flow and optimal gflow opens a new direction to unify further the ``flow'' structure and fully characterise the constructions behind the parallel power of MBQC. To begin with, we succeeded in extending the applicability of the compactification method from \cite{daSilvaG12} by proving that it can also be used for the signal shifted flow. This will allow us to translate back the obtained parallel structure of the MBQC into the circuit model without increasing the size, leading to an automated optimisation procedure for quantum circuits. The automated scheme explore the global structure of circuit to parallelise several $J$-gates and to group together several $CX$ gates, which allow further optimisations to take place using specific method for Clifford gate parallelisation. Interestingly the scheme fails to compact the parallel pattern obtained via Pauli simplification rules \cite{BroadbentK09}. In other words one needs to keep the extra space to keep the parallel depth obtained due to the Pauli measurements. This further indicates the crucial role that Clifford computation (corresponding to Pauli measurements) play in obtaining the superior parallel power of MBQC over quantum circuit. 

The above trade-off naturally suggests the consideration of a hybrid model for quantum computing, where part of the computation is processed using the quantum circuit model and the other part using MBQC. Such a computation can be obtained by a \textit{partial} compactification of the extended circuits. First a quantum circuit acting on $n$ qubits is translated to the MBQC model, where it requires $m>n$ qubits. Then, we use an automated compactification procedure to obtain many different circuits implementing the same computation but with different depth and number of qubits (ranging from $n$ to $m$ many qubits). In those circuits, part of the computation is implemented using MBQC to obtain the parallel depth due to Pauli measurement while the rest are performed in the circuit model as one would not achieve any parallel advantage.  We believe that this can be of great value for experimental implementation of quantum computation, since one can design algorithms more adapted to the available experimental resources.

The application of the different techniques introduced in this paper to known quantum algorithms, as well as a full comparison with other known optimization methods beyond MBQC, constitute an interesting subject to be explored in the future. 

\subsection*{Acknowledgments}  We are grateful to Ernesto F. Galv\~ao for many helpful discussions. We acknowledge financial support by the Instituto Nacional de Ci\^encia e Tecnologia de Informac\~ao Qu\^antica (INCT- IQ/CNPq - Brazil) that sponsored R. Dias da Silva visit to Scotland where this work was completed. E. Kashefi also acknowledges support from UK Engineering and Physical Sciences Research Council (EP/E059600/1).

\end{document}